\documentclass[]{elsarticle}

\usepackage{lineno,hyperref}
\modulolinenumbers[5]

\journal{Journal of Symbolic Computation}

\usepackage{color,amssymb,amsmath}
\usepackage{xy}
\xyoption{all}

\newtheorem{lemma}{Lemma}
\newtheorem{theorem}{Theorem}
\newtheorem{corollary}{Corollary}
\newtheorem{definition}{Definition}
\newtheorem{proposition}{Proposition}
\newtheorem{remark}{Remark}
\newtheorem{example}{Example}
\newenvironment{proof}{\noindent{\bf Proof.}}{\qed}

\newcommand{\ie}{{i.e., }}

\newcommand{\st}{\mbox{ such that }}

\newcommand{\pisiSE}{$\Pi\Sigma^*$}
\newcommand{\sigmaSE}{$\Sigma^*$}
\newcommand{\piE}{$\Pi$}

\newcommand{\bb}[1]{{\mathbb{#1}}}			

\newcommand{\paren}[1]{{\left({#1}\right)}}

\newcommand{\Ker}{{\rm Ker}\,}

\newcommand{\lc}{\operatorname{lc}}

\newcommand{\Const}{\mbox{Const}}
\newcommand{\spr}{\mbox{Spr}}
\newcommand{\dis}{\mbox{Dis}}

\newcommand{\ordSymbol}{\omega}

\let\set\mathbb

\newcommand{\semiper}[2]{{#1}^{{#2}^\ast}}

\newcommand{\dfact}[3]{{#1}^{{#3},{#2}}}

\newcommand{\coprime}{\,\bot\,}
\newcommand{\divides}{\,|\,}

\newcommand{\fallingFac}[2]{{#1}^{\underline{#2}}}

\newcommand{\MainRed}{\sf}

\definecolor{blaugrau}{rgb}{0.796887, 0.789075, 0.871107}
\newcounter{mmacnt}
\def\restartmma{\setcounter{mmacnt}{0}}
\restartmma \catcode`|=\active
\def|#1|{\mathrm{#1}}
\catcode`|=12
\newenvironment{mma}{
	\par
	\catcode`|=\active
	\parskip=6pt\parindent=0pt 
	\small
	\def\In##1\\{%
		\def\linebreak{\hfill\break\null\qquad}%
		\refstepcounter{mmacnt}
		\hangindent=2.5em\hangafter=0
		\leavevmode
		\llap{\tiny\sffamily In[\arabic{mmacnt}]:=\kern.5em}%
		\mathversion{bold}\scriptsize$\tt\bf\displaystyle##1$\normalsize
		\mathversion{normal}\par
	}%
	\def\Print##1\\{%
		\def\linebreak{\hfill\break}%
		\hangindent=2.5em\hangafter=0
		\leavevmode\scriptsize ##1\par}%
	\def\Out##1\\{%
		\vspace*{-0.5cm}\def\linebreak{$\hfill\break\null\hfill$}%
		\kern\abovedisplayskip\par
		\hangindent=2.5em\hangafter=0
		\leavevmode
		\llap{\tiny\sffamily Out[\arabic{mmacnt}]=\kern.5em}
		\scriptsize$\displaystyle\tt##1$\normalsize\hfill\null\par
		\kern\belowdisplayskip\vspace*{-0.3cm}
	}%
	\def\Warning##1##2\\{%
		\def\linebreak{\hfill\break}%
		\hangindent=2.5em\hangafter=0
		\leavevmode
		{\scriptsize##1 : ##2}\par}%
}{%
	\par\smallskip
}

\newcommand{\LoadP}[1]{\fcolorbox{black}{blaugrau}{
\begin{minipage}[t]{10.7cm}
			\scriptsize\normalfont #1
\end{minipage}}}

\newcommand{\myIn}[1]{{\sffamily In[#1]}}

\def\MLabel#1{{\refstepcounter{mmacnt}\label{#1}}\addtocounter{mmacnt}{-1}}
\newcommand{\MText}[1]{\textbf{\ttfamily#1}}

\allowdisplaybreaks[4]

\begin{document}

\begin{frontmatter}

\title{On Rational and Hypergeometric Solutions of Linear Ordinary Difference Equations
in $\Pi\Sigma^*$-field extensions}

\author[1]{Sergei A.~Abramov}
\ead{sergeyabramov@mail.ru}
\author[2]{Manuel Bronstein$\dagger$}
\author[3,4]{Marko Petkov\v sek\corref{cor1}}
\ead{marko.petkovsek@fmf.uni-lj.si}
\author[5]{Carsten Schneider}
\ead{Carsten.Schneider@risc.jku.at}

\address[1]{Dorodnicyn Computing Center,  Federal Research Center ``Computer Science and Control'' of the Russian Academy of Sciences, Moscow, Russia}
\address[2]{Passed away on June 6, 2005; leading researcher in the initial phase of the project}
\address[3]{Faculty of Mathematics and Physics, University of Ljubljana, Ljubljana, Slovenia}
\address[4]{Institute of Mathematics, Physics and Mechanics, Ljubljana, Slovenia}
\address[5]{Johannes Kepler University Linz, Research Institute for Symbolic Computation (RISC), Linz, Austria}

\cortext[cor1]{Corresponding author}

\begin{abstract}
We present a complete algorithm that computes all hypergeometric solutions of homogeneous linear difference equations and rational solutions of parameterized linear difference equations in the setting of \pisiSE-fields. More generally, we provide a flexible framework for a big class of difference fields that are built by a tower of \pisiSE-field extensions over a difference field that enjoys certain algorithmic properties. As a consequence one can compute all solutions in terms of indefinite nested sums and products that arise within the components of a parameterized linear difference equation, and one can find all hypergeometric solutions of a homogeneous linear difference equation that are defined over the arising sums and products.
\end{abstract}


\end{frontmatter}

\section{Introduction}\label{Sec:Introduction}

A key problem in symbolic summation is to solve linear ordinary difference equations
\begin{equation}\label{Equ:InhomRec}
a_0(\nu)F(\nu)+a_1(\nu)F(\nu+1)+\dots+a_n(\nu)F(\nu+n)=b(\nu)
\end{equation}
of order\footnote{$\set N$ denotes the set of nonnegative integers.} $n\in\set N$ where the coefficients $a_0,\dots,a_n$ and the inhomogeneous part $b$ are given in a certain domain $k$ of special functions and the solutions $F$ 
belong to the same domain $k$ of functions or an appropriate extension of it.
 
The algorithmic foundation was laid by Abramov's algorithms \cite{Abramov:71,Abramov:89a,Abramov:89b} and their variants~\cite{vanHoeij:98} to find rational solutions $F$ in a rational function field $k= C(\nu)$ of homogeneous linear difference equations ($b=0$) with rational coefficients $a_0,\dots,a_n\in k$. Many variants and extensions, see e.g.~\cite{Abramov:95a,BP:99}, are available
to find also $q$-rational solutions in $C(q^{\nu})$, multibasic solutions in $C(q_1^{\nu},\dots,q_n^{\nu})$ or their mixed versions $C(\nu,q_1^{\nu},\dots,q_n^{\nu})$. Another milestone was Petkov{\v{s}}ek's algorithm Hyper~\cite{Petkovsek1992} and its improvement~\cite{vanHoeij:99} to find all hypergeometric solutions of~\eqref{Equ:InhomRec}, i.e., hypergeometric products of the form $\prod_{r=l}^{\nu} h(r)$ with $l\in\set N$ and a rational function $h(r)\in C(r)$. 
More generally, algorithms have been derived for $q$-hypergeometric solutions~\cite{APP:98}, their multibasic mixed versions~\cite{BP:99}, and interlaced hypergeometric solutions~\cite{HornKS12}. Exploiting these algorithms one can search not only for hypergeometric solutions and their variants, but also for d'Alembertian solutions~\cite{Abramov:94,Abramov:96}, i.e., solutions
in terms of indefinite nested sums defined over hypergeometric products. In addition, incorporating also the interlacing of sequences, one obtains algorithms to find all Liouvillian solutions~\cite{Singer:99,Petkov:2013}. 

Recently, several new approaches to finding explicit solutions of linear functional equations have been proposed. For second-order differential equations with rational-function coefficients, in~\cite{vanHoeij:17} heuristic (but very effective) algorithms are presented that search for solutions of the form 
\[
e^{\int\! r\,dx} r_0\cdot {_2}F_1(a_1,a_2;b_1;f) + r_1\cdot {_2}F_1'(a_1,a_2;b_1;f)
\]
where $r,r_0,r_1,f \in \overline{{\mathbb Q}(x)}$. For linear difference equations, algorithms are given in~\cite{Petkov:18} that search for solutions involving definite sums of certain products of binomial coefficients. In this article we will follow another branch of generalization whose foundation was led by Karr's seminal difference field approach~\cite{Karr81,Karr85}. With the underlying machinery one can solve first-order linear difference equations~\eqref{Equ:InhomRec} with $n=1$ where the coefficients $a_0$ and $a_1$ and the inhomogeneous part $b$ can be represented in terms of indefinite nested sums and products. More precisely, the components of the recurrence are represented within a given \pisiSE-field $k$ and the solutions are searched within this field $k$ or a given extension of it.

\smallskip

\noindent{\it Definition.} Let $k$ be a field of characteristic
$0$ and let $\sigma$ be a field automorphism of $k$. Then
$(k,\sigma)$ is called a difference field; the constant field of
$k$ is defined by $C=\{c\in k:\sigma(c)=c\}$. A difference
field $(k,\sigma)$ with constant field $C$ is called a
\pisiSE-field if $k=C(t_1,\dots,t_e)$ where for all $1\leq i\leq
e$ each $k_i=C(t_1,\dots,t_i)$ is a transcendental field
extension of $k_{i-1}=C(t_1,\dots,t_{i-1})$ (we set $k_0=C$)
and $\sigma$ has the property that $\sigma(t_i)=a\,t_i$ ($t_i$ is a \piE-monomial) or
$\sigma(t_i)=t_i+a$ ($t_i$ is a \sigmaSE-monomial) for some $a\in k_{i-1}$.

\smallskip

We remark that Karr's algorithm and refinements of it~\cite{Schneider:08,Schneider:15,DR1} enable one to carry out Zeilberger's creative telescoping paradigm~\cite{Zeilberger:91} in order to compute linear recurrence relations of definite sums whose summands can be represented in \pisiSE-fields.

 The first steps to generalize Karr's algorithm to $n$th-order linear difference equations~\eqref{Equ:InhomRec}  have been accomplished already in earlier work. Using Bronstein's generalization~\cite{Bronstein2000} of Abramov's denominator bounding~\cite{Abramov:89a,Abramov:89b}
and Schneider's contributions~\cite{Schneider:01,Schneider:04c,Schneider:05a,Schneider:05b} one obtains a general method to solve linear difference equations of arbitrary order in \pisiSE-fields.
Restricting to the case that the coefficients $a_0,\dots,a_n$ are from $C(\nu)$ (or, e.g., from the mixed multibasic case) and the inhomogeneous part is from a \pisiSE-field $k$, one obtains a complete algorithm to find all d'Alembertian solutions of~\eqref{Equ:InhomRec}. 
This machinery, implemented within the package \texttt{Sigma}~\cite{Schneider:07a}, enables one to produce (inhomogeneous) recurrences for big classes of definite multi-sums and to solve them in terms of d'Alembertian solutions efficiently. Typical examples can be found, e.g., in recent calculations within particle physics~\cite{QCD}. However, if also the coefficients $a_0,\dots,a_n$ are from a general \pisiSE-field, the available method is incomplete: one has to guess certain denominator factors or has to predict manually degree bounds of polynomial solutions. Furthermore, no algorithms for finding hypergeometric solutions over general \pisiSE-fields have been available so far. 

It is remarkable that the continuous analog of the difference field approach has been fully solved for three decades. Starting with Risch's algorithm~\cite{Risch:69} for indefinite integration, which can be considered as the differential analogue of Karr's summation algorithm, many important contributions, like~\cite{Singer:85}, have been accomplished; see also~\cite{Bron:97} and references therein. This finally led in~\cite{Singer:91} to a complete algorithm that can find all Liouvillian solutions (the differential version of Liouvillian sequence solutions introduced above) of linear differential equations whose coefficients are given in terms of Liouvillian extensions. 

Following Singer's pioneering work we will push forward in this article the discrete version initiated by Karr's work. More precisely, given a homogeneous linear difference equations where the coefficients are represented in a \pisiSE-field $k$, we will elaborate a general algorithm that finds all solutions in terms of hypergeometric solutions $\prod_{i=l}^{\nu}h(i)$ with $h\in k^*$. As a consequence, our framework not only enables one to find all hypergeometric solutions~\cite{Petkovsek1992}, $q$-hypergeometric solutions~\cite{APP:98} or multibasic solutions and their mixed version~\cite{BP:99}, but also allows the multiplicand of the hypergeometric product to be built in terms of indefinite nested sums and products arising in the coefficients $a_i$ in~\eqref{Equ:InhomRec}.

\begin{example}\label{Exp:RecExp}
	\rm
(I) Given the linear recurrence
	\begin{multline*}
	\big(
	1
	+H_\nu
	+\nu H_\nu
	\big)^2 \big(
	3
	+2 \nu
	+2 H_\nu
	+3 \nu H_\nu
	+\nu^2 H_\nu
	\big)^2 F(\nu)\\
	-(1+\nu) (3+2 \nu) H_\nu\big(
	3
	+2 \nu
	+2 H_\nu
	+3 \nu H_\nu
	+\nu^2 H_\nu
	\big)^2 F(\nu+1)\\ 
	+(1+\nu)^2 (2+\nu)^3 H_\nu\big(
	1
	+H_\nu
	+\nu H_\nu
	\big)  F(\nu+2)
	=0
	\end{multline*}
	in $F(\nu)$ of order $2$ with coefficients in terms of the harmonic numbers $H_{\nu}=\sum_{r=1}^{\nu}\frac1r$, our algorithmic machinery presented below produces the hypergeometric solutions 
	\begin{equation}\label{Equ:HnSol}
	H_{\nu}\,\prod_{l=1}^{\nu} H_l\quad\text{ and }\quad H_{\nu}^2\,\prod_{l=1}^{\nu} H_l.
	\end{equation}

	\noindent(II) Similarly, given the linear recurrence 
	\begin{multline*}
	-2 (1+\nu)^2 (2+\nu) \nu!^2\big(
	7
	+6 \nu
	+\nu^2
	+3 \nu!
	+5 \nu \nu!
	+2 \nu^2 \nu!
	\big) F(\nu)\\
	+(1+\nu) (2+\nu) \nu!\big(
	16
	+16 \nu
	+3 \nu^2
	+7 \nu!
	+12 \nu \nu!
	+4 \nu^2 \nu!
	\big) F(\nu+1)\\
	-\big(
	2
	+4 \nu
	+\nu^2
	+\nu!
	+2 \nu\nu! 
	\big) F(\nu+2)=0
	\end{multline*}
	in $F(\nu)$ of order $2$ with coefficients in terms of the factorials $\nu!$, our new algorithm computes the hypergeometric solutions 
	\begin{equation}\label{Equ:n!Sol}
	\prod_{l=1}^{\nu} l!\quad\text{ and }\quad (\nu!+\nu^2)2^{\nu}\prod_{l=1}^{\nu} l!.
	\end{equation}
\end{example}

Similar to the special case $k=C(\nu)$ in~\cite{Petkovsek1992} (and in Singer's differential version~\cite{Singer:91}) the underlying algorithm requires as subtask to find polynomial and rational solutions of parameterized\footnote{In the parameterized version of~\eqref{Equ:InhomRec} the inhomogeneous part is of the form $b(\nu)=c_1\,b_1(\nu)+\dots+c_m\,b_m(\nu)$ for given $b_1,\dots,b_m\in k$ and one searches not only for $F(\nu)$ but also for all constants $c_1,\dots,c_{m}\in C\subseteq k$ such that~\eqref{Equ:InhomRec} holds.} linear difference equations with coefficients in a \pisiSE-field.
Therefore the second main contribution of this work is the completion of the available toolbox~\cite{Bronstein2000,Schneider:01,Schneider:04c,Schneider:05a,Schneider:05b} to a complete algorithm that finds all solutions of parameterized linear difference equations in a given \pisiSE-field $k$.
More precisely, we will derive complete algorithms that solve the denominator and degree bounding problems in a given \pisiSE-field $k$. To our surprise we succeeded in this task only by computing hypergeometric solutions (or at least all hypergeometric candidates) in a subfield of $k$ (which is again a \pisiSE-field). Summarizing, the tasks of solving parameterized linear difference equations and finding hypergeometric solutions are intrinsically tied up to each other: solving one problem requires the solution of the other one.

All the algorithms presented in this article have been implemented within the summation package~\texttt{Sigma}~\cite{Schneider:07a} exploiting the already implemented toolboxes from~\cite{Karr81,Karr85,
Schneider:01,Bronstein2000,AbraBron2000,Schneider:04c,Schneider:05a,Schneider:05b,Schneider05c,Schneider:08,Petkov:10,Schneider:15,DR1,DR3,OS:18,Schneider:20}. A demonstration how the recurrences from Example~\ref{Exp:RecExp} can be solved with ~\texttt{Sigma} is given in Example~\ref{Exp:MMAHyper} below.

The outline of the article is as follows. In Section~\ref{Sec:DF} we will rephrase the problem to find hypergeometric solutions in the setting of difference rings and fields. In Section~\ref{Sec:GeneralNF} we will introduce a normal form of rational functions in difference fields.  Using this representation we will present in Section~\ref{Sec:PiSiNF} a general strategy to compute hypergeometric solutions  over a general difference field $k(t)$ where $t$ is a \pisiSE-monomial ($t$ is an indeterminate that represents one indefinite sum or product on top, see the above definition). As it turns out, this yields a general reduction strategy from $k(t)$ to $k$ under the assumption that one can find all solutions of parameterized linear difference equations in the field $k(t)$. In Section~\ref{Sec:RationalSolutions} we will provide a general method how such rational solutions can be calculated in $k(t)$ if one can solve the corresponding problem in $k$.  Finally, we will combine all these reduction strategies in Section~\ref{Sec:Framework} in order to obtain a general framework that enables one to find all hypergeometric solutions and rational solutions of linear difference equations in a difference field $k=K(t_1)\dots(t_e)$ built by a tower of \pisiSE-monomials $t_i$. Here the base difference field $K$ must satisfy certain algorithmic properties. In particular, we obtain a complete algorithm if $k$ is a \pisiSE-field, i.e., if $K=C$ is the field of constants. Section~\ref{Sec.Conclusion} concludes the article.

\textit{A historical remark:} 
This investigation was initiated in 2003 at INRIA Sophia Antipolis within the framewok of \emph{Project Caf\'e} by its scientific director, Manuel Bronstein, together with Sergei A.~Abramov and Marko Petkov\v sek, with the intent to provide a difference version of Michael Singer's work presented in~\cite{Singer:91}. In the next two years all of the relevant subproblems were solved except one: obtaining a degree bound for polynomial solutions in difference-field extensions. In May of 2005, Manuel visited Carsten Schneider at RISC (JKU Linz) who was working independently on a version of the same algorithm in the setting of $\Pi\Sigma$-fields. Together they established the desired bound in $\Pi$-extensions, and Carsten joined the team. Unfortunately, on June 6, 2005, Manuel passed away at the age of 41 (see~\cite{Abramov06} for his biography). 
For various reasons, not the least of which was that it was not easy, the problem of obtaining a degree bound in $\Sigma$-extensions has eluded the remaining team for many years. But recently Carsten has  solved it completely, and with this paper we (S.A.A., M.P., and C.S.) repay our debt to Manuel, and pay homage to him.

\section{Difference rings and operators}\label{Sec:DF}

As motivated in the introduction (see Example~\ref{Exp:RecExp}) we will develop algorithms to find hypergeometric solutions in terms of indefinite nested sums and products. Here the arising objects are represented in a ring\footnote{All rings and fields in this paper have characteristic zero and by $R^*$ we denote the set of units of $R$.} $R$ and the shift-operator is modeled by a ring automorphism $\sigma:R\to R$. The tuple $(R,\sigma)$ is also called a difference ring. If $R$ is a field (often also denoted by $k$), $\sigma$ turns to a field automorphism and the tuple $(R,\sigma)$ is also called a difference field.  

We will use the following notations that are relevant in Karr's work~\cite{Karr81} and that have been refined and explored further in~\cite{Bronstein2000}.

\begin{definition}
	Let $(R,\sigma)$ be a difference ring. 
	\begin{itemize}
		\item $a\in R$ is called \emph{constant} if $\sigma(a)=a$. The \emph{set of constants} is denoted by $\Const_{\sigma}(R)$.
		\item $a\in R$ is called \emph{semi-invariant} if there is a $u\in R^*$ with $\sigma(a)=u\,a$. The \emph{set of semi-invariants} of $R$ is denoted by $R^{\sigma}$.
		\item $a\in R$ is called \emph{semi-periodic} if there are $n>0$ and $u\in R^*$ with $\sigma^n(a)=u^n\,a$. The \emph{set of semi-periodic elements} of $R$ is denoted by $\semiper{R}{\sigma}$.
	\end{itemize}
\end{definition}
By definition we have $\Const_{\sigma}(R)\subseteq R^{\sigma}\subseteq \semiper{R}{\sigma}$. Note further that $\Const_{\sigma}(R)$ is a subring of $R$. In particular, if $R$ is a field, $\Const_{\sigma}(R)$ is a subfield of $R$ which is also called the \emph{constant field} of $R$. Since all rings and fields have characteristic $0$, the rational numbers $\set Q$ are contained in $\Const_{\sigma}(R)$.

\noindent{\bf Notation:} For any $a$ in a difference ring $(R,\sigma)$
and any integer $n \ge 0$, we write 
\[
a_{\sigma,n} \ =\ \sum_{i=0}^{n-1} \sigma^i a, \qquad
\dfact{a}{n}{\sigma} \ =\  \prod_{i=0}^{n-1} \sigma^i a, \qquad
\dfact{a}{n}{\sigma,\sigma} \ =\  \prod_{i=0}^{n-1} \dfact{a}{i}{\sigma},
\]
where $a_{\sigma,0} = 0$ and $\dfact{a}{0}{\sigma} = \dfact{a}{0}{\sigma,\sigma} = \dfact{a}{1}{\sigma,\sigma} = 1$.

\begin{lemma}
\label{lm:ssi}
Let $(k,\sigma)$ be a difference field. Then
\begin{eqnarray}
\dfact{\left(\frac{\sigma\,a}{a}\right)}{n}{\sigma} &=&\ \frac{\sigma\, \dfact{a}{n}{\sigma}}{\dfact{a}{n}{\sigma}}
\ =\ \frac{\sigma^n a}{a}, \label{lm:sn} \\
\dfact{\left(\sigma^i a\right)}{n-i}{\sigma} &=&\ \frac{\dfact{a}{n}{\sigma}}{\dfact{a}{i}{\sigma}}, \label{lm:snsi} \\
\left(\dfact{a}{i}{\sigma}\right)^{\sigma, n-i} 
 &=&\ \frac{\dfact{a}{n}{\sigma,\sigma}}{\dfact{a}{i}{\sigma,\sigma}\, \dfact{a}{n-i}{\sigma,\sigma}} \label{lm:ssnssi} 
\end{eqnarray}
for all $a \in k^*$ and all integers $n \ge i \ge 0$.
\end{lemma}
\begin{proof}
By induction on $n$.
\end{proof}

For any $p$ and $q$ in a polynomial ring $R[x]$ over an integral
domain $R$, we write $p \coprime q$ to say that $p$ and $q$ are
coprime, and we denote the leading coefficient of $p$ by $\lc(p)$.

Let $(k,\sigma)$ be a difference field and $E$ be transcendental over $k$. The
\emph{difference operator ring over k} denoted by $k[E;\sigma]$ is a polynomial ring $k[E]$
with the noncommutative multiplication given by $E\,h=\sigma(h)\,E$ for any $h\in k$.
In Lemma~\ref{lm:riccati} below we will use the fact that $k[E;\sigma]$ is a right Euclidean domain. In addition, the action of the difference operator $L=a_0\,E^0+a_1\,E^1+\dots+a_m\,E^m\in k[E;\sigma]$ on an element $b\in k$ is denoted by
$$L(b)=a_0\,b+a_1\,\sigma(b)+\dots+a_m\,\sigma^m(b).$$
If $L\neq0$, we also define $\deg(L)=\max\{i: a_i\neq0\}$ as the leading degree and $\nu(L)=\min\{i: a_i\neq0\}$ as the trailing degree of $L$. 

Note that the difference operator ring is a special case of a left-Ore ring~\cite{Ore33}; for further properties and applications see, e.g., \cite{BronPet96,ChyzakSalvy98}.

\begin{definition}
Let $(k,\sigma)$ be a difference field and $C$ its subfield of constants. Let $V$ be a subspace of $k$ over $C$.
We say that we can \emph{compute all the solutions in $V$ of equations
with coefficients in $k$} if given any nonzero $L \in k[E;\sigma]$,
we can compute a finite basis of the $C$-vector space
$\{f \in V;\ L(f) = 0\}$.
We say that we can \emph{compute\footnote{This means that there is an algorithm that can perform the calculations.} all the hypergeometric candidates for
equations with coefficients in $k$} if given any nonzero $L \in k[E;\sigma]$,
we can compute a finite set $S \subset k$ such that for any $r \in k^\ast$,
if $E - r$ is a right factor of $L$ in $k[E;\sigma]$,
then $r = u \sigma(v)/v$
for some $u \in S$ and $v \in k^\ast$.
We say that we can \emph{compute all the hypergeometric solutions
of equations with coefficients in $k$} if given any nonzero $L \in k[E;\sigma]$,
we can compute a finite set $S \subset k^\ast$ and finite sets
$S_u \subset k^\ast$ for each $u \in S$
such that for any $r \in k^\ast$,
if $E - r$ is a right factor of $L$ in $k[E;\sigma]$,
then
$$
r = u\,\, \frac{\sum_{v \in S_u} c_v\, \sigma v}{\sum_{v\in S_u} c_v v}
$$
for some $u \in S$ and for some constants $c_v \in C$.
\end{definition}

The difference fields under consideration are built by a tower of certain difference field extensions. In general, a difference field $(K,\sigma')$ is called a \emph{difference field extension} of $(k,\sigma)$ if $K$ is a field extension of $k$ and $\sigma'(f)=\sigma(f)$ for all $f\in k$. Since $\sigma'$ and $\sigma$ agree on $k$, we do not distinguish between them anymore. Due to~\cite{Bronstein2000} (inspired by~\cite{Karr81}) we will use the following notions.

\begin{definition}
Let $(K,\sigma)$ be a difference field extension of $(k,\sigma)$. 
\begin{itemize}
\item We call $t\in K$ \emph{unimonomial over $k$} if $t$ is transcendental and $\sigma(t)=\alpha\,t+\beta$ with $\alpha\in k^*$ and $\beta\in k$.
\item A difference field extension $(K,\sigma)$ of $(k,\sigma)$ is \emph{given by a tower of unimonomials
over $k$} if $(K,\sigma) = \left(k(t_1)\dots(t_e),\sigma\right)$ 
where $t_i$ is a unimonomial over $k(t_1)\dots(t_{i-1})$ for all $i$ with $1\leq i\leq e$.
\end{itemize}
\end{definition}

\begin{example}\label{Exp:HyperDFHn1}
\rm
Consider the difference field $(k,\sigma)$ where $k=\mathbb{Q}(x)(h)$ is a rational function field and the field automorphism $\sigma:k\to k$ is defined by $\sigma(c)=c$ for all $c\in\mathbb{Q}$, $\sigma(x)=x+1$ and $\sigma(h)=h+\frac1{x+1}$. By construction, $k$ is built by the unimonomial $x$ over $\set Q$ and the unimonomial $h$ over $\set Q(x)$. Note that by construction of $(k,\sigma)$ the harmonic numbers $H_m$ with $H_{m+1}=H_m+\frac1{m+1}$ can be rephrased by the variable $h$. In particular, the linear recurrence in Example~\ref{Exp:RecExp}(I) can be rewritten in terms of the linear difference operator
	\begin{equation}\label{Equ:Operator2}
	L=a_0\,E^0+a_1\,E^1+a_2\,E^2
	\end{equation}
	with
	\begin{equation}\label{Equ:PCoeffHn}
	\begin{split}
	a_0&=(1
	+h
	+h x
	)^2 (
	3
	+2 h
	+2 x
	+3 h x
	+h x^2
	)^2,\\
	a_1&=-h (1+x) (3+2 x) (
	3
	+2 h
	+2 x
	+3 h x
	+h x^2
	)^2,\\
	a_2&=h (1+x)^2 (2+x)^3 (1
	+h
	+h x
	).
	\end{split}
	\end{equation}
A minimal set of hypergeometric candidates for $L$ is $S=\{\frac{1+h+h x}{x+1}\}$. Namely, with
\begin{equation}\label{Equ:rj}
r_j=\frac{1+h+h x}{x+1}\frac{\sigma(h^j)}{h^j},\quad j=1,2,
\end{equation}
we get all right-hand factors $E-r_j$ of $L$. Details how $S=\{\frac{1+h+h x}{x+1}\}$ (more precisely, a larger set that contains $S$) can be computed are given in Example~\ref{Exp:HyperDFHn4}.
\end{example}

\begin{lemma}
\label{lm:riccati}
Let $(k,\sigma)$ be a difference field and
$L = \sum_{i=0}^n a_i E^i \in k[E;\sigma]$
be a linear ordinary difference
operator with coefficients in $k$ and $a_n \ne 0$.
For any $r \in k^\ast$,
$E - r$ is a right-factor of $L$ in $k[E;\sigma]$ if and only if
$\sum_{i=0}^n a_i \dfact{r}{i}{\sigma} = 0$.
\end{lemma}
\begin{proof}
Let $y$ be transcendental over $k$ and extend $\sigma$ to an
endomorphism
of $k(y)$ by defining $\sigma y = r y$.
Let $L = Q(E - r) + s$ be the right-division of $L$ by $E - r$ where
$Q \in k[E;\sigma]$ and $s \in k$.
Then, $L(y) = Q(\sigma y - r y) + s y = s y$,
so $E - r$ is a right-factor of $L$ if and only if $L(y) = 0$.
Since $\sigma y = r y$, \eqref{lm:sn} implies that $\sigma^i y = \dfact{r}{i}{\sigma} y$
for any $i \ge 0$. Therefore,
$$
L(y) =
\sum_{i=0}^n a_i\, \sigma^i y = y \sum_{i=0}^n a_i\, \dfact{r}{i}{\sigma}
$$
and it follows that $E-r$ is a right-factor of $L$ if and only if
$\sum_{i=0}^n a_i\, \dfact{r}{i}{\sigma} = 0$.
\end{proof}

\begin{example}\label{Exp:HyperDFHn2}
\rm
Based on the results of Example~\ref{Exp:HyperDFHn1} the following equations with $j=1,2$ hold by Lemma~\ref{lm:riccati}:
\begin{equation}\label{Equ:Order2Ricc}
a_0\dfact{r_j}{0}{\sigma}+a_1\,\dfact{r_j}{1}{\sigma}+a_2\dfact{r_j}{2}{\sigma}=0.
\end{equation}
As it turns out, the representations~\eqref{Equ:rj} with the components $\tilde{r}:=\frac{1+h+h x}{x+1}$ and $h^j$ with $j=1,2$ yield
\begin{equation}\label{Equ:RicattiHn}
a_0\,h^j\dfact{\tilde{r}}{0}{\sigma}\,+a_1\,\sigma(h^j)\,\dfact{\tilde{r}}{1}{\sigma}+a_2\,\sigma^2(h^j)\dfact{\tilde{r}}{2}{\sigma}=0.
\end{equation}
With
\begin{equation*}
H_{\nu+i}=\sigma^i(h)|_{h\mapsto H_\nu,x\mapsto \nu}\quad\text{ and }\quad
\prod_{l=1}^{\nu+i}H_l=(\dfact{\tilde{r}}{i}{\sigma}|_{x\mapsto \nu})\prod_{l=1}^{\nu}H_l 
\end{equation*}
for $\nu,i\in\set N$ it follows from~\eqref{Equ:RicattiHn} that~\eqref{Equ:HnSol} are solutions of the recurrence given in Example~\ref{Exp:RecExp}(I).
\end{example}

\begin{theorem}
\label{th:riccati}
Let $(k,\sigma)$ be a difference field. If we can compute all
the solutions in $k$ and all the hypergeometric candidates for
equations with coefficients in $k$, then we can compute all the
hypergeometric solutions of equations with coefficients in $k$.
\end{theorem}
\begin{proof}
Let $L = \sum_{i=0}^n a_i E^i \in k[E;\sigma]$
be a linear ordinary difference
operator with coefficients in $k$ and $a_n \ne 0$.
By hypothesis, we can compute a finite set $S \subset k$ such that
for any $r \in k^\ast$, if $E - r$ is a right factor of $L$ in $k[E;\sigma]$,
then $r = u \sigma(v)/v$ for some $u \in S$ and $v \in k^\ast$.
Let $E - r$ be such a right factor. By Lemma~\ref{lm:riccati},
$\sum_{i=0}^n a_i \dfact{r}{i}{\sigma} = 0$. Since $r = u\sigma(v)/v$
for some $u \in S$ and $v \in k^\ast$, it follows from \eqref{lm:sn} that
$\dfact{r}{i}{\sigma} = \dfact{u}{i}{\sigma} \sigma^i(v)/v$, whence
$$
0 = \sum_{i=0}^n a_i\, \dfact{r}{i}{\sigma} =
\sum_{i=0}^n a_i\, \dfact{u}{i}{\sigma}\, \frac{\sigma^i v}v\,.
$$
By hypothesis, we can compute for each $u \in S$ a finite basis
$S_u \subset k^\ast$ of the $C$-vector space
$\{y \in k;\ L_u(y) = 0\}$ where
\begin{equation}\label{Equ:Lu}
L_u = \sum_{i=0}^n a_i\, \dfact{u}{i}{\sigma} E^i
\end{equation}
and $C = \Const_\sigma(k)$. Since $L_u(v) = 0$, it follows that there are constants
$c_w \in C$ such that $v = \sum_{w \in S_u} c_w w$, and so
$$
r = u\, \frac{\sigma v}{v}
= u \,\,\frac{\sum_{w \in S_u} c_w\, \sigma w}{\sum_{w\in S_u} c_w w}.
$$
As $E - r$ was an arbitrary right-factor of $L$ in $k[E;\sigma]$, this means that we can
compute all the hypergeometric solutions of equations with coefficients in $k$.
\end{proof}

\begin{example}\label{Exp:HyperDFHn3}
\rm
We take the hypergeometric candidate $u=\frac{1+h+h x}{x+1}\in S$ from Example~\ref{Exp:HyperDFHn1} and construct
\begin{multline}\label{Equ:HnPolyRecurrence}
L_u=(1 + h + h x) (3 + 2 h + 2 x + 3 h x + h x^2)\,E^0\\ 
-h (3 + 2 x) (3 + 2 h + 2 x + 3 h x + h x^2)\,E^1\\
+h (2 + x)^2 (1 + h + h x)\,E^2
\end{multline}
as given in~\eqref{Equ:Lu}.
Computing the set of solutions
\begin{equation}\label{Equ:HnPolyRecurrenceSol}
\{v\in\mathbb{Q}(x)(h):\,L_{u}(v)=0\}=\{c_1\,h+c_2\,h^2: c_1,c_2\in\mathbb{Q}\},
\end{equation} 
we can determine the components~\eqref{Equ:rj}. Later we will present an algorithm in Section~\ref{Sec:RationalSolutions} for computing the solutions~\eqref{Equ:HnPolyRecurrenceSol} of~\eqref{Equ:HnPolyRecurrence}. 
In addition, we have to explain how the set $S$ of hypergeometric candidates for $L$ can be determined. These aspects will be explored further in the next two sections. 
\end{example}

\section{A normal form for rational functions}\label{Sec:GeneralNF}

In this section we will require another important notion~\cite{Bronstein2000} that is closely related to the dispersion introduced in~\cite{Abramov:71,Abramov:89a}; compare also the definition of Karr's \emph{specification of equivalence} in~\cite{Karr81}.

\begin{definition}
Let $k$ be a field, $t$ transcendental over $k$, and $\sigma$ an automorphism of the polynomial ring\footnote{Note that any automorphism $\sigma$ of $k[t]$ has the form $\sigma(t)=a\,t+b$ for some $a\in k^*$ and $b\in k$; compare~\cite{Petkov:10}.}
$k[t]$. For any $a,b\in k[t]$ we define their \emph{spread} as
$$\spr_\sigma (a,b)=\{m\in\set N: \gcd(a,\sigma^m(b))\neq1\}$$
and their \emph{dispersion} as
$$\dis_\sigma(a,b)=\begin{cases}
-1& \text{ if } \spr_{\sigma}(a,b) \text{ is empty},\\
\max \spr_{\sigma}(a,b)& \text{ if } \spr_{\sigma}(a,b) \text{ is finite and nonempty},\\
+\infty & \text{ if } \spr_{\sigma}(a,b) \text{ is infinite.}
\end{cases};$$
In addition, we set $\spr_\sigma(a):=\spr_\sigma(a,a)$ for $a\in k[t]$.
\end{definition}

\begin{theorem}
\label{PNF}
Let $k$ be a field, $t$ transcendental over $k$, and $\sigma$ an automorphism of the polynomial ring
$k[t]$. Then for any $r \in k(t)^\ast$, there are polynomials $a,b,c \in k[t]$ such that
in $k(t)$
\begin{equation}
\label{eq:decomp}
r\ =\ \frac{a}{b} \cdot \frac{\sigma c}{c},
\end{equation}
where
\begin{enumerate}
\item[\rm (i)] $\spr_\sigma (a,b) = \emptyset$,
\item[\rm (ii)] $a \coprime c$,
\item[\rm (iii)] $b \coprime \sigma c$.
\end{enumerate}
If, in addition, $\sigma$ maps $k$ onto $k$, and $\semiper{k[t]}{\sigma} = k[t]^{\sigma}$, then
\begin{enumerate}
\item[\rm (iv)] $\spr_\sigma (c)$ is finite.
\end{enumerate}
If 
emptiness of spread in $k[t]$ is decidable, the polynomials $a,b,c$ can be computed.
\end{theorem}

\begin{proof}  Let $r = f/g$ where $f, g \in k[t]\setminus\{0\}$ with $f\coprime g$.
We obtain $a, b, c$ by performing the following steps:
\vskip 0.5pc
\noindent
\verb|input:  |
$f, g \in k[t]$, $f, g \ne 0$; \\
\verb|output: |
$a, b, c \in k[t]$ satisfying~(\ref{eq:decomp}) and (i) -- (iii); \\
\verb|   |\\
\verb|   1.| $\,a_0 := f; \ b_0 := g; \ c_0 := 1; \ i := 0;$ \\
\verb|   2. while |$\spr_\sigma (a_i, b_i) \ne \emptyset$ \verb|do|\\
\verb|   3.   |$i := i+1;$ \\
\verb|   4.   |$h_i := \min \spr_\sigma (a_{i-1}, b_{i-1})$; \\
\verb|   5.   |$d_i := \gcd (a_{i-1}, \sigma^{h_i} b_{i-1});$ \\
\verb|   6.   |$a_i := a_{i-1}/d_i;$ \\
\verb|   7.   |$b_i := b_{i-1}/\sigma^{-h_i}d_i;$ \\
\verb|   8.   |$c_i := c_{i-1} \cdot(\sigma^{-h_i}d_i)^{\sigma, h_i};$ \\
\verb|   9.| $\,n := i; \ a := a_n; \ b := b_n; \ c := c_n;$ \\
\verb|  10. return| $(a, b, c)$.\\
\vskip 0.1pt
\noindent
Clearly, this is an algorithm if one can decide emptiness of spread in $k[t]$:
in line 4 we have $\spr_\sigma (a_{i-1}, b_{i-1}) \ne \emptyset$, hence to find its minimum
test $m = 0, 1, 2, \ldots$ until $\deg \gcd(a_{i-1}, \sigma^m b_{i-1}) > 0$.
By definition of spread,
$\deg d_i > 0$. Therefore $\deg a_i < \deg a_{i-1}$ in line 6, hence the
\verb+while+ loop at lines 2 -- 8 eventually terminates.
At lines 6 and 7, $d_i \divides a_{i-1}$ and 
$ \sigma^{-h_i} d_i \divides b_{i-1}$,
so $a_i, b_i \in k[t]$.

Note that $\gcd(a_0,b_0)=\gcd(f,g)=1$ and thus $h_1>0$. Further, $a_i \divides a_{i-1}$, $b_i \divides b_{i-1}$  and $\spr_\sigma (a_i, b_i)
\subseteq \spr_\sigma (a_{i-1}, b_{i-1}) \setminus \{h_i\}$ for all $i \in \{1,\ldots,n\}$,
so $0<h_1 < h_2 < \cdots < h_n$. Clearly
%
\[
a  = \frac{f}{d_1 \cdots d_n},  \quad
b  = \frac{g}{\sigma^{-h_1}d_1 \cdots \sigma^{-h_n}d_n},  \quad
c  = (\sigma^{-h_1}d_1)^{\sigma, h_1} \cdots (\sigma^{-h_n}d_n)^{\sigma, h_n}, 
\]
therefore, by \eqref{lm:sn},
\[
\frac{a}{b} \cdot \frac{\sigma c}{c} \ =\ \frac{f}{d_1 \cdots d_n} \cdot
\frac{\sigma^{-h_1}d_1 \cdots \sigma^{-h_n}d_n}{g} \cdot 
\frac{d_1 \cdots d_n}{\sigma^{-h_1}d_1 \cdots \sigma^{-h_n}d_n}\ =\ \frac{f}{g}\ =\ r,
\]
proving~(\ref{eq:decomp}).

(i): Upon exiting the \verb+while+ loop, its condition $\spr_\sigma (a_i, b_i) \ne \emptyset$ is false.
Thus $\spr_\sigma (a, b) = \spr_\sigma (a_n, b_n) = \spr_\sigma (a_i, b_i) = \emptyset$.

(ii): Assume that $u \in k[t] \setminus k$ is an irreducible common factor of $a$ and $c$. Then $u \divides a_n$
and $u \divides \sigma^{-j} d_i$ for some $i \in \{1,\ldots, n\}$ with $h_i \ge 1$ and $j \in \{1,\ldots, h_i\}$.
As $b_{i-1} = b_i \sigma^{-h_i}d_i$, it follows that $u \divides \sigma^{h_i - j} b_{i-1}$,
and so $h_i - j \in \spr_\sigma (a_n, b_{i-1})$. On the other hand, $n > i-1$ implies that $a_n \divides a_{i-1}$,
hence $\spr_\sigma (a_n, b_{i-1}) \subseteq \spr_\sigma (a_{i-1}, b_{i-1}) \subseteq \{h_i, h_i+1, \ldots \}$.
As $j \ge 1$, we get $h_i - j \notin \spr_\sigma (a_n, b_{i-1})$. This contradiction shows that $a \coprime c$.

(iii): Assume that $u \in k[t] \setminus k$ is an irreducible common factor of $b$ and $\sigma c$. Then $u \divides b_n$
and $u \divides \sigma^{1-j} d_i$ for some $i \in \{1,\ldots, n\}$ with $h_i \ge 1$ and $j \in \{1,\ldots, h_i\}$.
As $a_{i-1} = a_i d_i$, it follows that $u \divides \sigma^{1 - j} a_{i-1}$. Hence $\sigma^{j-1}u \divides a_{i-1}$
and $\sigma^{j-1}u \divides \sigma^{j-1} b_{n}$, so $j-1 \in \spr_\sigma (a_{i-1}, b_n)$.
On the other hand, $n > i-1$ implies that $b_n \divides b_{i-1}$,
hence $\spr_\sigma (a_{i-1}, b_n) \subseteq \spr_\sigma (a_{i-1}, b_{i-1}) \subseteq \{h_i, h_i+1, \ldots \}$.
As $j \le h_i$, it follows that $j-1 \notin \spr_\sigma (a_{i-1}, b_n)$. This contradiction shows that $b \coprime \sigma c$.

(iv): Assume that $\spr_\sigma (c)$ is infinite.
Then $c$ has a factor $p \in \semiper{k[t]}{\sigma} \setminus k$
by Theorem 6(iii) of~\cite{Bronstein2000}.
By Lemma 3(iv) of~\cite{Bronstein2000}, we can assume that $p$ is irreducible.
It follows that $p \divides (\sigma^{-h_i}d_i)^{\sigma, h_i}$ for some $i \in \{1,\ldots,n\}$
with $h_i \ge 1$, and further that $p \divides \sigma^{-j}d_i$ for some $j \in \{1,\ldots,h_i\}$.
Therefore $\sigma^j p \divides a_{i-1}$ and $\sigma^{j-h_i} p \divides b_{i-1}$.
As $p \in \semiper{k[t]}{\sigma} = k[t]^{\sigma}$ by assumption, this implies
that $p \divides a_{i-1}$ and $p \divides b_{i-1}$.
Hence $h_i = \min \spr_\sigma (a_{i-1},b_{i-1}) = 0$,
a contradiction.
\end{proof}

\begin{example}
\rm
Let $f(t) = t^3$ and $g(t) = (t-1)(t-2)(t-3)$. We apply the algorithm given in the proof of Theorem \ref{PNF} to
$f(t)/g(t)$ for two different automorphisms $\sigma$ of $k[t]$.

\smallskip
a) If $\sigma$ is the unique $k$-automorphism of $k[t]$ satisfying $\sigma t = t+1$, then
$\spr_\sigma (f,g) = \{1,2,3\}$, $a(t) = b(t) = 1$, $c(t) = (t-1)^3(t-2)^2(t-3)$,
and $\spr_\sigma (c) = \{0,1,2\}$ is finite. In this case $\semiper{k[t]}{\sigma} = k[t]^{\sigma} = k$.

\smallskip
b) If $\sigma$ is the unique $k$-automorphism of $k[t]$ satisfying $\sigma t = 1-t$, then
$\spr_\sigma (f,g) = \{1,3,5,\ldots\}$, $a(t) = t^2$, $b(t) = -(t-2)(t-3)$, $c(t) = 1-t$,
and $\spr_\sigma (c) = \{0,2,4,\ldots\}$ is infinite. Note that in this case $\sigma^2 p(t) = p(t)$ for all
$p \in k[t]$, so $\semiper{k[t]}{\sigma} = k[t]$, but $t \notin k[t]^{\sigma}$ because there is no $u \in k[t]^* = k^*$
such that $1-t = ut$. Hence $\semiper{k[t]}{\sigma} \ne k[t]^{\sigma}$.
\end{example}

\section{Hypergeometric solutions in \pisiSE-extensions}\label{Sec:PiSiNF}

In Section~\ref{Sec:GeneralNF} we have introduced a normal form for general automorphisms of $k[t]$. In the following we will specialize and refine this result further to unimonomial extensions that enjoy further properties.

\begin{lemma}
\label{lm:abdiv}
Let $(k,\sigma)$ be a difference field and $t$ be unimonomial over $k$.
Let $L = \sum_{i=0}^n a_i E^i \in k[t][E;\sigma]$
be a linear ordinary difference operator with coefficients in $k[t]$
where $a_0 a_n \ne 0$, and let $r \in k(t)^\ast$ be such that $E-r$ is a
right factor of $L$ in $k(t)[E;\sigma]$.
Then, in the decomposition~(\ref{eq:decomp}) for $r$,
we have $a \mid a_0$ and $b \mid \sigma^{1-n} a_n$.
\end{lemma}
\begin{proof}
Since $E-r$ is a right factor of $L$,
Lemma~\ref{lm:riccati} implies that
$\sum_{i=0}^n a_i\, \dfact{r}{i}{\sigma} = 0$.
Replacing $r$ in this equation by its decomposition~(\ref{eq:decomp})
and using \eqref{lm:sn}, we obtain
$$
\sum_{i=0}^n a_i\,
\frac{\dfact{a}{i}{\sigma}}{\dfact{b}{i}{\sigma}}
\frac{\sigma^i c}c\ =\ 0\,.
$$
Multiplying through by $\dfact{b}{n}{\sigma} c$ and using \eqref{lm:snsi},
we establish that
\begin{equation}
\label{eq:abc}
\sum_{i=0}^n a_i\, \dfact{a}{i}{\sigma}
\dfact{(\sigma^i b)}{n-i}{\sigma} \sigma^i c\ =\ 0\,.
\end{equation}
Since $a$ divides every term in the sum~(\ref{eq:abc}) for $i > 0$,
it follows that $a \mid a_0 \dfact{b}{n}{\sigma} c$,
and Theorem~\ref{PNF}\,(i),(ii) implies that $a \mid a_0$.
Similarly, $\sigma^{n-1} b$ divides  every term in
the sum~(\ref{eq:abc}) for $i < n$, so
$\sigma^{n-1} b \mid a_n \dfact{a}{n}{\sigma} \sigma^n c$
and $b \mid \sigma^{1-n}a_n\, \dfact{\left(\sigma^{1-n}a\right)}{n}{\sigma} \sigma c$.
It follows from Theorem~\ref{PNF}\,(i),(iii)
that $b \mid \sigma^{1-n} a_n$.
\end{proof}

In the following we will restrict to the class of \pisiSE-monomials introduced by~\cite{Karr81}.

\begin{definition}
Let $(k,\sigma)$ be a difference field and let $t$ be a unimonomial over $k$ with $\Const_\sigma(k(t))
= \Const_\sigma(k)$. We say that:
\begin{itemize}
\item $t$ is a \emph{$\Pi$-monomial over $k$} if $\sigma t / t \in k$,
\item $t$ is a \emph{$\Sigma^*$-monomial over $k$} if $\sigma t - t \in k$,
\item $t$ is a \emph{\pisiSE-monomial over $k$} if either $\sigma t / t \in k$ or $\sigma t - t \in k$,
\item $(k(t),\sigma)$ is a \emph{\pisiSE-extension (resp.\ a \piE-extension, a \sigmaSE-extension) of $(k,\sigma)$}
if $t$ is a \pisiSE-monomial (resp.\ a \piE-monomial, a \sigmaSE-monomial) over $k$,
\item a difference field extension $(K,\sigma)$ of $(k,\sigma)$ is \emph{given by a tower of \pisiSE-monomials
over $k$} if $(K,\sigma) = \left(k(t_1)\dots(t_e),\sigma\right)$ 
where $t_i$ is a \pisiSE-monomial over $k(t_1)\dots(t_{i-1})$ for all $i$ with $1\leq i\leq e$,
\item $(k,\sigma)$ is a \emph{\pisiSE-field over $C$} if $C=\Const_\sigma(k)$ and $(k,\sigma)$
is given by a tower of \pisiSE-monomials over $C$.
\end{itemize}
\end{definition}

As worked out in~\cite{Karr81,Karr85} (see also~\cite{Bronstein2000,Schneider:01,DR1}) one obtains alternative characterizations of \piE-monomials and \sigmaSE-monomials that will be used later. In this regard, the following extra notion is needed.

\begin{definition}
Let $(k,\sigma)$ be a difference field. 
$\alpha\in k$ is called a \emph{$\sigma$-radical} over $k$ if there are $g\in k^*$ and $n\in\set N\setminus\{0\}$ with $\sigma(g)=\alpha^n\,g$.
\end{definition}

\begin{theorem}[\cite{Karr81}]
\label{thm:pisiProp}
Let $(K,\sigma)$ be a difference field extension of $(k,\sigma)$, and let $t\in K^*$. 
Then  $t$ is:
\begin{itemize}
\item[\rm (i)] a \piE-monomial over $k$ iff $\sigma(t)/t \in k$ and
$\sigma(t)/t$ is not a $\sigma$-radical over $k$,
\item[\rm (ii)]  a \sigmaSE-monomial over $k$ iff $\sigma t - t \in k$ and there is no $w \in k$ such that
$\sigma t - t = \sigma w - w$.
\end{itemize}
\end{theorem}

\begin{example}
\rm
Consider the difference field $(\mathbb{Q}(x)(h),\sigma)$ from Example~\ref{Exp:HyperDFHn1}: Obviously, $x$ is unimonomial over $\mathbb{Q}$. Since there is no $w\in\mathbb{Q}$ with $\sigma(w)-w=1$, it is a \sigmaSE-monomial by part (ii) of Theorem~\ref{thm:pisiProp}. Similarly, $h$ is unimonomial over $\mathbb{Q}(x)$ and there is no $w\in\mathbb{Q}(x)$ such that $\sigma(w)-w=\frac1{x+1}$. Thus $h$ is a \sigmaSE-monomial over $\mathbb{Q}(x)$ by part (ii) of Theorem~\ref{thm:pisiProp}. Summarizing, 
$(\mathbb{Q}(x)(h),\sigma)$ is a \pisiSE-field over $\mathbb{Q}$. In particular, $\Const_\sigma(\mathbb{Q}(x)(h))=\mathbb{Q}$.
\end{example}

\begin{example}\label{Exp:FactorialDF}
\rm
Consider the difference field $(k,\sigma)$ that is built from the rational function $k=\mathbb{Q}(x)(p)$ and the field automorphism $\sigma:k\to k$ defined by $\sigma(c)=c$ for all $c\in\mathbb{Q}$, $\sigma(x)=x+1$ and $\sigma(p)=(x+1)p$. As observed in the previous example, $x$ is a \sigmaSE-monomial over $\mathbb{Q}$. Furthermore, there are no $m\in\set N\setminus\{0\}$ and $w\in\set Q(x)\setminus\{0\}$ such that $\sigma(w)=(x+1)^m\,w$. Thus by part (i) of Theorem~\ref{thm:pisiProp} it follows that $p$ is a \piE-monomial over $\mathbb{Q}(x)$. Summarizing, $(\mathbb{Q}(x)(p),\sigma)$ is a \pisiSE-field over $\mathbb{Q}$, in particular, $\Const_\sigma(\mathbb{Q}(x)(p))=\mathbb{Q}$. 
\end{example}

\begin{corollary}\label{Cor:NormalFormPiSi}
Let $(k,\sigma)$ be a difference field, let $t$ be a \pisiSE-monomial over $k$, and let $r\in k(t)^\ast$. 
Then there are $a,b,c\in k[t]\setminus\{0\}$ such that~\eqref{eq:decomp} holds where
\begin{enumerate}
\item[\rm (i)] $\spr_\sigma (a,b) = \emptyset$,
\item[\rm (ii)] $a \coprime c$,
\item[\rm (iii)] $b \coprime \sigma c$,
\item[\rm (iv)] $\spr_\sigma (c)$ is finite,
\item[\rm (v)] $c(0)\neq0$ if $\sigma(t)/t\in k$.
\end{enumerate}
If emptiness of spread in $k[t]$ is decidable, such polynomials $a,b,c$ can be computed.
\end{corollary}
\begin{proof}
Since $t$ is a unimonomial over $k$, it is transcendental over $k$ and $\sigma$ is an automorphism of $k[t]$.
Hence by Theorem~\ref{PNF}, there are polynomials $a,b,c\in k[t]\setminus\{0\}$ such that \eqref{eq:decomp} holds
and conditions (i)--(iii) are satisfied. 

If $t$ is a $\Sigma^*$-monomial over $k$ then it follows from
\cite[Theorem~4]{Karr81} or~\cite[Theorem~3]{Bronstein2000} that 
\[
\semiper{k[t]}{\sigma} = k[t]^\sigma = k.
\]
If $t$ is a $\Pi$-monomial over $k$ then it follows from
\cite[Theorem~4]{Karr81} or~\cite[Theorem~4]{Bronstein2000} that 
\[
\semiper{k[t]}{\sigma}\ =\ k[t]^\sigma\ =\ \{f\,t^m: f \in k, m \in \bb{Z}, m \ge 0\}.
\]
In either case, Theorem~\ref{PNF} implies that condition (iv) holds as well.
Finally, from~\cite[Theorem~5]{Bronstein2000} (and using the assumption that $t$ is a \piE-monomial) it follows that $f\in k[t]$ has
infinite spread iff it is divisible by $t$. Since $\spr_\sigma (c)$ is finite, 
we conclude that $c(0)\neq 0$.
\end{proof}

\begin{lemma}
\label{lm:lctc}
Let $(k,\sigma)$ be a difference field and $t$ be a \pisiSE-monomial over $k$ with $\sigma(t)=\alpha\,t+\beta$. Let
$p \in k[t]$ be nonzero and $n \ge 0$ be an integer. Then
\begin{enumerate}
\item[\rm (i)]
$\deg (\sigma^n p)\ =\ \deg p$,
\item[\rm (ii)]
$\lc(\sigma^n p)\ =\ \sigma^n(\lc(p)) \,\dfact{(\alpha^{\deg p})}{n}{\sigma}$,
\item[\rm (iii)]
$\deg (\dfact{p}{n}{\sigma})\ =\ n \deg p$,
\item[\rm (iv)]
$\lc(\dfact{p}{n}{\sigma})\ =\ \dfact{\lc(p)}{n}{\sigma}\dfact{(\alpha^{\deg p})}{n}{\sigma,\sigma}$.
\end{enumerate}
\end{lemma}

\begin{proof}
Write $p = \sum_{j=0}^d p_j t^j$ with $p_d \ne 0$. Since $\sigma t = \alpha t + \beta$,
$\sigma^n t = \dfact{\alpha}{n}{\sigma} t + a_n$ for some $a_n \in k$.
Therefore, $\sigma^n(p_j t^j) = \sigma^n(p_j) (\dfact{\alpha}{n}{\sigma}t + a_n)^j$, which implies that
$\sigma^n(p) = \sigma^n(p_d)\dfact{(\alpha^d)}{n}{\sigma} t^d + q_n$ where $q_n \in k[t]$
and $\deg(q_n) < d$. This proves (i) and (ii). Furthermore,
$$
\dfact{p}{n}{\sigma}\ =\ \prod_{i=0}^{n-1} \sigma^i p
\ =\ \prod_{i=0}^{n-1} (\sigma^i(p_d)\dfact{(\alpha^d)}{i}{\sigma} t^d + q_i)
\ =\ \dfact{p_d}{n}{\sigma}\dfact{(\alpha^d)}{n}{\sigma,\sigma} t^{nd} + q
$$
where $q \in k[t]$ and $\deg(q) < nd$. This proves (iii) and (iv).
\end{proof}

We are now ready to present our main result to compute hypergeometric solutions over a \pisiSE-monomial.
We note that a specialized version of this result has been also utilized in~\cite{Petkovsek1992,APP:98,BP:99}.

\begin{theorem}
\label{th:hyper}
Let $(k,\sigma)$ be a difference field and $t$ be a \pisiSE-monomial over $k$.
\begin{enumerate}
\item[\rm (i)]
If we can compute all the hypergeometric candidates for equations with coefficients in $k$,
then we can compute all the hypergeometric candidates for equations with coefficients in $k(t)$.
\item[\rm (ii)]
If, in addition, we can also compute a basis for the solutions in $k[t]$ (if $t$ is a \sigmaSE-monomial)
or in $k[t^{-1}]$ (if $t$ is a \piE-monomial) of equations with coefficients in $k(t)$,
then we can compute
all the hypergeometric solutions of equations with coefficients in $k(t)$.
\end{enumerate}
\end{theorem}

\begin{proof}
Let $t$ be a \pisiSE-monomial over $k$ with $\sigma\,t=\alpha\,t+\beta$
where $\alpha \in k^*$ and $\beta \in k$.
Let $L = \sum_{i=0}^n a_i E^i \in k[t][E;\sigma]$ with $a_n \ne 0$.
Let $r \in k(t)^\ast$ be such that $E - r$ is a right factor of $L$
in $k(t)[E;\sigma]$. By Lemma~\ref{lm:riccati}, we have
$\sum_{i=0}^n a_i \dfact{r}{i}{\sigma} = 0$.
From Corollary~\ref{Cor:NormalFormPiSi} it follows by extraction of leading coefficients
that there are $z \in k^*$ and monic polynomials $a,b,c \in k[t]$ such that 
\begin{equation}
\label{eq:zdecomp}
r\ =\ z\, \frac ab \,\frac{\sigma c}c
\end{equation}
and conditions (i)--(v) of Corollary~\ref{Cor:NormalFormPiSi} are satisfied. 
Replacing decomposition (\ref{eq:decomp}) by (\ref{eq:zdecomp}) in the proof of Lemma
\ref{lm:abdiv} changes equation~(\ref{eq:abc}) into
\begin{equation}
\label{Equ:LeadingCoeff}
\sum_{i=0}^n a_i\, \dfact{z}{i}{\sigma}
\dfact{a}{i}{\sigma} \dfact{(\sigma^i b)}{n-i}{\sigma} \sigma^i c = 0\,.
\end{equation}
Let $\mu = \max_{0\le i \le n} \tau(i)$ where $\tau(i) = \deg a_i + i \deg a + (n-i) \deg b$.
Using Lemmas~\ref{lm:ssi} and \ref{lm:lctc} and the fact that $a, b, c$ are monic, we obtain
\begin{eqnarray*}
0\ &=& \mbox{ coefficient of } t^{\mu+ \deg c} \mbox{ in }
{\sum_{i=0}^n a_i\, \dfact{z}{i}{\sigma}
\dfact{a}{i}{\sigma} \dfact{(\sigma^i b)}{n-i}{\sigma} \sigma^i c} \\
&=& \sum_{0 \le i \le n \atop \tau(i) = \mu} \lc(a_i)\,\dfact{z}{i}{\sigma}
\,\lc(\dfact{a}{i}{\sigma})\,\lc\left(\dfact{(\sigma^i b)}{n-i}{\sigma}\right)\,\lc(\sigma^i c) \\
&=& \sum_{0 \le i \le n \atop \tau(i) = \mu} \lc(a_i)\,\dfact{z}{i}{\sigma}
\,\dfact{\left(\alpha^{\deg a }\right)}{i}{\sigma,\sigma}
\,\dfact{\left(\lc(\sigma^i b)\right)}{n-i}{\sigma} \dfact{\left(\alpha^{\deg b }\right)}{n-i}{\sigma,\sigma}
\,\dfact{\left(\alpha^{\deg c }\right)}{i}{\sigma}\\
&=& \sum_{0 \le i \le n \atop \tau(i) = \mu} 
\,\lc(a_i)\,\dfact{\left(\alpha^{\deg a }\right)}{i}{\sigma,\sigma}
\,\left(\dfact{\left(\alpha^{\deg b }\right)}{i}{\sigma}\right)^{\sigma, n-i}
\dfact{\left(\alpha^{\deg b }\right)}{n-i}{\sigma,\sigma}
\,\dfact{\left(z\,\alpha^{\deg c }\right)}{i}{\sigma}\\
&=& \sum_{0 \le i \le n \atop \tau(i) = \mu}
\lc(a_i) \,\dfact{\left(\alpha^{\deg a -\deg b }\right)}{i}{\sigma,\sigma}
\,\dfact{\left(\alpha^{\deg b }\right)}{n}{\sigma,\sigma} \,\dfact{\tilde{z}}{i}{\sigma} 
\end{eqnarray*}
where $\tilde{z}=z\,\alpha^{\deg c }\in k^*$. Then Lemma~\ref{lm:riccati} implies that $E - \tilde{z}$ 
is a right factor of
\begin{equation}\label{Equ:Lab}
L_{a,b} =
\sum_{0 \le i \le n \atop \tau(i) = \mu}
\lc(a_i) \,\dfact{\left(\alpha^{\deg a -\deg b }\right)}{i}{\sigma,\sigma}
\,\dfact{\left(\alpha^{\deg b }\right)}{n}{\sigma,\sigma} \, E^i
\end{equation}
in $k[E;\sigma]$. Since we can compute all the hypergeometric candidates for equations with
coefficients in $k$, we can compute a finite set $S_{a,b} \subset k$ such
that $\tilde{z} = u \sigma(v)/v$ for some $u \in S_{a,b}$ and $v \in k^\ast$.
Then replacing $z$ by $\tilde{z}\alpha^{-\deg c }=\alpha^{-\deg c }u\sigma(v)/v$ in~(\ref{eq:zdecomp})
we establish that
\begin{equation}
\label{eq:uvdecomp}
r\ =\ u\,\alpha^{-\deg c }\ \frac ab\, \frac{\sigma (vc)}{vc}
\ =\begin{cases}
\ u\, \frac ab\, \frac{\sigma(vc)}{vc}&\text{if $t$ is a \sigmaSE-monomial},\\
\ u\, \frac ab\, \frac{\sigma(vct^{-\deg c })}{vct^{-\deg c }}&\text{if $t$ is a \piE-monomial}.
\end{cases}
\end{equation}
By Lemma~\ref{lm:abdiv},
the finite set $\{u\,a/b;\ u \in S_{a,b},\ a \mid a_0,\ b \mid \sigma^{1-n} a_n,\, a, b \mbox{ monic}\}$
is the set of all the hypergeometric candidates for $L$, which proves (i).

At this point, if we assumed that we can compute all the solutions
in $k(t)$ of equations with coefficients in $k(t)$, then we could
compute all the hypergeometric solutions of equations with coefficients
in $k(t)$ by Theorem~\ref{th:riccati}. We assume however only the weaker
condition that we can compute a basis for the solutions in $k[t]$ (if $t$ is a \sigmaSE-monomial)
resp.\ the solutions in $k[t^{-1}]$ (if $t$ is a \piE-monomial) of such equations. For a given candidate $u\,a/b$,
replacing $r$ by its decomposition~(\ref{eq:uvdecomp}) in $\sum_{i=0}^n a_i \dfact{r}{i}{\sigma} = 0$ we obtain
$L_{a,b,u}(vc) = 0$ if $t$ is a \sigmaSE-monomial, and $L_{a,b,u}(vct^{-\deg c }) = 0$ if $t$ is a \piE-monomial
where
\begin{equation}\label{Equ:Labu}
L_{a,b,u}\ =\
\sum_{i=0}^n a_i\, \dfact{u}{i}{\sigma}
\dfact{a}{i}{\sigma} \dfact{(\sigma^i b)}{n-i}{\sigma} E^i\,.
\end{equation}
Note that $vc \in k[t]$ and $vct^{-\deg c } \in k[t^{-1}]$.
By hypothesis, we can compute a basis for the solutions in $k[t]$, resp.\ in $k[t^{-1}]$
of $L_{a,b,u}$, so doing this for each of the finitely many candidates $u a/b$
we can compute all the hypergeometric solutions of $L$, which proves (ii).
\end{proof}

\begin{remark}\label{Remark:SwithToPolynomialCase}
\rm
In the proof of Theorem~\ref{th:hyper} we looked at the leading coefficient in equation~\eqref{Equ:LeadingCoeff}.
However, if $t$ is a \piE-monomial, one can also look at the trailing coefficient in this equation. 
In this way, the factor $\alpha^{\deg c }$ depending on the degree of $c$ can be avoided. It is this extra factor
that leads to the problem of finding solutions in $k[t^{-1}]$ instead of in $k[t]$. However, this is not
a restriction at all:
If $t$ is a \piE-monomial with $\sigma(t)=\alpha\,t$, then also $\hat{t}:=t^{-1}$ is a \piE-monomial with 
$\sigma(\hat{t})=\alpha^{-1}\hat{t}$. Thus finding solutions in $k[t^{-1}]$ --- as required in Theorem~\ref{th:hyper} --- 
is equivalent to finding solutions in $k[\hat{t}]$. Problems of this type will be solved in the next section; see also Example~\ref{Exp:HyperDFn!6}(3).
\end{remark}

\begin{example}\label{Exp:HyperDFHn4}
\rm
We reconsider Examples~\ref{Exp:HyperDFHn1}--\ref{Exp:HyperDFHn3} in the light of the proof of Theorem~\ref{th:hyper}.
Following its underlying procedure we compute a finite set of hypergeometric candidates for $L$ with the goal to find all its hypergeometric solutions.
First, take all monic factors that divide $a_0$ and all monic factors that divide $\sigma^{1-2}(a_2)$ and collect them in $A$ and $B$, respectively. 
Clearing in addition all denominators (w.r.t.\ $x$) and defining $\bar{\alpha}=1+h+h x$ and $\bar{\beta}=3+2 h+2 x+3 h x+h x^2$
we get the sets
\begin{align*}
A&=\{1,\bar{\alpha},\bar{\alpha}^2,\bar{\beta},\bar{\alpha}\,\bar{\beta},\bar{\alpha}^2\,\bar{\beta},\bar{\beta}^2,\bar{\alpha}\,\bar{\beta}^2,\bar{\alpha}^2\,\bar{\beta}^2\},\\
B&=\{1,h,h\,x-1,h(h\,x-1)\}.
\end{align*}
Next, we have to loop through all $(a,b)\in A\times B$ and compute for each operator $L_{a,b}$ given in~\eqref{Equ:Lab} (with $\alpha=1$) a finite set of hypergeometric candidates. For instance, for $(a,b)=(\bar{\alpha},1)$ (note that here we do not insist that $a$ is monic) we get (after clearing common factors) the operator
\begin{equation*}
L_{\bar{\alpha},1}=E^0-(3 + 2 x)\,E^1+(2 + x)^2\,E^2
\end{equation*}
and can extract the following set of hypergeometric  candidates for $L_{\bar{\alpha},1}$: $S_{\bar{\alpha},1}=\{1,\frac1{x+1},\frac1{(x+1)^2}\}$ by recursion (for the base case see Theorem~\ref{Thm:GroundField} below). Computing for each $a,b\in A\times B$ such a set $S_{a,b}$ of hypergeometric candidates, one can produce a set of hypergeometric candidates for $L$ by  
\begin{equation}\label{Equ:HgSetCalc}
S=\{\tfrac{a}{b}\,u:\,(a,b)\in A\times B, u\in S_{a,b}\}.
\end{equation}
In particular, $S$ contains the elements of $\{\frac{\bar{\alpha}}{1}\,u:u\in S_{\bar{\alpha},1}\}$. 
Among all candidates for $L$ only the candidate $\frac{\bar{\alpha}}{1}\frac1{x+1}=\frac{1+h+h x}{x+1}$ contributes.
In other words, after some checks (see also Example~\ref{Exp:HyperDFHn5}) we can restrict $S$ to the minimal set $\{\frac{1+h+h x}{x+1}\}$; compare Example~\ref{Exp:HyperDFHn1}.
More precisely, only for $(a,b,u)=(\bar{\alpha},1,\frac1{x+1})$ the operator~\eqref{Equ:Labu} has a nonzero solution in $\mathbb{Q}(x)[h]$. Namely, for $L_{\bar{\alpha},1,\frac{1}{x+1}}$, which is nothing else than~\eqref{Equ:HnPolyRecurrence},
we compute the set of solutions~\eqref{Equ:HnPolyRecurrenceSol}. Note that in Example~\ref{Exp:HyperDFHn3} (using Theorem~\ref{th:riccati}) we searched for solutions in $\set Q(x)(h)$. By the construction of the tuples $(a,b)$ and Theorem~\ref{th:hyper} it follows that the desired solutions have to be searched only in $\set Q(x)[h]$. As claimed earlier in Example~\ref{Exp:HyperDFHn2} it follows that
$E-r_j$ with~\eqref{Equ:rj}
are all right-hand factors of $L$. Together with Example~\ref{Exp:HyperDFHn2} we conclude that~\eqref{Equ:HnSol} are all hypergeometric solutions of the recurrence given in Example~\ref{Exp:RecExp}(I).
\end{example}

\begin{example}\label{Exp:HyperDFn!1}
\rm
In $(k,\sigma)$ from Example~\ref{Exp:FactorialDF} the factorials $m!$ with $(m+1)!=(m+1)m!$ can be rephrased by the variable $p$, and the linear recurrence in Example~\ref{Exp:RecExp}(II) can be represented by the linear difference operator~\eqref{Equ:Operator2} with 
\begin{align*}
a_0&=-2 p^2 (1+x)^2 (2+x)(
7
+3 p
+6 x
+5 p x
+x^2
+2 p x^2
),\\
a_1&=p (1+x) (2+x) (
16
+7 p
+16 x
+12 p x
+3 x^2
+4 p x^2
),\\
a_2&=-(2
+p
+4 x
+2 p x
+x^2).
\end{align*}
While computing a finite set of hypergeometric candidates for $L$ only the factors $a=p\mid a_0$ and $b=1\mid \sigma^{1-2}(a_2)$ contribute. We consider the operator $L_{a,b}$ defined in~\eqref{Equ:Lab} for $a=p,b=1$ which gives
$$L_{p,1}=-2 (1+x)^2 (2+x) (3+2 x)E^0+(2+x) (7+12 x+4 x^2)\,E^1-(1+2 x)\,E^2.$$
From this operator we can extract the hypergeometric candidates 
$1+x$ and $2 (1+x)^2$ that are relevant.
Thus we look at the corresponding operators~\eqref{Equ:Labu} with $(a,b,u)\in\{(p,1,x+1),(p,1,2(x+1)^2)\}$:\begin{align*}
L_{p,1,x+1}=&-2 (7 + 3 p + 6 x + 5 p x + x^2 + 2 p x^2)\,E^0\\
  &\hspace*{1cm}+ 
	(16 + 7 p + 16 x + 12 p x + 3 x^2 + 4 p x^2)E\\
  &\hspace*{2cm}	-(2 + p + 4 x + 2 p x + 
	x^2)E^2,\\
L_{p,1,2(x+1)^2}=&-(7 + 3 p + 6 x + 5 p x + x^2 + 
2 p x^2)\,E^0\\
	&\hspace*{1cm}+(1 + x) (16 + 7 p + 16 x + 12 p x + 3 x^2 + 
	4 p x^2)\,E\\
	&\hspace*{2cm}-2 (1 + x) (2 + x) (2 + p + 4 x + 2 p x + x^2)\,E^2.
\end{align*}
The solutions of the two operators in $\mathbb{Q}(x)[p^{-1}]$ are $1$ and $1+\frac{x^2}{p}$, respectively.  
Thus the factors $E-r_j$ with
$$r_1=(x+1)p\quad\text{ and }\quad r_2=2(x+1)^2p\frac{\sigma((p+x^2)/p)}{(p+x^2)/p}$$
are right-hand factors of $L$, and by Lemma~\ref{lm:riccati} the Riccati equation~\eqref{Equ:Order2Ricc} holds. 
Alternatively, taking $\tilde{r}_1=(x+1)p$, $v_1=1$ and $\tilde{r}_2=2(x+1)p$, $v_2=p+x^2$ we get
$r_j=\tilde{r}\frac{\sigma(v_j)}{v_j}$ for $j=1,2$. In particular, this yields
\begin{equation}\label{Equ:Ricattin!}
a_0\,v_j\dfact{\tilde{r}_j}{0}{\sigma}\,+a_1\,\sigma(v_j)\,\dfact{\tilde{r}_j}{1}{\sigma}+a_2\,\sigma^2(v_j)\dfact{\tilde{r}_j}{2}{\sigma}=0
\end{equation}
for $j=1,2$. Since
\begin{align*}
(\nu+i)!&=\sigma^i(p)|_{p\mapsto \nu!,x\mapsto \nu},\\
\prod_{l=1}^{\nu+i}l!&=(\dfact{\tilde{r}_1}{i}{\sigma}|_{x\mapsto \nu})\prod_{l=1}^{\nu}l!,\\
\prod_{l=1}^{\nu+i}(2\cdot l!)&=(\dfact{\tilde{r}_2}{i}{\sigma}|_{x\mapsto\nu})\prod_{l=1}^{\nu}(2\cdot l!)
\end{align*}
for $\nu,i\in\set N$, we conclude from~\eqref{Equ:Ricattin!} that~\eqref{Equ:n!Sol} are solutions of the recurrence given in Example~\ref{Exp:RecExp}(II).
\end{example}

\begin{remark}\label{Remark:Improvement}
\rm
Following the constructive proof of Theorem~\ref{th:hyper} the size of the set of candidates for an operator $L=\sum_{i=0}^na_iE^i\in k[t][E;\sigma]$ can get very large. First, the set $A\times B$ might be large where $A$ contains all monic factors $a\in k[t]$ with $a\mid a_0$ and $B$ contains all monic factors $b\in K[t]$ with $b\mid\sigma^{-n+1}(a_n)$ (as already exemplified in Example~\ref{Exp:HyperDFHn4}, it might be convenient to drop the constraint that the factors are monic), and second, also the obtained sets of hypergeometric candidates $S_{a,b}$ for the underlying operators $L_{a,b}\in k[E;\sigma]$ might be of considerable size. Thus the set~\eqref{Equ:HgSetCalc} of hypergeometric candidates for $L$ might increase dramatically. The following strategies  implemented in the summation package \texttt{Sigma}  (compare also~\cite[page~155]{PWZ}) might lead to smaller sets $S$:\\ 
(1) One can drop all tuples $(a,b)\in A\times B$ where only one summand in~\eqref{Equ:Lab} remains, i.e., if $|\{0\leq i\leq n\mid \tau(i)=\mu\}|=1$ holds.\\
(2) In case that one can compute the spread in $k[t]$, one can remove all pairs $(a,b)\in A\times B$ with $\spr_\sigma (a,b) \neq \emptyset$. By Theorem~\ref{Thm:SigmaComputable} this is, e.g., possible if the ground field $(k,\sigma)$ is $\sigma$-computable (see Definition~\ref{Def:computable}). In particular, one can compute $\spr_\sigma(a,b)$ if $(k,\sigma)$ is a \pisiSE-field where the constant field $C$ possesses certain algorithmic properties (see Theorem~\ref{Thm:GroundField} below).\\
(3) Looking at~\eqref{eq:uvdecomp} it follows that only those pairs $(a,b)\in A\times B$ need to be considered that can be combined to a hypergeometric solution of $L_{a,b}$. In particular, given the hypergeometric solutions of $L_{a,b}$, one can reuse this information for the computation of the hypergeometric solutions of $L$ itself.\\
(4) Consider the right factor $E-r$ of $L$ with $r=\frac{a}{b}\,\frac{\sigma(v)}{v}$. In (3) we suppose that $\spr_{\sigma}(a,b)=\emptyset$. Using ideas from~\cite{Schneider05c,Petkov:10} one can even find $a,b\in k[t]$ and $v\in k(t)$ with $r=\frac{a}{b}\,\frac{\sigma(v)}{v}$ such that $\spr_{\sigma}(a,b)=\emptyset$ and $\spr_{\sigma}(b,a)=\emptyset$. In other words, for any irreducible factors $\alpha,\beta\in k[t]$ with $\alpha\mid a$ and $\beta\mid b$ there does not exist an $l\in\set Z$ with  $\frac{\sigma^l(\alpha)}{\beta}\in k$. In this case we also say that $\alpha$ and $\beta$ are $\sigma$-coprime. As a consequence, it suffices to search for hypergeometric candidates $\frac{a}{b}$ where the irreducible factors in $a,b$ are $\sigma$-coprime. 
Using this extra insight, one can often decrease the number of candidates. But one has to pay a price: in general, $v$ is no longer a polynomial in $t$, but a rational function in $t$. In particular, one has to look for solutions of~\eqref{Equ:Labu} not in $k[t]$ (resp.\ $k[t^{-1}]$), but in $k(t)$. However, the solver in \texttt{Sigma} is rather efficient and this extra complication is often 
negligible compared to the advantage of obtaining a smaller number of candidates.
\end{remark}

\begin{example}\label{Exp:HyperDFHn5}
\rm
In Example~\ref{Exp:HyperDFHn4} there are $9\cdot 4=36$ tuples in $A\times B$. Among them one can remove all tuples $(a,b)$ with $\deg(a)-\deg(b)\neq1$ due to improvement~(1) in Remark~\ref{Remark:Improvement}. From the remaining $10$ tuples we can exclude all those where $\spr_{\sigma}(a,b)\neq\emptyset$; see improvement~(2) in Remark~\ref{Remark:Improvement}. With
$$\sigma(h)=\frac{1 + h + h x}{1 + x}\quad\text{ and }\quad \sigma^2(h)=\frac{
3 + 2 h + 2 x + 3 h x + h x^2}{(1 + x) (2 + x)}$$
we conclude that one can restrict the candidates further to $(a,b)\in\{(1+h+hx,1),(3 + 2 h + 2 x + 3 h x + h x^2,1)\}$.
\end{example}

\begin{example}\label{Exp:HyperDFHn5b}
\rm
Using in addition refinement (4) of Remark~\ref{Remark:Improvement}, it suffices to take the sets $A=\{1,h,h^2,h^3,h^4\}$ and $B=\{1,h,h^2\}$. In particular, only the $7$ pairs  $(1,1),(h,1),(h^2,1),(h^3,1),(h^4,1),(1,h),(1,h^2)$ of $A\times B$ have to be considered. Applying also improvement~(1) of Remark~\ref{Remark:Improvement} one can restrict this set further to $(h,1)$. This feature is illustrated also within the Mathematica session~\myIn{\ref{MMA:RatCase}} of Example~\ref{Exp:MMAHyper}.
\end{example}

\section{Rational solutions}\label{Sec:RationalSolutions}

Throughout this section,
$(k,\sigma)$ is a difference field and $t$ is a \pisiSE-monomial over $k$.
The results of the previous sections have reduced the problem of
computing hypergeometric solutions of linear ordinary difference equations with coefficients in $k(t)$ to 
computing all solutions in $k[t]$ of linear ordinary 
difference equations with coefficients in $k(t)$ and computing all the hypergeometric candidates
for equations with coefficients in $k$.

Our ultimate goal is to find hypergeometric solutions when $(k,\sigma)$ is given by a tower of \pisiSE-monomials 
over a difference field $(K,\sigma)$ where certain algorithmic subproblems can be handled in $K$; in particular, 
when $(k,\sigma)$ is a \pisiSE-field over $C$ where $C=\Const_{\sigma}(k)$. In this section we work out how
to compute all solutions in $k[t]$ of linear ordinary difference equations with coefficients in $k(t)$ under
the assumption that certain subproblems can be handled in $k$. Here the essential idea is to solve various linear
difference equations in a subfield (with fewer \pisiSE-monomials) and to combine the solutions to obtain a solution 
in the larger field. In order to accomplish this task, we consider a slightly more general problem: We need to find
solutions not in $k[t]$, but in $k(t)$. Moreover, within the reduction we are faced with the {\em problem
of solving parameterized linear difference equations}. To state this problem, let $V \subseteq k$ be a
linear subspace of $k$ over $C=\Const_\sigma(k)$.

\noindent Problem PLDE (Parameterized Linear Difference Equations) with solutions in~$V$:\\[0.2cm]
{\sc Given:} \ $0\neq a=(a_0,a_1,\dots,a_\ell)\in k^{\ell+1}$
and $b=(b_1,b_2,\dots,b_m)\in k^m$.\\[0.2cm]
{\sc Find:} \ a basis of the $C$-linear subspace ${\mathcal V}(a,b,V)$ of $k\times C^m$
of all solutions $(g,c_1,c_2,\dots,c_m) \in V\times C^m$ of
\begin{equation}
\label{PLDE}
a_\ell\, \sigma^\ell g + \cdots + a_1\, \sigma g + a_0 g\ =
\ c_1\,b_1 + c_2\,b_2 + \cdots + c_m b_m.
\end{equation}
Note that $\dim {\mathcal V}(a,b,V) \le \ell+m$ .
In the special case $V=k$, we say that we can {\em compute all solutions of parameterized linear difference equations
with coefficients in $k$\/} if, given any $a\in k^{\ell+1}$ and $b\in k^m$, we can compute a basis of the solution
space ${\mathcal V}(a,b,k)$.


\begin{example}\label{Exp:HyperDFHn6}
\rm
In Example~\ref{Exp:HyperDFHn3} (resp.~Example~\ref{Exp:HyperDFHn1}) we needed the solutions~\eqref{Equ:HnPolyRecurrenceSol} in $\mathbb{Q}(x)[h]$ of the operator~\eqref{Equ:HnPolyRecurrence}.
In the above notation, this problem can be encoded by the solution space
$${\mathcal V}((a_0,a_1,a_2),(0),\mathbb{Q}(x)[h])=\{(g,c)\in\mathbb{Q}(x)[h]\times\set Q : L_{\alpha,1,\frac{1}{x+1}}(g)=c\,0\}$$
with 
\begin{equation}\label{Equ:HnCoeff}
\begin{split}
a_0=&(1 + h + h x) (3 + 2 h + 2 x + 3 h x + h x^2),\\
a_1=&-h (3 + 2 x) (3 + 2 h + 2 x + 3 h x + h x^2),\\ 
a_2=&h (2 + x)^2 (1 + h + h x).
\end{split}
\end{equation}
Using our algorithmic machinery from below, it follows
that 
\begin{equation}\label{Equ:BasisHn}
\{(h,0),(h^2,0),(0,1)\}
\end{equation} 
is a basis of the vector space ${\mathcal V}((a_0,a_1,a_2),(0),\mathbb{Q}(x)[h])$ over $\mathbb{Q}$. As a consequence, we get~\eqref{Equ:HnPolyRecurrenceSol}.
\end{example}

\begin{example}\label{Exp:HyperDFn!2}
\rm
In Example~\ref{Exp:HyperDFn!1} we needed all solutions in $\mathbb{Q}(x)[p^{-1}]$ of 
the operator $L_{p,1,2(x+1)^2}=a_0\,E^0+a_1\,E^1+a_2\,E^2$ with 
\begin{equation}\label{Equ:LSubFieldn!}
\begin{split}
a_0=&-(7 + 3 p + 6 x + 5 p x + x^2 + 2 p x^2),\\ 
a_1=&(1 + x) (16 + 7 p + 16 x + 12 p x + 3 x^2 + 4 p x^2),\\
a_2=&-2 (1 + x) (2 + x) (2 + p + 4 x + 2 p x + x^2).
\end{split}
\end{equation}
In other words, we need a basis of the solution space
\begin{equation}\label{Equ:Vn!}
{\mathcal V}((a_0,a_1,a_2),(0),\mathbb{Q}(x)[p^{-1}])=\{(g,c)\in\mathbb{Q}(x)[p^{-1}]\times\set Q: L_{p,1,2(x+1)^2}(g)=c\,0\}.
\end{equation}
Using the algorithmic machinery given below we compute the basis
\begin{equation}\label{Equ:Basisn!}
\{(0,1),(1+\tfrac{x^2}{p},0)\} 
\end{equation}
which gives the solution $g=1+\frac{x^2}{p}$ of $L_{p,1,2(x+1)^2}(g)=0$. 
\end{example}

\begin{remark}\label{Remark:Summation}
\rm
Problem PLDE plays an important role in symbolic summation \cite{PWZ}. Apart from solving difference 
equations, it contains parameterized telescoping (where $\ell=1,a_0=-1,a_1=1$) and hence Zeilberger's creative
telescoping paradigm~\cite{Zeilberger:91}. It also covers an important subproblem in summation of holonomic functions~\cite{Zeilberger:90a,Chyzak:00,Koutschan:13}. This observation has been explored further in the setting of \pisiSE-fields in \cite{Schneider:07a,Schneider:13a}.   
Concrete examples with this refined toolbox can be found, e.g., in the proof of Lemma~\ref{Lemma:DoubleSum}. 
\end{remark}

Before stating the main result of this section (Theorem~\ref{Thm:RatSolver}), we need to introduce the
following associated problem.

Let $(k,\sigma)$ be a difference field.  We call $u,v \in k$ \emph{similar} (and write $u\sim_{k,\sigma} v$)
if $u = v\, \sigma(w)/w$ for some $w \in k^\ast$.
The \emph{pseudo-orbit} problem is, given $u,v \in k^\ast$, to
compute the set
\begin{equation}
\label{Gamma}
\Gamma(u,v;k) = \{\gamma\in\mathbb{Z} \st u^\gamma \sim_{k,\sigma} v\}\,.
\end{equation}
When $\sigma$ is the identity on $k$, this problem reduces to the {\em orbit problem},
\ie to the problem of finding all $\gamma \in \bb{Z}$ such that $u^\gamma = v$, which is solved for
``reasonable'' fields in~\cite{AbraBron2000} whenever $u$ is
not a root of unity.

\begin{theorem}\label{Thm:RatSolver}
	Let $(k,\sigma)$ be a difference field and $t$ be a \pisiSE-monomial over $k$. 
	\begin{enumerate}
		\item If one can solve the
		pseudo-orbit problem in $k$, compute
		all the hypergeometric candidates for equations with coefficients in $k(t)$,
		and compute all solutions of parameterized linear difference equations with coefficients in $k$,
		then one can compute all solutions of parameterized linear difference equations with coefficients in $k[t]$.
		\item If one can compute in addition dispersions in $k[t]$, then one can compute all solutions of parameterized linear difference equations with coefficients in $k(t)$.
	\end{enumerate} 
\end{theorem}


We prove the general case~(2) of Theorem~\ref{Thm:RatSolver} by giving an algorithm which, given $a \in k(t)^{\ell+1}$ and
$b \in k(t)^m$, finds a basis of the solution space ${\mathcal S} := {\mathcal V}(a,b,k(t))$ of equation~\eqref{PLDE}.
It consists of a preprocessing step, and of three main steps whose description comprises the rest of this section and
is interspersed with a series of auxiliary lemmas, theorems and corollaries which prove correctness of the algorithm. Part~(1) of Theorem~\ref{Thm:RatSolver} will follow by a simplified version of the proposed algorithm.

\textit{0. Preprocessing.} By clearing denominators we may assume that $a \in k[t]^{\ell+1}$ and
$b \in k[t]^m$. Since in the case $\ell=0$ solution of \eqref{PLDE} is straightforward, we may assume that $\ell \geq 1$.

\textit{1. Denominator bounding.} We compute a polynomial $d\in k[t]\setminus\{0\}$ (called a {\em universal
denominator\/} for rational solutions of equation~\eqref{PLDE}) such that for any
$(g, c_1, \dots, c_m) \in {\mathcal S}$, we have $d\,g\in k[t]$.

\begin{remark}
	If one is only interested in polynomial solutions in $k[t]$, one can skip step~1 of the algorithm and can proceed with step~2. This is for instance the case in Example~\ref{Exp:HyperDFHn6}: we are only interested in solutions in $k[t]$ with $k:=\set Q(x)$ and $t:=h$, and not in $k(t)$.
\end{remark}

For the rational case ($\sigma(t)=t+1$ and $k=\Const_\sigma(k)$) and for the $q$-rational case ($\sigma(t)=q\,t$, 
$k=\Const_\sigma(k)$, $q$ not a root of unity) this problem is solved in~\cite{Abramov:89a, Abramov:95a}. 
Bronstein~\cite{Bronstein2000} generalizes these ideas to an algorithm for \pisiSE-extensions under the assumption
that one can compute dispersions in $k[t]$. Namely, using Theorems 8 and 10 of~\cite{Bronstein2000} one can compute a universal denominator if $t$ is a \sigmaSE-monomial, and find a universal denominator -- up to a factor of the form $t^m$ -- if $t$ is a \piE-monomial. This aspect (among others) is also elaborated in~\cite[Theorem~2]{Schneider:04c} in the setting of \piE-monomials; for a generalizaton to coupled higher order systems, see~\cite{MS:2018}. For first-order parameterized difference equations, Karr~\cite{Karr81} computes the extra factor $t^m$ for \piE-monomials
under the assumption that one can solve the pseudo-orbit problem in $k$; a detailed proof can be found in~\cite[Theorem~6]{Schneider:04c}. In order to derive the extra factor $t^m$ for higher-order equations (for special cases see also~\cite[Sec.~4.4]{Schneider:04c}), we
need in addition 
all the hypergeometric candidates for equations with coefficients in $k(t)$; see Theorem~\ref{th:PiBounds}.
In short, by Theorems 8 and 10 of~\cite{Bronstein2000} (or more concretely by~\cite[Theorem~2]{Schneider:04c}) together with Theorem~\ref{th:PiBounds} below we obtain the following result.

\begin{corollary}\label{Cor:DenBound}
Let $(k,\sigma)$ be a difference field and $t$ a \pisiSE-monomial over $k$. 
If one can solve the pseudo-orbit problem in $k$, compute dispersions in $k[t]$, and compute
all the hypergeometric candidates for equations with coefficients in $k(t)$, then one can compute a universal
denominator for solutions of parameterized linear difference equations with coefficients in $k(t)$.
\end{corollary}

\begin{example}\label{Exp:HyperDFHn7}
\rm
	In Example~\ref{Exp:HyperDFHn6} we considered the problem to compute a basis of ${\mathcal V}((a_0,a_1,a_2),(0),\mathbb{Q}(x)[h])$. Using the algorithms from~\cite{Bronstein2000,Schneider:04c} it follows even that $d=1$ is a denominator bound of $V={\mathcal V}((a_0,a_1,a_2),(0),\mathbb{Q}(x)(h))$. This implies that $V={\mathcal V}((a_0,a_1,a_2),(0),\mathbb{Q}(x)[h]).$
\end{example}

\begin{example}\label{Exp:HyperDFn!3}
\rm
Following Example~\ref{Exp:HyperDFn!2} we have to calculate a solution of~\eqref{Equ:Vn!} with~\eqref{Equ:LSubFieldn!}. Using the algorithms given in~\cite{Bronstein2000,Schneider:04c}, it follows even that
\begin{equation}\label{Equ:n!DenBound}
{\mathcal V}((a_0,a_1,a_2),(0),\mathbb{Q}(x)(p))={\mathcal V}((a_0,a_1,a_2),(0),\mathbb{Q}(x)[p^{-1}])
\end{equation}
holds; here $p$ takes over the role of $t$. This means that a universal denominator is given by $d=p^m$ for some $m\in\set N$. Using our new algorithm provided in Theorem~\ref{th:PiBounds} we will compute in Example~\ref{Exp:HyperDFn!6} the denominator bound $d=p^m$ with $m=1$.
\end{example}

Given a universal denominator $d$, we substitute $f/d$ for $g$ in \eqref{PLDE}, and look for all polynomial
solutions $(f, c_1, \dots, c_m) \in k[t] \times C^m$ of the resulting parameterized difference equation
\begin{equation}
\label{Equ:PLDEPoly}
\tilde a_\ell\,\sigma^\ell f + \cdots + \tilde a_1\, \sigma f + \tilde a_0 f\ =
\ c_1\,b_1 + c_2\,b_2 + \cdots + c_m\,b_m
\end{equation}
where $\tilde a_i = a_i/\sigma^i d$ for $i = 0, 1, \ldots, \ell$.
Note that the set $V'$ of all such solutions forms a subspace of $k[t]\times C^m$ over $C$ whose dimension, say $r$,
is bounded by $m+\ell$. Namely, given a basis $\{(f_i, c_{i1}, \ldots, c_{im})\}_{1\leq i\leq r}$ of $V'$, 
$\{(f_i/d, c_{i1}, \ldots, c_{im})\}_{1\leq i\leq r}$ provides a basis of our desired solution space ${\mathcal S}$.

We remark that the coefficients $\tilde a_i$ are usually elements from $k(t)$. Repeating the preprocessing step 0 we clear denominators and may assume that the $\tilde{a}_i$ and $b_i$ are again elements from $k[t]$.

\begin{example}\label{Exp:HyperDFn!4}
\rm
	In Example~\ref{Exp:HyperDFn!3} (see also Example~\ref{Exp:HyperDFn!6})	we derived the denominator bound $d=p$ for~\eqref{Equ:n!DenBound}. Thus the solutions in~\eqref{Equ:n!DenBound} have the form $g=g_{-1}p^{-1}+g_0p^0=\frac{p\,g_{-1}+g_0}{p}$ with $g_{-1},g_0\in\set Q(x)$. So we set $\tilde{a}_i=\frac{a_i}{\sigma^i(p)}=\frac{a_i}{\dfact{(x+1)}{i}{\sigma}p}$
	 and look for all solutions 
	$f=p\,g_{-1}+g_0$ of $\tilde{a}_0\,f+\tilde{a}_1\,\sigma(f)+\tilde{a}_2\,\sigma^2(f)=0$. Note that in this particular instance the denominator in each of the $\tilde{a}_i$ is precisely $p$. By clearing this common denominator we get
	\begin{equation}\label{Equ:tildeai}
\begin{split}
\tilde{a}_0=&-(7 + 3 p + 6 x + 5 p x + x^2 + 2 p x^2),\\ 
\tilde{a}_1=&\ 16 + 7 p + 16 x + 12 p x + 3 x^2 + 4 p x^2,\\
\tilde{a}_2=& -2 (2 + p + 4 x + 2 p x + x^2).
\end{split}
\end{equation}	
In other words, we search for all solutions in ${\mathcal V}((\tilde a_0,\tilde a_1,\tilde a_2),(0),\mathbb{Q}(x)[p])$ with~\eqref{Equ:tildeai}. 
\end{example}

\textit{2. Degree bounding:} In order to determine all solutions of~\eqref{Equ:PLDEPoly}, we try to find a
{\em degree bound\/} for $f$, \ie an integer $b\in\set N\cup\{-1\}$ such that $\deg f \leq b$ for any
$(f,c_1, \ldots, c_m)\in V'$.

In the rational case~\cite{Abramov:89b,Petkovsek1992} and in the $q$-rational case~\cite{Abramov:95a}, such a degree
bound can be computed. Moreover, for first-order difference equations this bound can be computed~\cite{Karr81} for a
\pisiSE-monomial $t$ over $k$ provided that one can solve the pseudo-orbit problem and one can solve parameterized 
first-order difference equations with coefficients in $k$; for more details and generalizations, see
Schneider~\cite{Schneider:05b}. Combining the results of Theorems~\ref{th:PiBounds} and \ref{th:SigmaBounds} below,
the higher-order case can be summarized as follows.

\begin{corollary}\label{Cor:DegBound}
Let $(k,\sigma)$ be a difference field and let $t$ be a \pisiSE-monomial over $k$. If one can solve the pseudo-orbit
problem in $k$, one can compute all the hypergeometric candidates for equations with coefficients in $k(t)$, and one
can compute all solutions of parameterized linear difference equations with coefficients in $k$, then one can compute
a degree bound for polynomial solutions of parameterized linear difference equations with coefficients in $k(t)$.
\end{corollary}

\begin{example}\label{Exp:HyperDFHn8}
\rm
	Continuing  Example~\ref{Exp:HyperDFHn7} we have to compute a basis of the solution space $V={\mathcal V}((\tilde a_0,\tilde a_1,\tilde a_2),(0),\mathbb{Q}(x)[h])$ with $\tilde a_i=a_i$. In Example~\ref{Exp:HyperDFHn9} we will see that a degree bound is $b=2$. Together with the algorithms summarized in Step~3 we obtain the basis~\eqref{Equ:BasisHn} of $V$.
\end{example}

\begin{example}\label{Exp:HyperDFn!5}
	\rm
	Continuing Example~\ref{Exp:HyperDFn!4} we need a basis of the solution space $V={\mathcal V}((\tilde a_0,\tilde a_1,\tilde a_2),(0),\mathbb{Q}(x)[p])$. By Example~\ref{Exp:HyperDFn!6}(2) below we will see that $b=1$ is a degree bound of $V$. Note that this agrees with the observation in Example~\ref{Exp:HyperDFn!4}. Using the reduction technique of the next subsection we find the needed basis~\eqref{Equ:Basisn!}.
\end{example}

\textit{3. Finding the polynomial coefficients.}
Having obtained a degree bound $b$, one looks for $c_i\in C$ and $f_i\in k$ such that~\eqref{Equ:PLDEPoly} holds for
$f=\sum_{i=0}^b f_i t^i$. If $t$ is a \piE-monomial, this task can be accomplished by a very general algorithm
summarized in~\cite[Theorem~1]{Bronstein2000}: it works for unimonomial extensions which possess a {\em special
element}, \ie an element $h\in k(t)\setminus k$ with $\sigma h / h \in k$. Since in \sigmaSE-extensions such
elements do not exist, we rely on the following strategy that is applicable to both, \piE-monomials and
\sigmaSE-monomials (for the first-order case, see~\cite{Karr81}, and for the higher-order case together with detailed
proofs, see~\cite{Schneider:05a}): First determine the possible leading coefficients $f_b$ and the parameters $c_i$ by
solving a specific instance of a parameterized linear difference equation in $(k,\sigma)$, then substitute the
obtained solution into~\eqref{Equ:PLDEPoly} and look for $f - f_b t^b = \sum_{i=0}^{b-1} f_it^i$ recursively.

During the substitution mentioned above, the number of unknown parameters might increase (or decrease). 
This is the reason why we need to solve parameterized difference equations and not merely homogeneous 
difference equations (without parameters).

Summarizing, one can obtain the solution of~\eqref{Equ:PLDEPoly}
by solving several parameterized linear difference equations in the smaller field $(k,\sigma)$. 
In particular, collecting all the computational properties of Corollaries~\ref{Cor:DenBound} and~\ref{Cor:DegBound} gives part~(2) of Theorem~\ref{Thm:RatSolver}.
If one solves equations only in $k[t]$, step~1 of our algorithm can be skipped. Thus the computation of spreads in $k[t]$ (see Corollary~\ref{Cor:DenBound}) is not required, yielding part~(1) of Theorem~\ref{Thm:RatSolver}.

In Subsections~\ref{pi-case} resp.~\ref{sigma-case} we prove Theorems~\ref{th:PiBounds} resp.~\ref{th:SigmaBounds}
in order to establish Corollaries~\ref{Cor:DenBound} and~\ref{Cor:DegBound},
and hence also Theorem~\ref{Thm:RatSolver}.

\subsection{Denominator bounds and degree bounds for \piE-monomials}
\label{pi-case}

In this subsection, $t$ is a \piE-monomial over $(k,\sigma)$ satisfying $\sigma t = \alpha t$
with $\alpha \in k^\ast$.

Any $L \in k[t][E;\sigma]\setminus\{0\}$ can be uniquely decomposed as
\begin{equation}
\label{eq:project}
L = \sum_{j=\nu(L)}^{\deg L } t^j L_j
\end{equation}
where the $L_j$ are in $k[E;\sigma]$
and $L_{\nu(L)} \ne 0 \ne L_{\deg L }$; recall that $\nu(L)$ and $\deg L$ denote the trailing and leading degrees of $L$, respectively.

The following is the basis of our algorithm for bounding
the order and degree of solutions in $k[t,t^{-1}]$ of operators with
coefficients in $k[t]$; see also~\cite[Lemma~3]{Schneider:05b}. Here $\nu(f)$ denotes the order (sometimes also called valuation) of a Laurent polynomial $f\in k[t,t^{-1}]$.

\begin{lemma}
\label{lm:expbound}
Let $L\in k[t][E;\sigma]\setminus\{0\}$ and $b_1,\dots,b_m \in k[t,t^{-1}]$.
If there are integers $\gamma \le \delta$,
$y_\gamma,\dots,y_\delta \in k$ and $c_1,\dots,c_m \in \Const_\sigma(k)$
such that $y_\gamma y_\delta \ne 0$ and
\begin{equation}
\label{eq:expbound}
L\paren{y_\gamma t^{\gamma}+\cdots+y_{\delta} t^{\delta}}
= c_1 b_1+\cdots+c_m b_m\,,
\end{equation}
then,
\begin{enumerate}
\item[(i)]
either $\gamma \ge \min_{1\le j\le m}\nu(b_j) - \nu(L)\ $
or $\ L_{\nu(L)}(y_\gamma t^\gamma) = 0$, and
\item[(ii)]
either $\delta \le \max_{1\le j\le m}\deg b_j  - \deg L \ $
or $\ L_{\deg L }(y_\delta t^\delta) = 0$
\end{enumerate}
where the minimum and maximum are taken over those $j$ for which $b_j \ne 0$.
\end{lemma}

\begin{proof}
Using the decomposition~(\ref{eq:project}) we have
$$
L\paren{y_\gamma t^{\gamma}+\cdots+y_{\delta} t^{\delta}} \ =
\sum_{j=\nu(L)}^{\deg L } t^j L_j \paren{\sum_{i=\gamma}^{\delta} y_i t^i} \ =
\sum_{j=\nu(L)}^{\deg L } \sum_{i=\gamma}^{\delta} t^j L_j (y_i t^i).
$$
Since $E(y_i t^i) = \sigma(y_i) \sigma(t)^i = (\sigma(y_i) \alpha^i) t^i$
where $\alpha = \sigma(t)/t \in k$, it follows that for each $j$,
$L_j(y_i t^i) = z_{ji} t^i$ for some $z_{ji} \in k$. Therefore,
\begin{equation}
\label{eq:lsum}
L\paren{y_\gamma t^{\gamma}+\cdots+y_{\delta} t^{\delta}}\ =
\sum_{j=\nu(L)}^{\deg L } \sum_{i=\gamma}^{\delta} z_{ji} t^{j+i}\,.
\end{equation}

\noindent
(i) Suppose that $\gamma + \nu(L) < \min_{1\le j\le m}\nu(b_j)$.
It then follows from~(\ref{eq:expbound}) and~(\ref{eq:lsum}) that
$z_{\nu(L), \gamma} = 0$, hence that $L_{\nu(L)}(y_\gamma t^\gamma) = 0$.

\noindent
(ii) Suppose that $\delta + \deg L > \max_{1\le j\le m}\deg b_j$.
It then follows from~(\ref{eq:expbound}) and~(\ref{eq:lsum}) that
$z_{\deg L , \delta} = 0$, hence that $L_{\deg L }(y_\delta t^\delta) = 0$.
\end{proof}

We use the following fact that has been exploited already in~\cite{Karr81} to handle the first-order case;
for details and clarifications, see~\cite[Lemma~3]{Schneider:04c} and~\cite[Lemma~6]{Schneider:05b}.

\begin{lemma}\label{Lemma:Uniqueness}
Let $(k,\sigma)$ be a difference field and let $u,v\in k$ where $u$ is not a  $\sigma$-radical over $k$. Then
$\Gamma(u,v;k)$  $($as defined in~\eqref{Gamma}$)$ has at most one element.
\end{lemma}
\begin{proof}
Suppose that $\gamma_1\neq\gamma_2$
are in $\Gamma(u,v;k)$. Then $u^{\vert\gamma_1 - \gamma_2\vert} \sim_{\sigma,k} 1$, and hence $u$ is 
a $\sigma$-radical, a contradiction.
\end{proof}

\begin{theorem}
\label{th:PiBounds}
Let $(k,\sigma)$ be a difference field and let $t$ be a $\Pi$-monomial over $k$ with $\sigma t = \alpha t$.
If we can solve pseudo-orbit problems in $k$ and if we can compute
all the hypergeometric candidates for equations with coefficients in $k(t)$, then we can bound the order and degree of solutions in $k[t,t^{-1}]$ of parameterized linear difference equations with coefficients in $k[t]$ and inhomogeneous parts in $k[t,t^{-1}]$.
\end{theorem}
\begin{proof}
Let $L = \sum_{i=0}^n a_i E^i \in k[t][E;\sigma]$
be a linear ordinary difference operator with coefficients in $k[t]$
where $a_0 a_n \ne 0$, and take the inhomogeneous parts
$b_1,\dots,b_m \in k[t,t^{-1}]$. Let $c_1,\dots,c_m\in \Const_\sigma(k)$ and
$y = y_\gamma t^{\gamma}+\cdots+y_{\delta} t^{\delta}$ 
with $\gamma\leq\delta$, $y_i\in k$ for $\gamma\leq i\leq\delta$ and $y_\gamma y_\delta \ne 0$ be such that~\eqref{eq:expbound} holds.
First suppose that 
\begin{equation}\label{Equ:TailLCConstraint}
L_{\nu(L)}(y_\gamma t^\gamma)
\ =\ 0\ =\ L_{\deg L }(y_\delta t^\delta)
\end{equation}
where $L_{\nu(L)}$ and $L_{\deg L }$ are given by~(\ref{eq:project}).
Following the arguments in the proof of Lemma~\ref{lm:riccati} we conclude that $E - \alpha^\gamma \sigma y_\gamma / y_\gamma$
and $E - \alpha^\delta \sigma y_\delta / y_\delta$ are right factors
of $L_{\nu(L)}$ and $L_{\deg L }$, respectively.
Since we can compute all the hypergeometric candidates for equations
with coefficients in $k$, we can compute from the input operators $L_{\nu(L)}$ and $L_{\deg L }$ the
finite sets
$S_\nu,S_{\deg} \subset k$ such that
$$
\alpha^\gamma\, \frac{\sigma y_\gamma}{y_\gamma} = u_\nu \frac{\sigma v}v
\mbox{\quad and\quad }
\alpha^\delta\, \frac{\sigma y_\delta}{y_\delta} = u_{\deg} \frac{\sigma w}w
$$
for some $u_\nu \in S_\nu$, $u_{\deg} \in S_{\deg}$ and $v,w \in k^\ast$.
This is equivalent to
$$
\alpha^\gamma = u_\nu \frac{\sigma v_\nu}{v_\nu}
\mbox{\quad and\quad }
\alpha^\delta = u_{\deg} \frac{\sigma v_{\deg}}{v_{\deg}}
$$
where $v_\nu = v/y_\gamma \in k^\ast$ and $v_{\deg} = w/y_\delta \in k^\ast$.
Therefore $\gamma \in \Gamma(\alpha,u_\nu;k)$
and $\delta \in \Gamma(\alpha,u_{\deg};k)$.
Since we can solve pseudo-orbit problems in $k$, we can compute
the sets $\Gamma(\alpha,u;k)$ and $\Gamma(\alpha,v;k)$ for any $u\in S_{\deg}$ and $v\in S_{\nu}$, in particular for $u_{\deg}\in S_{\deg}$ and $u_{\nu}\in S_{\nu}$.
By Theorem~\ref{thm:pisiProp}(i), $\alpha$ is not a $\sigma$-radical over $k$, 
hence by Lemma~\ref{Lemma:Uniqueness} each of these sets has at most one element.
Thus $\bigcup_{u \in S_\nu} \Gamma(\alpha,u; k)$ and
$\bigcup_{u\in S_{\deg}} \Gamma(\alpha,u; k)$ are finite non-empty sets
of candidates for $\gamma$ and $\delta$, respectively. With
\begin{equation}\label{Equ:DefinemM}
m := \min \bigcup_{u\in S_\nu} \Gamma(\alpha,u; k)\mbox{\quad and\quad }M := \max \bigcup_{u\in S_{\deg}} \Gamma(\alpha,u; k)
\end{equation}
we get $\gamma\geq m$ and $\delta\leq M$ by construction.\\
Now suppose that~\eqref{Equ:TailLCConstraint} does not hold. If $L_{\nu(L)}(y_\gamma t^\gamma)\neq0$ 
then it follows by part (i) of Lemma~\ref{lm:expbound} that $\gamma \ge \min_{1\le j\le m}\nu(b_j) - \nu(L)$, and  
if $L_{\deg L }(y_\delta t^\delta)\ \neq\ 0$, then $\delta \le \max_{1\le j\le m}\deg b_j  - \deg L \ $ by part (ii) of Lemma~\ref{lm:expbound}. So regardless of whether~\eqref{Equ:TailLCConstraint} holds or not, we have
\begin{align*}
\gamma&\geq\min(m,\min_{1\le j\le m}\nu(b_j) - \nu(L)),\\
\delta&\leq\max(M,\max_{1\le j\le m}\deg b_j  - \deg L)
\end{align*}
where the bounds on the right-hand sides can be computed from the given inhomogeneous parts $b_1,\dots,b_m$ and the operators $L_{\nu(L)}$ and $L_{\deg L }$, respectively. 
\end{proof}

\begin{example}\label{Exp:HyperDFn!6}
\rm
We demonstrate the algorithm given in the proof of Theorem~\ref{th:PiBounds}.\\
(1) Given the operator $L=L_{p,1,2(x+1)^2}
=a_0\,E^0+a_1\,E+a_2\,E^2$ with~\eqref{Equ:LSubFieldn!} we compute with $\deg L=1$ and $\nu(L)=0$ the operators
$$L_{1}=-(3 + 2 x)\,E^0+(7 + 12 x + 4 x^2)\,E -2 (2 + x) (1 + 2 x)\,E^2$$
and 
\begin{multline*}
L_{0}=-(7 + 6 x + x^2)\,E^0 +(1 + x) (4 + x) (4 + 3 x)\,E\\ 
-2 (1 + x) (2 + x) (2 + 
4 x + x^2)\,E^2;
\end{multline*}
note that in the derived $L_1$ we removed the common factor $x+1$. For $L_{1}$ we can compute the set of hypergeometric candidates 
$S_{\deg}=\{1\}$; more precisely we computed the hypergeometric solution $1$ of $L_1$. In addition, we obtain that $\Gamma(x+1,1;\set Q(x))=\emptyset$ (see Example~\ref{Exp:KarrM}), and thus set $M=0$.\\
For $L_0$ we get the hypergeometric candidates $S_{\nu}=\{\frac1{2 (1 + x)},\frac{1 + x}{x^2}\}$ and we obtain $\Gamma(x+1,\frac1{2 (1 + x)};\set Q(x))=\emptyset$
and 
$\Gamma(x+1,\frac{1 + x}{x^2};\set Q(x))=\{-1\}$; see Example~\ref{Exp:KarrM}. This gives
$m=\min(0,-1)=-1$. Summarizing, the solutions of $L(g)=0$ are of the form $g=g_0+g_{-1}p^{-1}$. In particular, $d=p$ is a denominator bound of~\eqref{Equ:Vn!}.\\
(2) Similarly, we can apply this degree bounding method to the operator $\tilde{L}=\tilde{a}_0\,E^0+\tilde{a}_1\,E+\tilde{a}_2\,E^2$ with~\eqref{Equ:tildeai} and obtain $m=0$ and $M=1$; note that $\tilde{L}=\frac{a_0}{p}\,E^0+\frac{a_1}{(x+1)p}\,E^1+\frac{a_2}{(x+1)(x+2)p}\,E^2$. Thus $b=1$ is a degree bound of ${\mathcal V}((\tilde a_0,\tilde a_1, \tilde a_2),(0),\mathbb{Q}(x)[p])$. \\
(3) As mentioned in Remark~\ref{Remark:SwithToPolynomialCase} there is an alternative approach to compute a basis of~\eqref{Equ:Vn!} by working with the \pisiSE-field $(k(x)(\hat{p}),\sigma)$ with $\sigma(x)=x+1$ and $\sigma(\hat{p})=\frac1{x+1}\hat{p}$. Namely, define $a'_i=a_i|_{p\to\hat{p}^{-1}}$ for $i=0,1,2$. 
Then one is faced with the problem of finding a basis of ${\mathcal V}((a'_0,a'_1,a'_2),(0),\mathbb{Q}(x)[\hat{p}])$. Applying our degree bounding algorithm shows that the solutions are of the form $g'=g'_0+g'_1\hat{p}$. Indeed, we find the basis 
$\{(0,1),(1+x^2\hat{p},0)\}$ which gives the basis~\eqref{Equ:Basisn!} by replacing $\hat{p}$ with $p^{-1}$.
\end{example}

\begin{remark}
\rm
If we consider homogeneous linear difference equations, i.e., we search for solutions $y = y_\gamma t^{\gamma}+\cdots+y_{\delta} t^{\delta}$
with $y_\gamma y_\delta \ne 0$ such that $L(y)=0$, then~\eqref{Equ:TailLCConstraint} holds.
In particular, the following extra information can be gained. 
If either of the sets $\bigcup_{u \in S_\nu} \Gamma(\alpha,u; k)$ or
$\bigcup_{u\in S_{\deg}} \Gamma(\alpha,u; k)$ is empty, or if 
for $m$ and $M$ given in~\eqref{Equ:DefinemM} we have $m>M$, 
then $L(y) = 0$ has no nonzero solution in $k[t,t^{-1}]$. Otherwise, any solution is of the form
$y=y_m t^{m}+\cdots+y_{M} t^{M}$, so $m \le \gamma \le \delta \le M$ are the desired bounds.
\end{remark}

\subsection{Degree bounds for \sigmaSE-monomials}
\label{sigma-case}




We start with a simple observation that is needed in this subsection but also in Subsection~\ref{Sec:ControlConstants} below.

\begin{lemma}\label{Lemma:XExt}
Let $(F(x),\sigma)$ be a difference field extension
of $(F,\sigma)$ and $\sigma x = x$. Then 

$\Const_\sigma\, F(x)=(\Const_\sigma F)(x)$.
\end{lemma}
\begin{proof}
Obviously $(\Const_\sigma F)(x)\subseteq \Const_\sigma\, F(x)$. Now take $f\in \Const_\sigma F(x)$. 
First suppose that $x$ is transcendental over $F$ .
Write $f = p/q$ where $p, q \in F[x]$ are coprime and $q \ne 0$ is monic.
From $\sigma f = f$ we obtain $q\, \sigma p = p\, \sigma q$, whence $q \mid \sigma q$. Since $\deg\sigma q = \deg q$,
there is $u \in F^\ast$ such that $\sigma q = u\, q$. From $\sigma x = x$ and monicity of $q$ it follows by comparison
of leading coefficients that $u = 1$. So $\sigma q = q$ and $\sigma p = p$. Comparing coefficients of like powers of
$x$ in these two equations we find that $p, q \in (\Const_\sigma F)[x]$, hence $f \in (\Const_\sigma F)(x)$.\\ 
Otherwise suppose that $x$ is algebraic over $F$ and let $\mu=\mu_0+\dots+\mu_d y^d\in F[y]$ be the minimal polynomial, i.e., $\mu$ is monic, $\mu(x)=0$ and $\deg(\mu)=d$ is minimal among all such polynomials. Then we can write
$f=\sum_{i=0}^{d-1}f_ix^i$ with $f_i\in F$ and obtain
$$0=\sigma(f)-f=\sum_{i=0}^{d-1}\sigma(f_i)x^i-\sum_{i=0}^{d-1}f_ix^i=\sum_{i=0}^{d-1}(\sigma(f_i)-f_i)x^i.$$ 
Since $\mu$ is the minimal polynomial, $x^0,\dots,x^{d-1}$ are linearly independent over $F$. Therefore $\sigma(f_i)-f_i=0$, i.e., $f_i\in\Const_\sigma F$ for all $0\leq i<d$ and thus again $f \in (\Const_\sigma F)(x)$. -- So, in both cases it follows that $\Const_\sigma\, F(x) \subseteq (\Const_\sigma F)(x)$, and the lemma is proven.
\end{proof}


In Theorem~\ref{th:SigmaBounds} (using Lemmas~\ref{Lemma:DegBoundn=0} and~\ref{Lemma:DegBound:n->n+1}) a rather involved algorithm is elaborated that determines a degree bound for \sigmaSE-extensions. Similarly to~\cite{Singer:91} we have to apply rather technical constructions to prove its termination. In order to accomplish this task, we will explore the following abstract difference fields further.
Let $(F, \sigma)$ be a difference field and $t$ be a \sigmaSE-monomial over $F$ with $\sigma t = t + \eta$. 
We extend $\sigma$ to the field $F((t^{-1}))$ of formal Laurent series in $t^{-1}$ by defining
$\sigma t^{-1} = (\sigma t)^{-1} = 1/(t + \eta) = t^{-1}/(1 + \eta t^{-1})$
(cf.~\cite[Example 1.2]{Singer:97}).
Then $(F((t^{-1})),\sigma)$ becomes a difference field extension of $(F(t),\sigma)$ with
%
\begin{eqnarray}
\label{auto}
\sigma \left(\sum_{i \ge i_0} a_i t^{-i}\right) &=&  
\sum_{i \ge i_0} \left(\sum_{j=i_0}^i {i-1 \choose i-j} (-\eta)^{i-j} \,\sigma a_j\right) t^{-i}, \\
\sigma^{-1} \left(\sum_{i \ge i_0} a_i t^{-i}\right) &=&  
\sum_{i \ge i_0} \left(\sum_{j=i_0}^i {i-1 \choose i-j} (\sigma^{-1} \eta)^{i-j} \,\sigma^{-1} a_j\right) t^{-i}
\nonumber 
\end{eqnarray}
%
%
for every $\sum_{i \ge i_0} a_i t^{-i} \in F((t^{-1}))$.

\begin{lemma}\label{Lemma:LaurentExt}
Let $(F, \sigma)$ be a difference field, and let $t$ be a \sigmaSE-monomial over $F$ with $\sigma t = t + \eta$. Then
$\Const_\sigma\, F((t^{-1})) = \Const_\sigma F$.
\end{lemma}
\begin{proof}
%
Clearly $\Const_\sigma F \subseteq \Const_\sigma{F((t^{-1}))}$. To prove the converse,
take any $y \in \Const_\sigma{F((t^{-1}))}$. If $y = 0$ then $y \in \Const_\sigma F$. Otherwise,
there are $i_0 \in \mathbb{Z}$ and $a_{i_0}, a_{i_0+1}, \ldots \in F$ such that
$y = \sum_{i \ge i_0} a_i t^{-i}$ and $a_{i_0} \ne 0$. From $\sigma y = y$ it follows by (\ref{auto}) that
\begin{equation}
\label{coeffs}
\sum_{j=i_0}^i {i-1 \choose i-j} (-\eta)^{i-j} \,\sigma a_j
\ =\ a_i, \qquad {\rm\ for\ all\ }i \ge i_0. 
\end{equation}
For $i = i_0$ this gives $\sigma a_{i_0} \ =\ a_{i_0}$, hence $a_{i_0} \in \Const_\sigma F$.
For $i = i_0+1$, (\ref{coeffs}) yields $\sigma a_{i_0+1} - \eta i_0 a_{i_0} \ =\ a_{i_0+1}$.
If $i_0 \ne 0$ then $\sigma w - w = \eta$ where $w = a_{i_0+1}/(i_0 a_{i_0}) \in F$. 
By Theorem \ref{thm:pisiProp}(ii), this is impossible as $t$ is a \sigmaSE-monomial over $F$,
so $i_0 = 0$ and
\begin{equation}
\label{coeffs0}
\sum_{j=0}^i {i-1 \choose i-j} (-\eta)^{i-j} \,\sigma a_j
\ =\ a_i, \qquad {\rm\ for\ all\ }i \ge 0. 
\end{equation}
Assume $n \ge 1$ and $a_1 = a_2 = \cdots = a_{n-1} = 0$. 
For $i = n$, (\ref{coeffs0}) gives $\sigma a_n \ =\ a_n$, hence $a_n \in \Const_\sigma F$.
For $i = n+1$, (\ref{coeffs0}) yields $\sigma a_{n+1} - \eta n a_n \ =\ a_{n+1}$.
If $a_n \ne 0$ then $\sigma w - w = \eta$ where $w = a_{n+1}/(n a_n) \in F$. 
By Theorem \ref{thm:pisiProp}(ii), this is impossible as $t$ is a \sigmaSE-monomial over $F$,
so $a_n = 0$.
By induction on $n$ it follows that $a_n = 0$ for all $n \ge 1$.
Hence $y = \sum_{i \ge 0} a_i t^{-i} = a_0 \in \Const_\sigma F$, proving that 
$\Const_\sigma{F((t^{-1}))} \subseteq \Const_\sigma F$.
\end{proof}

We will need a \piE-monomial $\xi$ over $k(x)((t^{-1}))$, behaving like the power $t^x$. Informally, $\sigma t^x/t^x =
(\sigma t/t)^x = (1 + \eta t^{-1})^x = \sum_{i=0}^{\infty}\binom{x}{i}\eta^i t^{-i}$, so we introduce $\xi$
in the following way:
\begin{lemma}\label{Lemma:UnchangedConstantXi}
Let $t$ and $x$ be algebraically independent over $k$ where $t$ is a \sigmaSE-monomial over $k$ with $\eta=\sigma(t)-t\in k$ and $\sigma x = x$.
Let $\xi$ be transcendental over $k(x)((t^{-1}))$ with $\sigma \xi = \alpha \xi$ where 
$\alpha = \sum_{i=0}^{\infty}\binom{x}{i}\eta^i t^{-i}$.
Then $\xi$ is a \piE-monomial over $k(x)((t^{-1}))$.
\end{lemma}
\begin{proof}
By Theorem~\ref{thm:pisiProp}(i), it suffices to show that $\alpha$ is not a $\sigma$-radical over $k(x)((t^{-1}))$.
Suppose that it is. Then there are $r, n \in \set Z$, $r>0$, and $u = \sum_{i=n}^{\infty} a_i t^{-i}$ with 
$a_i \in k(x)$ such that $\alpha^r = \sigma u / u$. Since $\alpha$ is a binomial series and by
(\ref{auto}) we obtain
\begin{eqnarray*}
\alpha^r \ &=&\ \left(\sum_{i=0}^{\infty}\binom{x}{i}\eta^i t^{-i}\right)^r\ =\ 
\sum_{i=0}^{\infty}\binom{r x}{i}\eta^i t^{-i}\ =\ 1 + r x \eta t^{-1} + O(t^{-2}), \\
\frac{\sigma u}{u}
\ &=&\ \frac{\sum_{i = n}^\infty \left(\sum_{j=n}^i {i-1 \choose i-j} (-\eta)^{i-j} \,\sigma a_j\right) t^{-i}}
{\sum_{i=n}^{\infty} a_i t^{-i}} \\
\ &=&\ \frac{\sigma a_n + (\sigma a_{n+1} - n\eta \sigma a_n)t^{-1} + O(t^{-2})}
{a_n + a_{n+1}t^{-1} + O(t^{-2})} \\
\ &=&\ \frac{\sigma a_n}{a_n}\left(1 + \left(\sigma\frac{a_{n+1}}{a_n} - n\eta - \frac{a_{n+1}}{a_n}\right)t^{-1}
+ O(t^{-2}) \right).
\end{eqnarray*}
Comparison of coefficients of $t^{-i}$ for $i=0,1$ in
$\alpha^r$ and $\sigma u / u$ shows that $\sigma a_n = a_n$ and $\sigma (a_{n+1}/a_n) - a_{n+1}/a_n = (n + rx)\eta$.
Note that $n + rx \ne 0$ since $x$ is transcendental over $k$, so $\sigma w - w = \eta$ with 
$w = a_{n+1}/(a_n(n + rx)) \in k(x)$. Theorem~\ref{thm:pisiProp}(ii) now implies that $t$ is not a \sigmaSE-monomial
over $k(x)$, a contradiction.
\end{proof}

In particular, we have the following property.

\begin{corollary}\label{Cor:UnchangedConstant}
	Let $t$, $x$ and $\xi$ be as in Lemma \ref{Lemma:UnchangedConstantXi}. Then 
	$\Const_\sigma\, k(x)((t^{-1}))(\xi) = (\Const_\sigma k)(x)$.
\end{corollary}
\begin{proof}
	In order to be able to use Lemma \ref{Lemma:LaurentExt}, we need to show first that $t$ is a \sigmaSE-monomial over
	$k(x)$. By the assumptions of Lemma \ref{Lemma:UnchangedConstantXi}, $t$ is transcendental over $k(x)$. Furthermore,
	\begin{eqnarray*}
		\Const_\sigma\, k(x)(t)\ &=&\ \Const_\sigma\, k(t)(x) \quad {\rm by\ the\ assumptions\ of
			\ Lemma\ \ref{Lemma:UnchangedConstantXi}} \\
		\ &=&\ (\Const_\sigma k(t))(x) \quad {\rm by\ Lemma\ \ref{Lemma:XExt}\ with\ } F=k(t) \\
		\ &=&\ (\Const_\sigma k)(x) \quad {\rm since\ } t \mbox{\ is\ a\ \sigmaSE-monomial\ over\ } k \\
		\ &=&\ \Const_\sigma\, k(x) \quad {\rm by\ Lemma\ \ref{Lemma:XExt}\ with\ } F=k, 
	\end{eqnarray*}
	so $t$ is indeed a \sigmaSE-monomial over $k(x)$. Hence
	\begin{eqnarray*}
		\Const_\sigma\, k(x)((t^{-1}))(\xi)\ &=&\ \Const_\sigma\, k(x)((t^{-1})) \quad {\rm by\ 
			\ Lemma\ \ref{Lemma:UnchangedConstantXi}} \\
		\ &=&\ \Const_\sigma\, k(x) \quad {\rm by\ Lemma\ \ref{Lemma:LaurentExt}\ with\ } F=k(x) \\
		\ &=&\ (\Const_\sigma k)(x) \quad {\rm by\ Lemma\ \ref{Lemma:XExt}\ with\ } F=k,
	\end{eqnarray*}
	proving the claim.
\end{proof}

For the termination of our degree bounding procedure we need Lemma~\ref{Lemma:ShiftRelationInXi} below. There a double summation appears that we prove in Lemma~\ref{Lemma:DoubleSum} first. It is remarkable that in our proof below we will utilize symbolic summation algorithms from~\cite{Schneider:05c}  that utilize as backbone our solvers of parameterized linear difference equations that we are currently investigating.

\begin{lemma}\label{Lemma:DoubleSum}
Let $x$ be transcendental\footnote{As a consequence this identity also holds if one replaces $x$ by any element of $k$ or a ring extension of it.} over $k$. For $l\in\set N$ and $a,b\in k$ we have
\begin{equation}\label{Equ:DoubeSumId}
\sum_{k=0}^l\sum_{r=0}^k(-1)^{k-r}a^{l-r}b^{r}\binom{k-1}{k-r}\binom{x}{r}\binom{x}{l-k}=\binom{x}{l}(a+b)^l.\end{equation} 
\end{lemma}

\newcommand{\ShortVersion}[1]{}

\begin{proof}
Denote the double sum on the left-hand side by $S(l)$,  
the inner sum by
$$F(l,k)=\sum_{r=0}^kf(l,k,r)=\sum_{r=0}^k(-1)^{k-r}a^{l-r}b^{r}\binom{k-1}{k-r}\binom{x}{r}\binom{x}{l-k}$$ 
and its innermost summand by $f(l,k,r)$. In a first step we compute for the inner sum $F(l,k)$ the recurrences
\begin{multline}\label{Equ:InnerSumRec1}
k (a
+b
) (k
-l
) (1
+k
-l
) F(l,k)\\
-(1
+k
-l
) (1
+k
-l
+x
) (2 a
+b
+2 a k
+b k
-b x
) F(l,1+k)\hspace*{1cm}\\
+a (2+k) (1
+k
-l
+x
) (2
+k
-l
+x
) F(l,2+k)
=0
\end{multline}
and
\begin{equation}\label{Equ:InnerSumRec2}
a (k
-l
+x
) F(l,k)
+(-1
+k
-l
) F(1+l,k)
=0.
\end{equation}
This can be accomplished, e.g., with the summation packages \cite{PS:95,Schneider:07a}. Internally, Zeilberger's creative telescoping paradigm~\cite{PWZ} is utilized (compare Remark~\ref{Remark:Summation}) which enables one to verify the correctness of the computed recurrences. E.g., one derives the summand recurrence
\begin{multline}\label{Equ:InnerSummandRec}
\alpha_0(l,k,r)f(l,k,r)+\alpha_1(l,k,r)f(l,k+1,r)+\alpha_2(l,k,r)f(l,k+2,r)\\
=g(l,k,r+1)-g(l,k,r)
\end{multline}
with
\begin{align*}
\alpha_0(l,k,r)&=k^2 (a
+b
),\\
\alpha_1(l,k,r)&=-\frac{k (1
	+k
	-l
	+x
	)(
	2 a
	+b
	+2 a k
	+b k
	-b x
	)}{k
	-l
},\\
\alpha_2(l,k,r)&=\frac{a k (2+k) (1
	+k
	-l
	+x
	) (2
	+k
	-l
	+x
	)}{(k
	-l
	) (1
	+k
	-l
	)},\\
g(l,k,r)&=\frac{(-1)^{1+k+r} a^{1+l-r} b^r k^2 (-1+r) r}{(1
	+k
	-r
	) (2
	+k
	-r
	)} \binom{-1+k}{k	
	-r
} \binom{x}{r} \binom{x}{-k+l}
\end{align*}
which holds for all $r,k,l\in\set N$ with $0\leq r\leq k$ and $k<l$. Thus one can sum~\eqref{Equ:InnerSummandRec} over $r$ from $0$ to $k$. Taking care of missing terms one obtains~\eqref{Equ:InnerSumRec1} together with the information that it holds for all $k,l\in\set N$ with $k<l$. The case $k=l$ can be checked separately. The verification of~\eqref{Equ:InnerSumRec2} is even simpler and is omitted here. \\
Together with the verified recurrences~\eqref{Equ:InnerSumRec1} and~\eqref{Equ:InnerSumRec2} we can utilize the summation package \texttt{Sigma} to compute the recurrence
\begin{equation}\label{Equ:DoubeSumRec}
(a + b) (l - x) S(l) + (1 + l)S(1 + l)=0
\end{equation}
for the double sum $S(l)$.
Internally we apply the algorithms from~\cite{Schneider:05c} (which activates our parameterized linear difference equation solver explored in this section) which provide the summand recurrence
\begin{equation}\label{Equ:OuterSummandRec}
\beta_0(l)F(l,k)+\beta_1(l)F(l+1,k)=G(l,k+1)-G(l,k)
\end{equation}
with 
\begin{equation}\label{Equ:GTelescoper}
\begin{split}
\beta_0(l)&=-\frac{x (a+b)}{a},\quad \beta_0(l)=-\frac{(1+l) x}{a (l-x)},\\
G(l,k)&=g_0(l,k)F(l,k)+g_1(l,k)F(l,1+k) 
\end{split}
\end{equation}
where
\begin{align*}
g_0(l,k)&=\frac{1}{a (-1
	+k
	-l
	) (l
	-x
	)} \big(
-a k^2
-b k^2
+a k^3
+b k^3
+a k l
\\
&\hspace*{2cm}
+b k l-2 a k^2 l
-2 b k^2 l
+a k l^2
+b k l^2
-a k x
-b l x
\\
&\hspace*{2cm}
+b k l x
-b l^2 x
+b x^2
-a k x^2
-b k x^2
+b l x^2
\big),\\
g_1(l,k)&=\frac{1}{(k
	-l
	) (l
	-x
	)} \big(
-k
-2 k^2
-k^3
+l
+3 k l
+2 k^2 l
-l^2
-k l^2\\
&\hspace*{2cm}
-x
-3 k x
-2 k^2 x
+2 l x
+2 k l x
-x^2
-k x^2
\big).
\end{align*}
The correctness of the recurrence can be verified as follows. 
Plug~\eqref{Equ:GTelescoper} into the expression $\beta_0(l)F(l,k)+\beta_1(l)F(l+1,k)-(G(l,k+1)-G(l,k))$ and replace the occurrences of $F(l+1,k)$ and $F(l,k+2)$ by linear combinations of $F(l,k)$ and $F(l,k+1)$ using the recurrences~\eqref{Equ:InnerSumRec1} and~\eqref{Equ:InnerSumRec2}. By simple rational function arithmetic it turns out that the obtained expression collapses to 0. Since the expressions in~\eqref{Equ:GTelescoper} are well defined and the recurrences~\eqref{Equ:InnerSumRec1} and~\eqref{Equ:InnerSumRec2} are valid for all $l,k\in\set N$ with $k<l$, it follows that~\eqref{Equ:GTelescoper} satisfies the summand recurrence~\eqref{Equ:OuterSummandRec} for all $l,k\in\set N$ with $k<l$.
Hence we can sum~\eqref{Equ:OuterSummandRec} over $k$ from $0$ to $l-1$ and taking care of missing terms (in particular the summand at $k=l$) shows that~\eqref{Equ:DoubeSumRec} holds for all $l\in\set N$.\\ 
To complete the proof, one verifies that also 
the right-hand side of~\eqref{Equ:DoubeSumId} is a solution of the recurrence~\eqref{Equ:DoubeSumRec} and that both sides in~\eqref{Equ:DoubeSumId} agree at $l=0$.
\end{proof}

\ShortVersion{
	\begin{proof}
		Denote the double sum on the left-hand side by $S(l)$. Using the summation package~\texttt{Sigma} (executing the methods from~\cite{Schneider:05c}) one can compute the recurrence
		\begin{equation*}
		(a + b) (l - x) S(l) + (1 + l)S(1 + l)=0.
		\end{equation*}
As a byproduct also proof certificates are produced which enable one to verify the correctness of the derived recurrence; for further details see~\cite{ABPSArxiv:20}.		
		To complete the proof, one verifies that also 
		the right-hand side of~\eqref{Equ:DoubeSumId} is a solution of the above recurrence and that both sides in~\eqref{Equ:DoubeSumId} agree at $l=0$.
\end{proof}}

\begin{lemma}\label{Lemma:ShiftRelationInXi}
Let $t$, $x$ and $\xi$ be as in Lemma \ref{Lemma:UnchangedConstantXi} with $\eta=\sigma(t)-t\in k$ and define $\eta_j=\sum_{\nu=0}^{j-1}\sigma^{\nu}(\eta)\in k$.
Then for $i,j\in\set N$ we have
	$$\sigma^j(t^{-i}\xi)=\xi\sum_{l=0}^{\infty}\binom{x-i}{l}\eta_j^lt^{-i-l}.$$
\end{lemma}
\begin{proof}
The case $j=0$ holds trivially with $\eta_0=0$.
For $j\geq1$ we prove
\begin{equation}
\label{sigma_j_xi}
\sigma^j(\xi)=\xi\sum_{l=0}^{\infty}\binom{x}{l}\eta_j^lt^{-l}
\end{equation}
or equivalently
$$\prod_{k=0}^{j-1}\sigma^k(\alpha)=\sum_{l=0}^{\infty}\binom{x}{l}\eta_j^lt^{-l}$$
by induction on $j$. For $j=1$ the statement clearly holds. Now suppose that it holds for $j\geq1$. Then
\begin{align*}
\prod_{k=0}^{j}\sigma^k(\alpha)=&\alpha\,\sigma\Big(\prod_{k=1}^{j-1}\sigma^k(\alpha)\Big)=
\Big(\sum_{l=0}^{\infty}\binom{x}{l}\eta^lt^{-l}\Big)
\sigma\Big(\sum_{l=0}^{\infty}\binom{x}{l}\eta_j^lt^{-l}\Big)\\
=&\Big(\sum_{l=0}^{\infty}\binom{x}{l}\eta^lt^{-l}\Big)
\Big(\sum_{l=0}^{\infty}\Big(\sum_{r=0}^l\binom{l-1}{l-r}(-\eta)^{l-r}\binom{x}{r}\sigma(\eta_j)^r\Big) t^{-l}\Big)\\
=&\sum_{l=0}^{\infty}t^{-l}\sum_{k=0}^l\Big(\sum_{r=0}^k\binom{k-1}{k-r}(-\eta)^{k-r}\binom{x}{r}\sigma(\eta_j)^r\Big)\binom{x}{l-k}\eta^{l-k}\\
=&\sum_{l=0}^{\infty}\binom{x}{l}(\eta+\sigma(\eta_j))^lt^{-l}=\sum_{l=0}^{\infty}\binom{x}{l}\eta_{j+1}^lt^{-l}
\end{align*}
which proves~\eqref{sigma_j_xi}; in the second last line we carried out the Cauchy product and in the last line we used Lemma~\ref{Lemma:DoubleSum} with $a=\eta$ and $b=\sigma(\eta_j)$. 
Finally, we have by (\ref{sigma_j_xi})
\begin{eqnarray*}
\frac{\sigma^j(t^{-i} \xi)}{\xi} &=& \sigma^j(t^{-i}) \frac{\sigma^j(\xi)}{\xi} \ =\ (t+\eta_j)^{-i} \sum_{\ell=0}^\infty \binom{x}{\ell} \eta_j^\ell t^{-\ell}\\
&=& \left(1+\frac{\eta_j}{t}\right)^{-i} \sum_{\ell=0}^\infty \binom{x}{\ell} \eta_j^\ell t^{-\ell-i}\ =\ \sum_{u=0}^\infty \sum_{\ell=0}^\infty \binom{-i}{u}\binom{x}{\ell} \left(\frac{\eta_j}{t}\right)^{\!u+\ell} \frac{1}{t^{i}}\\
&=& \sum_{m=0}^\infty \left(\sum_{\ell=0}^\infty \binom{-i}{m-\ell}\binom{x}{\ell}\right) \left(\frac{\eta_j}{t}\right)^{\!m}\frac{1}{t^{i}}
\ =\ \sum_{m=0}^\infty \binom{x-i}{m} \eta_j^m t^{-i-m}
\end{eqnarray*}
where, in the last line, we introduced $m=u+\ell$, replacing $u$ by $m-\ell$, and used Chu-Vandermonde's convolution.
\end{proof}


Before we can start to present our algorithm to bound the degree of polynomial solutions of \sigmaSE-monomials, we need the following result from linear algebra.

\begin{lemma}\label{Lemma:NullSpace}
Let $A(x)\in F[x]^{r\times s}$ be a matrix of polynomials over a field $F$. Then we can compute
a $F[x]$-basis $v_1(x),\dots,v_{\nu}(x)\in F[x]^s$ of $\Ker A(x)$ together with $\delta\in\set N$
such that for any $n\in\set N$ with $n\geq\delta$, the vectors $v_1(n),\dots,v_{\nu}(n) \in F^s$ form a basis of $\Ker A(n)$.
\end{lemma}
\begin{proof}
By row and column operations in the Euclidean domain $F[x]$, compute the Smith normal form of $A(x)$, \ie $\mu \in
\set N$, a diagonal matrix $B(x) = \text{diag }(d_1(x),\dots,d_{\mu}(x),0,\dots,0) \in F[x]^{r\times s}$ with
$d_i(x)\in F[x]\setminus\{0\}$ for $1 \le i \le \mu$ and $d_i\mid d_{i+1}$ for $1\leq i<\mu$, and invertible matrices 
$P(x)\in F[x]^{r\times r}$ and $Q(x)\in F[x]^{s\times s}$ such that $A(x) = P(x)B(x)Q(x)$.
Then the last $\nu:=s-\mu$ columns of $Q^{-1}(x)$, say $v_1(x),\dots,v_{\nu}(x)$, form a $F[x]$-basis of\/ $\Ker A(x)$.
During this elimination collect all those polynomials $p$ used for row multiplications $r_i\to p\,r_i$, 
resp.\ column multiplications $c_j\to p\,c_j$, and compute the set of their non-negative integer roots. Let $\delta$
be a non-negative integer, larger than any of these roots and any of the non-negative integer roots of $d_i(x)$
for $1\leq i\leq\mu$. Then for each $n\in\set N$ with $n\geq\delta$, the operations employed in the computation of the 
Smith normal form remain well defined even after substituting $n$ for $x$, and the rank of $B(n)$ equals that of
$B(x)$. Hence $v_1(n),\dots,v_{\nu}(n)$ is a basis of\/ $\Ker A(n)$ for all $n\geq\delta$.
\end{proof}


The following result will be needed in the base case of our degree bounding algorithm.

\begin{lemma}\label{Lemma:LCConstraint}
Let $(k,\sigma)$ be a difference field and let $t$ be a \sigmaSE-monomial over $k$.
For any $L\in k[t][E;\sigma]$, $w\in k$ and $d\in\set N$ we have
\begin{equation}\label{Equ:LCConstraint}
L(w\,t^d)=L_{\alpha}(w)t^{\alpha+d}+g
\end{equation}
where $L_{\alpha}$ with $\alpha=\deg L $ is given by~\eqref{eq:project}, and $g\in k[t]$ with $\deg g<\alpha+d$.
\end{lemma}
\begin{proof}
Let $\eta=\sigma(t)-t\in k$. It is easy to show by induction on $j$ that for all $j\in\set N$,
\begin{equation}\label{Equ:SigmaPowerRelation}
\sigma^j(t^d)=\Big(t+\sum_{i=0}^{j-1}\sigma^i(\eta)\Big)^d=t^d+g_j
\end{equation}
with $g_j\in k[t]$ and $\deg g_j<d$. By~\eqref{eq:project},
$$L(w\,t^d)=\sum_{j=\nu(L)}^{\alpha}t^jL_j(w\,t^d)$$
where $L_j\in k[E;\sigma]$ for all $j$. Since
$\deg t^jL_j(w\,t^d)\leq j+d$ by~\eqref{Equ:SigmaPowerRelation}, it follows that
\begin{equation}\label{Equ:SemiExpandL}
L(w\,t^d)=t^{\alpha}L_{\alpha}(w\,t^d)+h
\end{equation}
for some $h\in k[t]$ with $\deg h<\alpha+d$. By~\eqref{Equ:SigmaPowerRelation} again,
$$L_{\alpha}(w\,t^d)=\sum_{j=0}^{\mu}a_j\sigma^j(w\,t^d)=\sum_{j=0}^{\mu}a_j\sigma^j(w)(t^d+g_j)=L_{\alpha}(w)t^d+\bar{h}$$
for some $\mu\in\set N$, $a_j\in k$ and $g_j,\bar{h}\in k[t]$ with $\deg g_j,\deg\bar{h} <d$, so~\eqref{Equ:SemiExpandL} implies~\eqref{Equ:LCConstraint}.
\end{proof}

We are now ready to present our degree bounding algorithm.
Let $t$ be a \sigmaSE-monomial over $k$ and let $L\in k[t][E;\sigma]$ with $\alpha:=\deg L \geq0$ and $\Phi_1,\dots,\Phi_m \in k[t]$. Then we are interested in finding a degree bound for the parameterized difference equation
\begin{equation}\label{Equ:OrgEquation}
L(y)=c_1 \Phi_1+\dots+c_m\Phi_m,
\end{equation}
i.e., we seek $b\in\set N$ such that for any solution $y\in k[t]$ with $c_i\in K:=\Const_\sigma(k)$ we have that
$\deg(y)\leq b$. Subsequently, we define $\beta:=\max_{1\leq i\leq m}\deg(\Phi_i)$.
The underlying algorithm proceeds stepwise for $n=0,1,\dots$. At a particular step $n$, it either finds a degree bound
and we are done, or it produces a $b\, (=b_n)\in\set N$ with $b\geq\max(\beta-\alpha,0)+n$ and $r\,(=r_n)\geq1$ vectors 
$f_u\,(=f^{(n)}_u)=(f_{u,0},\dots,f_{u,n})\in k\times k[x]^{n}$ (for $1\leq u\leq r$)  with entries from a polynomial ring $k[x]$ with the following properties.

\medskip

\noindent(i) Degree property: for any $d\geq b$ and any $1\leq u\leq r$ we have
\begin{equation}\label{Equ:LCancelation}
\deg(L\left(\sum_{j=0}^nt^{d-j}f_{u,j}(d)\right))<d+\alpha-n.
\end{equation}
(ii) Completeness: for any $h\in k[t]$ with $d:=\deg(h)\geq b$ and $\deg(L(h))< d+\alpha-n$ there are $\kappa_{u}\in K$
with $(1\leq u \leq r)$ such that
$$\deg\big(h-\sum_{u=1}^r\kappa_u\sum_{j=0}^nt^{d-j}f_{u,j}(d)\big)<d-n.$$
(iii) Maximal degree: we have $0\neq(f_{1,0},\dots,f_{r,0})\in k^r$.

\medskip

After the step $n$ we can take a vector $h=(h_0,\dots,h_{n})\in k\times k[x]^{n}$ (actually an appropriate $K$-linear combination of $f_1,\dots,f_u$) with $h_{0}\in k^*$ such that for  
$$y=h_0\,t^{b}+h_1(b)t^{b-1}+\dots+h_{n}(b)t^{b-n}\in k[t]$$
we have
$$\deg(L(y))<b+\alpha-n.$$
This does not mean that $y$ with $\deg(y)=b$ can be prolonged (by adding lower terms of degree $<b-n$) to a solution~\eqref{Equ:OrgEquation} for some $c_1,\dots,c_m\in K$.
But it is a necessary step towards such a solution. In particular,
if one fails at a certain moment to derive such a vector $h=(h_0,h_1(x),\dots,h_n(x))\in k\times k[x]^{n}$ with $h_0\neq0$ (i.e., there do not exist vectors $f_1,\dots,f_u$ as claimed above whose $K$-linear combination produces $h$) then this shows that the degree of any solution of $y\in k[t]$ for~\eqref{Equ:OrgEquation} is bounded by $b$.
As we will see later, the step $n\to n+1$ can be carried out by solving a particular parameterized linear difference equation in $k$ (and using the already found vectors $f_u^{(n-1)}\in k[x]^n$ with $1\leq u\leq r_{n-1}$ of the previous step $n-1$).
More precisely, our degree bounding strategy can be sketched as follows:

\medskip

\noindent\sigmaSE-\verb|DegreeBound|$(L; \Phi_1,\dots,\Phi_m; k[t])$\\
\verb|input:  |\begin{minipage}[t]{10.8cm} a \sigmaSE-monomial $t$ over $k$, $L\in k[t][E;\sigma]$ with $\alpha:=\deg L \geq0$ and\\ 
	$\Phi_1,\dots,\Phi_m \in k[t]$ with $\beta:=\max_{1\leq i\leq m}\deg(\Phi_i)$;\end{minipage}\\
\verb|output: |a degree bound $b$ of the
solutions in $k[t]$ of~\eqref{Equ:OrgEquation};\\[0.2cm]
\verb| 1. |$b^{(-1)}:=\max(\beta-\alpha,0)$;\\
\verb| 2. for| $n=0,1,2, \dots$ \verb|do|\\
\verb| 3.   |\begin{minipage}[t]{10.8cm}solve a particular instance of a parameterized linear difference equation in $k$ and decide constructively if there are a $b\in\set N$ larger than $b^{(n-1)}$ with $b\geq \max(\beta-\alpha,0)+n$ and vectors $f_1^{(n)},\dots,f_{u}^{(n)}\in k[x]^{n+1}$ such that properties (i), (ii) and (iii) hold;\end{minipage}\\[2pt]
\verb| 4.   if| this is not possible \verb|then|\\ 
\verb| 5.       |\begin{minipage}[t]{9cm} extract a degree bound $b$ of the
	solutions in $k[t]$ of~\eqref{Equ:OrgEquation}\\ and return $b$;\end{minipage}\\ 
\verb| 6.   | $b^{(n)}:=b$.\\

\medskip

We will now supplement the missing details to obtain a complete algorithm to derive the desired degree bound $b$. More precisely, we will proceed as follows:
\begin{enumerate}
	\item The base case $n=0$ is considered in Lemma~\ref{Lemma:DegBoundn=0}.
	\item The induction step (iteration) $n\to n+1$ with $n\geq0$ is elaborated in Lemma~\ref{Lemma:DegBound:n->n+1}.
	\item Finally, the exit of the for-loop is proven in Theorem~\ref{th:SigmaBounds} which completes the proof that a degree bound for \sigmaSE-monomials can be determined.
\end{enumerate}

\begin{lemma}\label{Lemma:DegBoundn=0}
Let $t$ be a \sigmaSE-monomial over $k$ with $K=\Const_\sigma(k)$.
Let $L\in k[t][E;\sigma]$ with $\alpha:=\deg L \geq0$ and $\Phi_1,\dots,\Phi_m \in k[t]$ with $\beta:=\max_{1\leq i\leq m}\deg(\Phi_i)$ and set $b=\max(0,\beta-\alpha)$. If we can solve parameterized linear difference equations with coefficients in $k$, then we can decide if $b$ is a degree bound of the 
solutions in $k[t]$ of~\eqref{Equ:OrgEquation} or one can compute $r_0\geq1$ vectors $f^{(0)}_1,\dots,f^{(0)}_{r_0}\in k^1$ such that properties (i), (ii) and (iii) hold. 
\end{lemma}
\begin{proof}
Making the ansatz $y=w\,t^{\delta}+g$ for unknown $\delta\in\set N$ with $\delta> b$, $w\in k^*$ and $g\in k[t]$ with $\deg(g)<\delta$ such that $\deg(L(y))<\alpha+\delta$ holds,
gives by Lemma~\ref{Lemma:LCConstraint} the constraint
\begin{equation}\label{Equ:HomogeneousConstraint}
L_{\alpha}(w)=0
\end{equation}
where $L_{\alpha}$ with $\alpha=\deg L $ is given by~\eqref{eq:project}.
By assumption we can solve this difference equation in $k$. If there is only the trivial solution $w=0$, it follows that $b$ is a degree bound of~\eqref{Equ:OrgEquation}, and we return $b$.\\
Otherwise, let $f_{1,0},\dots,f_{r_0,0}$ be a basis of the $K$-vector space of solutions of~\eqref{Equ:HomogeneousConstraint} in $k$. Then the $r_0$ vectors $f^{(0)}_1=(f_{1,0}),\dots,f^{(0)}_{r_0}=(f_{r_0,0})\in k^1$ satisfy the above properties~(i), (ii) and~(iii). To link to this setting, we simply have to set $r:=r_0$ and $f_u:=f^{(0)}_u$ for $u=1,\dots,r_0$; in addition observe that the vectors $f_u$ with entries from $k$ are treated as entries from $k[x]$ and thus $f_{u,0}(d)$ (i.e., $x$ is replaced by $d$) equals $f_{u,0}$. 
To prove (i), let $u\in\{1,2,\dots,r\}$. By Lemma~\ref{Lemma:LCConstraint},
$$L(t^df_{u,0})=L_{\alpha}(f_{u,0})t^{\alpha+d}+g$$
where $g\in k$ with $\deg g<\alpha+d$. But $L_{\alpha}(f_{u,0})=0$, so $\deg L(t^df_{u,0})=\deg g<\alpha+d$, proving (i). To prove (ii), let $h\in k[t]$ with $d=\deg(h)\geq b$ and $\deg L(h)<d+\alpha$. Write $h=w\,t^d+g$ where $w\in k$ and $g\in k[t]$ with $\deg(g)<d$. Then
$$L(h)=L(w\,t^d)+L(g)$$
where $\deg L(g)<\alpha+d$, so $\deg L(h)<d+\alpha$ implies
$$\deg L(w\,t^d)=\deg(L(h)-L(g))<d+\alpha.$$
By Lemma~\ref{Lemma:LCConstraint}, $L(w\,t^d)=L_{\alpha}(w)\,t^{\alpha+d}+g_1$ for some $g_1\in k[t]$ with $\deg g_1<\alpha+d$, hence $L_{\alpha}(w)=0$. By the definition of $f_{1,0},\dots,f_{r,0}$, there are $\kappa_1,\dots,\kappa_r\in K$ such that $w=\sum_{l=1}^{r}\kappa_l\,f_{l,0}$, therefore
$$\deg(h-\sum_{l=1}^r\kappa_l\,t^df_{l,0})=\deg(h-w\,t^d)=\deg g<d$$
proving (ii). Finally, (iii) holds by the definition of $f_{1,0},\dots,f_{r,0}$. This proves the lemma.
\end{proof}

\begin{lemma}\label{Lemma:DegBound:n->n+1}
Let $t$ be a \sigmaSE-monomial over $k$ with $K=\Const_\sigma(k)$.
Let $L\in k[t][E;\sigma]$ with $\alpha:=\deg L \geq0$ and $\Phi_1,\dots,\Phi_m \in k[t]$ with $\beta:=\max_{1\leq i\leq m}\deg(\Phi_i)$.
Suppose that we are given $b=b^{(n)}\in\set N$ with $b\geq\max(\beta-\alpha,0)+n$ and vectors $f_u^{(n)}=(f_{u,0},\dots,f_{u,n})\in k\times k[x]^{n}$ for $1\leq u\leq r_n$ with the properties~(i), (ii) and (iii) as stated above.\\
If we can solve parameterized linear difference equations with coefficients in $k$, then we can either compute a degree bound of the 
solutions in $k[t]$ of~\eqref{Equ:OrgEquation} or one can compute $b^{(n+1)}\in\set N$ and $r_{n+1}\geq1$ vectors $f_u^{(n+1)}\in k\times k[x]^{n+1}$ with $1\leq u\leq r_{n+1}$ such that the properties (i), (ii) and (iii) hold with the replacements $n\to n+1$, $b\to b^{(n+1)}$, $r\to r_{n+1}$, and $f_1,\dots,f_{r} \to f^{(n+1)}_1,\dots,f^{(n+1)}_{r_{n+1}}$.
\end{lemma}
\begin{proof}
Suppose that there exists
a solution $y\in k[t]$ with $\delta:=\deg(y)>b$ for~\eqref{Equ:OrgEquation} for some $c_1,\dots,c_m\in K$. 
In the following we will either show that this is not possible and thus the given $b$ is a degree bound, or we will determine a new $b'(=b^{(n+1)})$ which is larger than $b=(b^{(n)})$ as a new candidate of a possible degree bound.
Since $\delta\geq\max(\beta-\alpha,0)+n+1$ it follows that 
\begin{equation}\label{Equ:FullSolBound}
\deg L(y)<\delta+\alpha-n-1.
\end{equation}
In particular, $\deg L(y)<\delta+\alpha-n$ holds and we can apply 
property (ii) of our assumption with $h=y$. As a consequence there exist $\kappa_1,\dots,\kappa_r\in K$ such that for
$$\tilde{y}:=\sum_{u=1}^r \kappa_u\sum_{i=0}^nf_{u,i}(\delta)t^{\delta-i}$$
we get $\deg(y-\tilde{y})<\delta-n$. 
Hence the first $n$ leading coefficients of $y$ are determined by $\tilde{y}$ and we get
$$y=\tilde{y}+w\,t^{\delta-n-1}+g$$ 
for some $w\in\ k$ and $g\in k[t]$ with $\deg(g)<\delta-n-1$. From $\deg L(g)<\alpha+\delta-n-1$ and~\eqref{Equ:FullSolBound} it follows for
$\bar{y}=\tilde{y}+w\,t^{\delta-n-1}$
that $$\deg L(\bar{y})=\deg(L(y)-L(g))\leq\max(\deg L(y),\deg L(g))<\delta+\alpha-n-1.$$ 
Summarizing,
if there exists
a solution $y\in k[t]$ with $\delta:=\deg(y)>b$ for~\eqref{Equ:OrgEquation} for some $c_1,\dots,c_m\in K$, then 
there are a $w\in k$ and $\kappa_1,\dots,\kappa_r\in K$ such that for
\begin{equation}\label{Equ:IncrementalAnsatz}
\bar{y}=\tilde{y}+w\,t^{\delta-n-1}
\end{equation}
with
\begin{equation}\label{Equ:yTilde}
\tilde{y}=\sum_{u=1}^r \kappa_u\sum_{i=0}^nf_{u,i}(\delta)t^{\delta-i}
\end{equation}
we have
\begin{equation}\label{Equ:CheckCriterion}
\deg(\bar{y})=\delta\quad\text{and}\quad\deg(L(\bar{y}))< \delta+\alpha-n-1.
\end{equation}
Note: if we prove that there is no such $\bar{y}$ with $\deg\bar{y}>b'$ for some $b'\in\set N$ with $b'\geq b$ then it follows that $b'$ is a degree bound for~\eqref{Equ:OrgEquation}. 
In a nutshell, we try to construct all $\bar{y}$ with this property and hope that this construction eventually fails. 
By properties (i) and (ii) we have
\begin{equation}\label{Equ:UpperCancelations}
\deg(L(\tilde{y}))<\delta+\alpha-n
\end{equation}
for any $\kappa_1,\dots,\kappa_r\in K$ and we have to check if there is a $w\in k$ such that also the coefficient of $t^{\delta+\alpha-n-1}$ in $L(\bar{y})$ vanishes. This gives the constraint
\begin{equation}\label{Equ:CoeffConstraint0}
\big[\text{coefficient of }t^{\delta+\alpha-n-1}\text{ in }L(\tilde{y})\big]+L_{\alpha}(w)=0.
\end{equation}
Define $\eta=\sigma(t)-t\in\ k$. By $\sigma^j(t)=t+\eta_j$ with $\eta_j=\sum_{\nu=0}^{j-1}\sigma^{\nu}(\eta)\in k$ and
using the binomial theorem in its form
\begin{equation}\label{Equ:BinomThm}
\sigma^j(t^{\delta-i})=(t+\eta_j)^{\delta-i}=\sum_{l=0}^{\delta-i}\frac{\fallingFac{(\delta-i)}{l}}{l!}\eta_j^{l}t^{\delta-i-l}
\end{equation} 
with $i\in\set N$ where $0\leq i\leq\delta$ and $j\geq1$,
the constraint~\eqref{Equ:CoeffConstraint0} can be restated as
\begin{equation}\label{Equ:CoeffConstraint}
L_{\alpha}(w)=\sum_{j=1}^{r}\sum_{i=0}^{m}\kappa_{j}h_{j,i}\delta^i
\end{equation}
for some $m\in\set N$ with $m\geq n+1$ and explicitly given $h_{j,i}\in k$. Based on this information we consider the parameterized linear difference equation
\begin{equation*}
L_{\alpha}(w)=\sum_{j=1}^{r}\sum_{i=0}^{m}e_{j,i}h_{j,i}
\end{equation*}
with the unknown $w\in k$ and the $u=(m+1)r$ unknown parameters $e_{j,i}\in K$. By assumption we can compute a basis
$\{(w_i,d_{i,1},\dots,d_{i,u})\}_{1\leq i\leq s}$
for the corresponding solution space. 
W.l.o.g.\ suppose that the first $r$ columns of $D:=(d_{i,j})_{1\leq i\leq s,1\leq j\leq u}$ correspond to the right-hand sides
$$h_{1,0},h_{2,0},\dots,h_{r,0}$$
ordered in this way. Note that in~\eqref{Equ:CoeffConstraint}, the corresponding parameters $e_{j,0} = \delta^0\, \kappa_j$ for $1 \le j \le r$ are free of $\delta$. For later use we define the submatrix
$$\bar{D}:=(d_{i,j})_{1\leq i\leq s,1\leq j\leq r}$$
of $D$.
In order to obtain all the possible $\kappa_{j}$ for our original ansatz~\eqref{Equ:CoeffConstraint} we have to impose additional relations on the $e_{j,i}$: 
$$\kappa_j=e_{j,0}=\frac{e_{j,1}}{\delta}=\frac{e_{j,2}}{\delta^2}=\dots=\frac{e_{j,m}}{\delta^m}$$
for $1\leq j\leq r$.  
Namely, for each $j,i$ let  $C_{j,i}$ be the corresponding column in $D$. Then for these coefficients we have to consider all $(\lambda_1,\dots,\lambda_s)\in K^s$ such that for $1\leq j\leq r$ and $0\leq i\leq n$ we have
\begin{equation}\label{Equ:CoeffConstraint2}
\begin{split}
(\lambda_1,\dots,\lambda_s)C_{j,0}&=(\lambda_1,\dots,\lambda_s)\frac{C_{j,1}}{\delta}=\dots=
(\lambda_1,\dots,\lambda_s)\frac{C_{j,m}}{\delta^m}.
\end{split}
\end{equation}
Putting all these equations together, replacing $\delta$ by $x$ and clearing denominators leads to a matrix
$A(x)$ with $s$ columns and entries from $K[x]$ with the property that for any $(\lambda_1,\dots,\lambda_s)\in K^s$ we have that
\begin{equation}\label{Equ:MatrixEncoded}
A(\delta)(\lambda_1,\dots,\lambda_s)^T=0
\end{equation} 
iff~\eqref{Equ:CoeffConstraint2} holds for all $1\leq j\leq r$ and $0\leq i\leq n$. By construction this is equivalent to the claim that
$(\kappa_1,\dots,\kappa_r)\in K^r$ and $w\in k$ with 
\begin{equation}\label{Equ:NewRelationForF}
\begin{split}
(\kappa_1,\dots,\kappa_r)&=(\lambda_1,\dots,\lambda_s)\bar{D},\\
w&=(\lambda_1,\dots,\lambda_s)(w_1,\dots,w_s)^T
\end{split}
\end{equation}
are a solution of~\eqref{Equ:CoeffConstraint}.
In short, for any $(\lambda_1,\dots,\lambda_s)\in K^s$ we have that ~\eqref{Equ:MatrixEncoded} holds iff $(\kappa_1,\dots,\kappa_r)$ and $w$ with~\eqref{Equ:NewRelationForF} are a solution of~\eqref{Equ:CoeffConstraint}. In particular, by~\eqref{Equ:UpperCancelations} this is equivalent to the desired property that~\eqref{Equ:CheckCriterion} holds with~\eqref{Equ:IncrementalAnsatz}. \\
By Lemma~\ref{Lemma:NullSpace} we can compute a $K[x]$-basis $v_1(x),\dots,v_{\rho}(x)\in K[x]^{s}$ of $\Ker A(x)$ and
$\tilde{b}\in\set N$ with the following property: For any $\delta\in\set N$ with $\delta\geq \tilde{b}$, the vectors $v_1(\delta),\dots,v_{\rho}(\delta)$ form a basis of $\Ker A(\delta)$, i.e., of the solution space of~\eqref{Equ:MatrixEncoded}. If the basis is empty, i.e., $\rho=0$, or 
if the first $r$ entries in each of the basis elements $v_1(x),\dots, v_{\rho}(x)$ are zero (recall that the linear combination of the first $r$ entries gives all combinations for the highest coefficient in our ansatz~\eqref{Equ:IncrementalAnsatz}), it follows that the ansatz~\eqref{Equ:IncrementalAnsatz} with $\delta>\tilde{b}$ and~\eqref{Equ:CheckCriterion} is not possible. Hence $b'=\max(\tilde{b},b)$ is a degree bound for~\eqref{Equ:OrgEquation} and we are done. Otherwise, define 
$$b'(=b^{(n+1)}):=\max(\tilde{b},b+1)$$
with $b'\geq\max(\beta-\alpha,0)+n+1$. Furthermore, for $1\leq u\leq \rho$ define
\begin{align}
(f'_{u,0},\dots,f'_{u,n})&=(v_u\,\bar{D})(f_1,\dots,f_r)^T\in k\times k[x]^{n},\label{Equ:NewFByOldf}\\
w'_u&=v_u\,(w_1,\dots,w_s)^T\in k[x]\nonumber
\end{align}
and take $f'_u=(f'_{u,0},\dots,f'_{u,n},w'_u)\in k\times k[x]^{n+1}$ for $1\leq u\leq\rho$. By construction (compare~\eqref{Equ:NewRelationForF}) property (i) holds from above where $n$ is replaced by $n+1$, $r$ is replaced by $\rho$, the  $f_1,\dots,f_{r}$ are replaced by $f'_1,\dots,f'_{\rho}$ and $b$ is replaced by $b'$.\\ 
In addition, property (ii) holds:
Let $h\in k[t]$ with $d'=\deg(h)\geq b'$ and $\deg(L(h))< d'+\alpha-n-1$. By construction, $h=\sum_{u=1}^r \kappa_u\sum_{i=0}^nf_{u,i}(d')t^{d'-i}+w\,t^{d'-n-1}+w'$ where $w'\in k[t]$ with $\deg(w')<d'-n-1$ and where
$(\kappa_1,\dots,\kappa_r)\in K^r$ and $w\in k$ are given by~\eqref{Equ:NewRelationForF} for some $\lambda=(\lambda_1,\dots,\lambda_s)\in K^s$. In particular, we have that $\lambda\in\Ker A(d')$. Thus with~\eqref{Equ:NewRelationForF} we get 
\begin{align*}
(\kappa_1,\dots,\kappa_r)&=(\kappa'_1,\dots,\kappa'_{\rho})(v_1(d'),\dots,v_{\rho}(d'))^T\,D,\\  
w&=(\kappa'_1,\dots,\kappa'_{\rho})(v_1(d'),\dots,v_{\rho}(d'))^T(w_1,\dots,w_s)^T
\end{align*}
for some $(\kappa'_1,\dots,\kappa'_{\rho})\in K^{\rho}$. With~\eqref{Equ:NewFByOldf} we get  
$h=\sum_{u=1}^{\rho}\kappa'_u\sum_{j=0}^nt^{d'-j}f'_{u,j}(d')\big)+w'$. Since $\deg(w')<d'-n-1$, property (ii) holds for $n\to n+1$.\\
If $(f'_{1,0},f'_{2,0},\dots,f'_{\rho,0})=0$, it follows by property~(ii) that there does not exist a $\bar{y}$ of the form~\eqref{Equ:IncrementalAnsatz} with~\eqref{Equ:CheckCriterion} 
where $\delta\geq b'$. Consequently, $b'-1$ is a degree bound for~\eqref{Equ:OrgEquation} and we are done. Otherwise, also property~(iii) holds for the vectors $f^{(n+1)}_1:=f'_1,\dots,f^{(n+1)}_{r_{n+1}}:=f'_{r_{n+1}}$ with $r_{n+1}=\rho$ and $b^{(n+1)}=b'$. This completes the proof of the lemma.
\end{proof}

Finally, we can assemble Lemmas~\ref{Lemma:DegBoundn=0} and Lemma~\ref{Lemma:DegBound:n->n+1} to obtain in Theorem~\ref{th:SigmaBounds} an algorithm that enables one to calculate a degree bound for
solutions in $k[t]$ of parameterized linear difference equations with coefficients in $k[t]$. 
Here we will use in addition the following notation. For a set $V\subseteq k\times k[x]^{r}$ and $n\in\set N$ with $n\leq r$ we define
\begin{equation}\label{Equ:DefPn}
P_n(V)=\{(v_0,v_1,v_2,\dots,v_n): (v_0,v_1,\dots,v_n,\dots,v_r)\in V\}.
\end{equation}
Note: if $V$ is a $K$-subspace of $k\times k[x]^{r}$, then $P_n(V)$ is a $K$-subspace of $k\times k[x]^{n}$.

\begin{theorem}\label{th:SigmaBounds}
Let $(k,\sigma)$ be a difference field and let $t$ be a \sigmaSE-monomial over $k$.
If we can solve parameterized linear difference equations with coefficients in $k$, then we can bound the degree of
solutions in $k[t]$ of parameterized linear difference equations with coefficients in $k[t]$.
\end{theorem}

\begin{proof}
Let $L\in k[t][E;\sigma]$ with $\alpha:=\deg L \geq0$ and $\Phi_1,\dots,\Phi_m \in k[t]$ with $\beta:=\max_{1\leq i\leq m}\deg(\Phi_i)$. 
Then by Lemmas~\ref{Lemma:DegBoundn=0} and~\ref{Lemma:DegBound:n->n+1}, i.e., the execution of \sigmaSE-\verb|DegreeBound|$(L; \Phi_1,\dots,\Phi_m; k[t])$, one obtains a degree bound when the for-loop is quit. 
Suppose that this iteration process (the loop over $n=0,1,2,\dots$) does not terminate. For $n\geq 0$ let $f_1^{(n)},\dots,f_{r_n}^{(n)}\in k\times k[x]^{n}$ be the corresponding vectors in the $n$th iteration and consider the $K$-subspace $V_n=\langle f_1^{(n)},\dots,f_{r_n}^{(n)}\rangle_{K}$ of $k\times k[x]^{n}$. Since in each step the new vectors are linear combinations of the previous vectors (see~\eqref{Equ:NewFByOldf}), it follows that $P_{n}(V_{n+l+1})$ is a subspace of $P_n(V_{n+l})$ for any $n,l\in\set N$; recall the definition of $P_n$ in~\eqref{Equ:DefPn}. In particular, for any $n\in\set N$ we get the chain of subspaces
\begin{equation}\label{Equ:VSChain}
P_n(V_n)\supseteq P_n(V_{n+1})\supseteq P_n(V_{n+2})\supseteq\dots\supset\{0\};
\end{equation}
note that each of these subspaces contains a nonzero vector due to property~(iii); more precisely, the first entry is nonzero. 
We show that there is a sequence $g=(g_0,g_1,g_2,\dots)\in k\times k[x]^{\set N}$ with $g_0\neq0$ such that for all $n,l\in\set N$ we have
\begin{equation}\label{Equ:StableChain}
(g_0,g_1,\dots,g_n)\in P_{n}(V_{n+l}).
\end{equation}
Suppose that we obtained already $g=(g_0,g_1,\dots,g_{n-1})\in k\times k[x]^{n-1}$ for $n\geq 1$ with $g_0\neq0$ such that
\begin{equation}\label{Equ:StableChainBelow}
(g_0,g_1,\dots,g_{n-1})\in P_{n-1}(V_{n-1+l})
\end{equation}
 holds for all $l\geq0$. Let $\rho=\min_{i\geq0}\dim P_n(V_{n+i})$. Then $\rho\geq1$ and there is an $s\in\set N$ with $s\geq n$ and $\dim P_n(V_{s})=\rho$. By~\eqref{Equ:VSChain} we conclude that $P_n(V_s)= P_n(V_{s+i})$ holds for all $i\geq0$. Thus we can take $g_n\in k[x]$ with $(g_0,\dots,g_{n-1},g_n)\in P_n(V_{s+l})$ for all $l\geq0$ which proves~\eqref{Equ:StableChain} for all $l\geq0$. Setting $l=0$, we get $g=(g_0,g_1,g_2,\dots)\in k[x]^{\set N}$ with $(g_0,g_1,\dots,g_n)\in V_n$ for all $n\geq0$.\\ 
For $u\in\set N$ define
$$Y_{u}=\xi t^{-u}\sum_{i=0}^{\infty}g_it^{-i}$$
in the difference field $(k(x)((t^{-1}))(\xi),\sigma)$ introduced in Lemma~\ref{Lemma:UnchangedConstantXi}. 
We show that
\begin{equation}\label{Equ:LYSol}
L(Y_{u})=0
\end{equation}
holds for any $u\in\set N$. Suppose otherwise that $L(Y_{\ell})\neq0$ 
for some $\ell\geq0$. 
Then there is an $n$ such that
the coefficient of $L(Y_{\ell})$ at $\xi t^{\alpha-\ell-n}$ is nonzero, but all the coefficients of $L(Y_{\ell})$ at $\xi t^{\alpha-\ell-\nu}$ with $0\leq\nu<n$ are zero. Looking at the coefficient at $\xi t^{\alpha-\ell-n}$ gives a polynomial $p(x)\in k[x]\setminus\{0\}$. With Lemma~\ref{Lemma:ShiftRelationInXi} it follows that
$$\sigma^j(\xi\,t^{-i})|_{x\mapsto\delta,\xi\mapsto t^{\delta}}=\xi\sum_{l=0}^{\infty}\binom{x-i}{l}\eta_j^lt^{-i-l}|_{x\mapsto\delta,\xi\mapsto t^{\delta}}=\sum_{l=0}^{\delta-i}\frac{\fallingFac{(\delta-i)}{l}}{l!}\eta_j^lt^{\delta-i-l}$$
holds for all $i,j,\delta\in\set N$; in particular, with~\eqref{Equ:BinomThm} we get
$$\sigma^j(\xi\,t^{-i})|_{x\mapsto\delta,\xi\mapsto t^{\delta}}=\sigma^j(t^{\delta-i})$$
for all $i,j\in\set N$ with $0\leq i\leq\delta$
where the automorphism $\sigma$ on the left-hand and right-hand side is taken from $(k(x)((t^{-1}))(\xi),\sigma)$ and
$(k(t),\sigma)$, respectively. This implies that
\begin{equation}\label{Equ:p(delta)}
\begin{split}
p(\delta)&=\Big(\text{coefficient of }\xi t^{\alpha-\ell-n}\text{ in }L(\xi t^{-\ell}\sum_{i=0}^{\infty}g_it^{-i})\Big)|_{x\mapsto\delta,\xi\mapsto t^{\delta}}\\
&=\Big(\text{coefficient of }\xi t^{\alpha-\ell-n}\text{ in }L(\xi t^{-\ell}\sum_{i=0}^{n}g_it^{-i})\Big)|_{x\mapsto\delta,\xi\mapsto t^{\delta}}\\
&=\text{coefficient of }t^{\delta-\ell+\alpha-n}\text{ in }L(\sum_{i=0}^{n}g_i(\delta)t^{\delta-i-\ell});
\end{split}
\end{equation}
note that the second equality follows by Lemma~\ref{Lemma:ShiftRelationInXi}: $L$ applied to $g_i\xi t^{-\ell-i}$ with $i>n$ yields only term contributions where the exponent of $t$ is smaller than $\alpha-\ell-n$.
By property~(i) from above and $(g_0,g_1,\dots,g_n)\in V_n$ it follows that there exists a $\rho\in\set N$ with $\rho\geq n+\ell$ such that for any $\delta\geq\rho$ we have that
\begin{equation}\label{Equ:SpecialSol}
\deg(L(\sum_{i=0}^{n}g_i(\delta)t^{\delta-i-\ell}))< \delta-\ell+\alpha-n.
\end{equation}
Since~\eqref{Equ:p(delta)} holds, it follows that $p(\delta)=0$ for all $\delta\geq\rho$ which implies that $p(x)$ must be the zero-polynomial; a contradiction. Consequently~\eqref{Equ:LYSol} holds for all $u\in\set N$.\\ 
To finish the proof of termination, let $\ordSymbol$ be the order of the operator $L$. Since $\Const_{\sigma}{k(x)((t^{-1}))(\xi)}=K(x)$ by Corollary~\ref{Cor:UnchangedConstant}, we can apply~\cite[Thm. XII (page 272)]{Cohn:65} and it follows that there exist in $(k(x)((t^{-1}))(\xi),\sigma)$
at most $\ordSymbol$ linearly independent solutions over $K(x)$ for $L(y)=0$. However, $Y_0,Y_1,\dots,Y_{\ordSymbol+1}$ are linearly independent over $K(x)$. Namely, suppose that there are $c_0,\dots,c_{\ordSymbol+1}\in K(x)$ with 
$$0=c_0\,Y_0+c_1\,Y_1+\dots+c_{\ordSymbol+1}\,Y_{\ordSymbol+1}=\xi(c_0+c_1\,t^{-1}+\dots+c_{\ordSymbol+1}t^{-\ordSymbol-1}).$$
Since $\xi$ is transcendental over $k(x)((t^{-1}))$, $c_0+c_1\,t^{-1}+\dots+c_{\ordSymbol+1}t^{-\ordSymbol-1}=0$. Furthermore, $c_i=0$ for $0\leq i\leq \ordSymbol+1$ since $t$ is transcendental over $k(x)$. This proves that there are $\ordSymbol+1$ linearly independent solutions of $L(y)=0$; a contradiction to our construction. Hence the assumption must be wrong that the iteration does not terminate.
\end{proof}

\begin{remark}
\rm
\begin{enumerate}
\item The proof of Theorem~\ref{th:SigmaBounds} is inspired by \cite[Lemma 3.8]{Singer:91} that treats the continuous version of \sigmaSE-extensions in order to solve linear differential equations in Liouvillian extensions. More precisely, a degree bound for polynomial solutions for a primitive extension $k(t)$ is derived where the derivative of $t$ is an element in $k$ (i.e., $t$ represents an indefinite integral).
Similar to our strategy, Singer's algorithm produces step-wise constraints for the top most coefficients of a potential solution. Eventually, these constraints enable one to determine a degree bound for a primitive extension. As in our case, the termination argument does not yield an explicit upper bound of the required iteration steps.

\item In~\cite{Abramov:89b} and~\cite{Petkovsek1992} degree bound algorithms are elaborated that deal with the special case $(k(t),\sigma)$ where $\Const_{\sigma}(k)=k$ and $\sigma(t)=t+1$. As it turns out, both algorithms are equivalent~\cite{PetkovWeix} and can be related to the algorithm in Theorem~\ref{th:SigmaBounds}. 
In contrast to the general situation in Theorem~\ref{th:SigmaBounds}, it can be shown that in the algorithm given in~\cite{Petkovsek1992} at most $\deg(L)$ steps are required until the degree bound procedure terminates.
 
\item In~\cite[Chapter~3.4]{Schneider:01} a variant of the above method has been elaborated for a \sigmaSE-monomial $t$, but the termination proof could not be provided. Interestingly enough, it could be shown in~\cite[Cor.~3.14.3]{Schneider:01} that after at most $\deg(L)$ steps the iteration process terminates if the operator $L$ is free of $t$, i.e., $L\in k[E;\sigma]$. Even more is valid in this special case: for the parameterized difference equation~\eqref{Equ:OrgEquation} with $\Phi_1,\dots,\Phi_m \in k[t]$ a polynomial degree bound is $\deg(L)+\max(\deg(\Phi_1),\deg(\Phi_2),\dots,\deg(\Phi_m),0)$; see also~\cite[Section~8]{Schneider:05b} for further discussions.
\end{enumerate}
\end{remark}

\begin{example}\label{Exp:HyperDFHn9}
\rm
Given the operator $L=a_0\,E^0+a_1\,E^1+a_2\,E^2$ with~\eqref{Equ:HnCoeff} we start the degree bounding algorithm from Theorem~\ref{th:SigmaBounds} with $n=0$; here $k=\set Q(x)$ and $t=h$. We have $\alpha=\deg(L)=2$ and thus get (after removing common factors) the operator
$$L_2=(1 + x)\,E^0 -(3 + 2 x)\,E^1+ (2 + x)\,E^2.$$
Solving the underlying recurrence in $\set Q(x)$ gives $\{w\in\set Q(x): L_2(w)=0\}=\set Q$.
Thus we get $f_1^{(0)}=1$. In particular, with $y=\kappa_1\,h^{\delta}$ for $\kappa_1\in\set Q$ we get $\deg(L(y))<\delta+\alpha=\delta+2$; by luck we even have $\deg(L(y))<\delta+1$. 
We repeat this construction for $n=1$ and try to find all $\kappa_1\in\set Q$ and $w\in\set Q(x)$ such that for $y=\kappa_1\,(1\cdot h^{\delta})+w\,h^{\delta-1}$ we get $\deg(L(y))<\delta+\alpha-1=\delta+1$; compare~\eqref{Equ:IncrementalAnsatz} with~\eqref{Equ:yTilde}.
This gives the constraint  $L_2(w)=0$, i.e., $w\in\set Q$. Consequently, we can set $f_1^{(1)}=(1,0)$ and $f_2^{(1)}=(0,1)$. We repeat the process for $n=2$ and make the ansatz $y=\kappa_1(1\cdot h^{\delta}+0\cdot h^{\delta-1})+\kappa_2(0\cdot h^{\delta}+1\cdot h^{\delta-1})+w\,h^{\delta-2}$; compare again~\eqref{Equ:IncrementalAnsatz} with~\eqref{Equ:yTilde}. This time we get the constraint
\begin{equation}\label{Equ:n=2Constraint}
L_2(w)=-\kappa_1\frac{(3 + 2 x)}{(x+1)(x+2)}\binom{\delta-1}{2}.
\end{equation}
To complete the step $n=2$ we must find a solution with $\kappa_1\neq0$, which is not possible. This stops our degree bounding algorithm. In particular, we can conclude that for a generic $\delta$ with $\delta>1$ only the two highest terms $h^{\delta+2}$ and $h^{\delta+1}$ in $L(y)$ can vanish. This indicates that $n=2$ is  a degree bound. Furthermore, only for the special cases $\delta=1,2$ and $w\in\set Q$ in~\eqref{Equ:n=2Constraint} one can find $\kappa_1\neq0$. Thus no special cases for $\delta>2$ can arise, and $b=2$ is indeed a degree bound. 
\end{example}

\section{A general framework and a complete algorithm for \pisiSE-fields}
\label{Sec:Framework}

In this section we will present a complete algorithm to compute all hypergeometric solutions of homogeneous linear difference equations and rational solutions of parameterized linear difference equations in a \pisiSE-field $(F,\sigma)$ with the property that the constant field $\Const_{\sigma}{F}$ is a rational function field over an algebraic number field. More generally, we will provide a general framework that solves these problems in a difference field  $(F,\sigma)$ that is built by a tower of \pisiSE-monomials over a difference field $(k,\sigma)$ that satisfies certain (algorithmic) properties.

\subsection{A general framework for nested \pisiSE-monomials}

In order to accomplish this task, we start to glue together the main results of the previous sections.
Theorem~\ref{th:hyper} combined with part~(1) of Theorem~\ref{Thm:RatSolver} provides the following reduction mechanism: if one can compute all hypergeometric candidates in $k$, can solve parameterized linear difference equations (PLDE) in $k$ and can solve the pseudo-orbit problem in $k$, then one can compute all hypergeometric solutions of homogeneous difference equations in a difference field $(k(t),\sigma)$ built by a \pisiSE-monomial $t$. This reduction process can be visualized in Figure~\ref{Fig:HypSolutions}.
\begin{figure}[h]
\footnotesize
$$\xymatrix@R=0.7cm@C=0.7cm{
	&&*+[F-:<4pt>]\txt{\MainRed hg.\ solutions}\\
	\txt{$k(t)$}&\txt{\MainRed hg.\ candidates}\ar[ur]&\ar[u]\txt{\MainRed  poly.\ solutions\\\MainRed of homogeneous\\\MainRed equations} & \\
	\txt{$k$}&\txt{\MainRed hg.\ candidates}\ar[u]\ar[ur]&\txt{\MainRed rat.\ solutions\\\MainRed of PLDE}\ar[u]&\ar[lu]\txt{pseudo orbit}\\
}$$
\normalsize
\caption{Reduction step for hypergeometric solutions}\label{Fig:HypSolutions}
\end{figure}
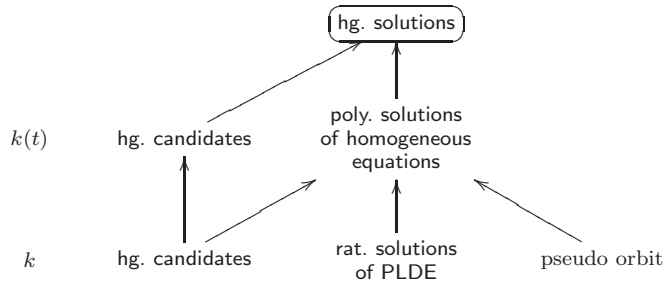

In addition, Theorem~\ref{th:hyper} combined with part~(2) of Theorem~\ref{Thm:RatSolver} delivers the reduction step illustrated in Figure~\ref{Fig:RatSolutions}.

\begin{figure}[h]
	\footnotesize
$$\xymatrix@R=0.7cm@C=0.7cm{
	\txt{$k(t)$}&\txt{\MainRed hg.\ candidates}&\txt{\MainRed rat.\ solutions\\\MainRed of PLDE} && \ar[ll]\txt{dispersion}\\
	\txt{$k$}&\txt{\MainRed hg.\ candidates}\ar[u]\ar[ur]&\txt{\MainRed solutions\\\MainRed of PLDE}\ar[u]&\ar[lu]\txt{pseudo orbit}\\
}$$
\normalsize
\caption{Reduction step for rational solutions (and hypergeometric candidates)}\label{Fig:RatSolutions}
\end{figure}
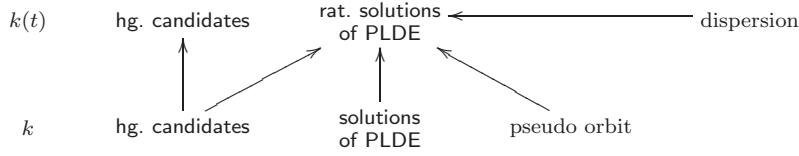

Combining Theorems~\ref{th:hyper} and~\ref{Thm:RatSolver}, i.e., combining Figure~\ref{Fig:HypSolutions} with iterative applications of Figure~\ref{Fig:RatSolutions} yields the following basic framework.

\begin{theorem}\label{Thm:PiSiRec}
	Let $(k,\sigma)$ be a difference field and let $(k(t_1),\dots(t_e),\sigma)$ be a tower of \pisiSE-monomials over $k$.
	\begin{itemize}
		\item[(i)]
		If we can compute all the hypergeometric
		candidates for equations with coefficients in $k$, then we can compute
		all the hypergeometric candidates for equations with coefficients in $k(t_1)\dots(t_e)$.
		\item[(ii)]
		If, in addition, we can compute the spread in $k(t_1,\ldots,t_{i-1})[t_i]$ for $1\le i\le e$, can solve the pseudo-orbit problem in  $k(t_1,\ldots,t_{i-1})$ for $1\le i\le e$ and can solve parameterized linear difference equations in $k$, then we can compute all solutions of parameterized linear difference equations and all the hypergeometric solutions of equations with coefficients in $k(t_1)\dots(t_e)$.
	\end{itemize}
\end{theorem}

It follows from~\cite{Karr81} that the pseudo-orbit problem can be solved algorithmically in a \pisiSE-field $(k,\sigma)$ over $C = \Const_\sigma(k)$ (with some additional properties required on the constant field $C$ which are given below). Namely, the following parameterized pseudo-orbit (PPO) problem can be handled~\cite[Theorem~9]{Karr81}. Given $f_1,\dots,f_n \in k^*$, one can compute a basis of
$$
M(f_1,\dots,f_n; k) = \{(e_1,\dots,e_n) \in \mathbb{Z}^n
\st 1 \sim_{k,\sigma} \prod_{i=1}^n f_i^{e_i}\}.
$$
Note that $M(f_1,\dots,f_n; k)$ is a finite-dimensional free $\mathbb{Z}$-submodule
of $\mathbb{Z}^n$; we refer to~\cite[Lemma 6]{Karr81}. Since
$$\Gamma(u,v;k) = \{\gamma \in \mathbb{Z} \st (\gamma, -1) \in M(u,v;k)\},$$
this yields in turn an algorithm for computing $\Gamma(u,v;k)$ in such fields.

\begin{example}\label{Exp:KarrM}
	\rm
	Using Karr's algorithm from~\cite{Karr81} we get
	\begin{align*}
	M(x+1,1;\set Q(x))&=(0,1)\set Z=(0,-1)\set Z,\\
	M(x+1,\tfrac1{2 (1 + x)};\set Q(x))&=\{(0,0)\},\\
	M(x+1,\tfrac{1 + x}{x^2};\set Q(x))&=(1,1)\set Z=(-1,-1)\set Z
	\end{align*}
	in the \pisiSE-field $(\set Q(x),\sigma)$ with $\sigma(x)=x+1$. This gives 
	$\Gamma(x+1,1;\set Q(x))=\{0\}$,
	$\Gamma(x+1,\frac1{2 (1 + x)};\set Q(x))=\emptyset$
	and $\Gamma(x+1,\frac{1 + x}{x^2};\set Q(x))=\{-1\}$
	as used in Example~\ref{Exp:HyperDFn!6}(1).
\end{example}

More generally, following~\cite{KS:06} this problem and also the dispersion can be computed over a tower of \pisiSE-monomials over $(k,\sigma)$ whenever the ground field $(k,\sigma)$ is $\sigma^*$-computable.

\begin{definition}\label{Def:computable}
	A difference field $(k,\sigma)$ is called
	$\sigma^*$-com\-pu\-ta\-ble if the following holds.
	
	\begin{enumerate}
		\item There is an algorithm that can factor multivariate polynomials
		over $k$.
		
		\item $(k,\sigma^s)$ is torsion free for any $s\in{\set Z}^*$, i.e.,
		\begin{equation}\label{Equ:TorsionFree}
		\forall s,r\in{\set Z}^*\;\forall f\in k^*:
		f\sim_{k,\sigma^s}1\text{ and }f^r=1\Rightarrow f=1.
		\end{equation}
		
		\item There is an algorithm that solves the $\Pi$-Regularity problem:
		Given $(k,\sigma)$ and $f,g\in k^*$; find, if possible, an $n\geq0$ with
		$\dfact{f}{n}{\sigma}=g$.
		
		\item There is an algorithm that solves the $\Sigma$-Regularity problem: Given $(k,\sigma)$, $r\in{\set Z}^*$ and $f,g\in k^*$;
		find, if possible, an $n\geq0$ with
		$\dfact{f}{0}{\sigma^r}+\cdots+\dfact{f}{n}{\sigma^r}=g$.
		
		\item There is an algorithm that solves the parameterized pseudo-orbit problem:
		Given $f_1,\dots,f_n\in k^*$;
		compute a basis of
		$M(f_1,\dots,f_n; k)$ over $\set Z$.
	\end{enumerate}
	$(k,\sigma)$ is called $\sigma$-computable if it is $\sigma^*$-computable and, in addition, one can compute all solutions in $k$ of parameterized linear difference equations with coefficients in $k$ and one can compute all hypergeometric candidates for linear difference equations with coefficients in $k$.
\end{definition}

Obviously, if a difference field extension $(F,\sigma)$ of a difference field $(k,\sigma)$ is $\sigma^*$-computable (resp.\ $\sigma$-computable), then also the ground field $(k,\sigma)$ is $\sigma^*$-computable (resp.\ $\sigma$-computable). But also the other direction holds if $(F,\sigma)$ is built by \pisiSE-monomials.

\begin{theorem}\label{Thm:SigmaComputable}
	Let $(k,\sigma)$ be a difference field and $t$ be a \pisiSE-monomial over $k$.
	\begin{itemize}
		\item[(i)]
		If $(k,\sigma)$ is
		$\sigma^*$-com\-pu\-ta\-ble, $(k(t),\sigma)$ is
		$\sigma^*$-com\-pu\-ta\-ble; in particular, one can compute the spread in $k[t]$.
		\item[(ii)]
		If $(k,\sigma)$ is $\sigma$-computable, $(k(t),\sigma)$ is $\sigma$-computable; in particular, one can compute all hypergeometric solutions of equations with coefficient in $k(t)$.
	\end{itemize}
\end{theorem}
\begin{proof}
	(i)~Suppose that $(k,\sigma)$ is $\sigma^*$-computable.
	By \cite[Theorem~1]{KS:06} (analyzing Karr's algorithm~\cite{Karr81}) the difference field $(k(t),\sigma)$ is $\sigma^*$-computable and one can decide algorithmically if for any $f,g\in k[t]^*$ there exists an $r\in\set Z$ such that $\sigma^r(f)/g\in k$. With~\cite[Lemma~15]{Bronstein2000} and the text below of~\cite[Example~10]{Bronstein2000} (see~\cite{Schneider:04c} for further details) it follows that one can compute spreads in $k[t]$.\\
	(ii)~Suppose in addition that $(k,\sigma)$ is $\sigma$-computable. Then it follows by Theorem~\ref{th:hyper}.(i) that one can compute all hypergeometric candidates for linear difference equations with coefficients in $k(t)$.  Hence by Theorem~\ref{Thm:RatSolver} we can compute all solutions of parameterized linear difference equations with coefficients in $k(t)$. Thus $(k(t),\sigma)$ is $\sigma$-computable. Moreover, by Theorem~\ref{th:hyper}.(ii) we can compute all hypergeometric solutions of linear difference equations with coefficients in $k(t)$.
\end{proof}

The reduction step of Theorem~\ref{Thm:SigmaComputable} is also visualized in Figure~\ref{Fig:SigmaComputable}. 
\begin{figure}[h]
	\footnotesize
	$$\xymatrix@R=0.7cm@C=0.5cm{
		\txt{$k(t)$}&\txt{$\sigma$-computable}\ar@{.}[]+<3.5em,0em>;[d]+<3.5em,0em>&\txt{\MainRed hg.\ candidates}&\txt{\MainRed rat.\ solutions\\ \MainRed of PLDE} &\txt{$\sigma^*$-computable}\\
		\txt{$k$}&\txt{$\sigma$-computable}\ar[u]&\txt{\MainRed hg.\ candidates}\ar[u]\ar[ur]&\txt{\MainRed solutions\\\MainRed of PLDE}\ar[u]&\ar[u]\ar[ul]\txt{$\sigma^*$-computable}\\
	}$$
	\normalsize
	\caption{Reduction step using $\sigma^*$-comparability}\label{Fig:SigmaComputable}
\end{figure}
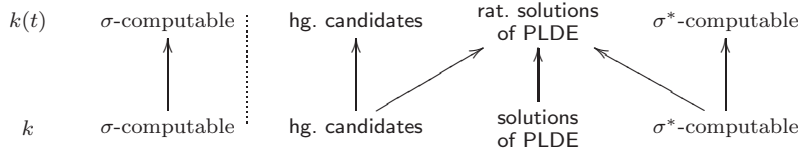

Namely, if $(k,\sigma)$ is $\sigma^*$-computable and one can compute all hypergeometric candidates and finds all solutions of PLDE in $(k,\sigma)$ (in short $(k,\sigma)$ is $\sigma$-computable), then also $(k(t),\sigma)$ is  $\sigma^*$-computable and one can compute all hypergeometric candidates and finds all solutions of PLDE in $(k,\sigma)$ (in short $(k(t),\sigma)$ is $\sigma$-computable).

Summarizing, a difference field $(k(t_1)\dots(t_e),\sigma)$ built by a tower of \pisiSE-monomials is $\sigma^*$-computable (resp.\ $\sigma$-computable)  whenever $(k,\sigma)$ is $\sigma^*$-com\-putable (resp.\ $\sigma$-computable). In particular, setting $k_j=k(t_1)\dots(t_j)$ for $1\leq j\leq e$ (note that $k_0=k$) we obtain the tower of reductions in Figure~\ref{Fig:SigmaComputableRecursive}.
\begin{figure}[h]
	\footnotesize
	$$\xymatrix@R=0.7cm@C=0.5cm{
	&&&*+[F-:<4pt>]\txt{\MainRed hg.\ solutions}\\
		\txt{$k_e$}&&\txt{\MainRed hg.\ candidates}\ar[ur]&\txt{\MainRed  poly.\ solutions\\\MainRed of homogeneous\\\MainRed equations}\ar[u] &\\\
		\txt{$k_{e-1}$}&\txt{$\sigma$-computable}\ar@{.}[]+<3.5em,0em>;[dddd]+<3.5em,0em>&\txt{\MainRed hg.\ candidates}\ar[u]\ar[ur]&\txt{\MainRed rat.\ solutions\\\MainRed of PLDE}\ar[u]&\ar[ul]\txt{$\sigma^*$-computable}\\
		\txt{$k_{e-2}$}&\txt{$\sigma$-computable}\ar[u]&\txt{\MainRed hg.\ candidates}\ar[u]\ar[ur]&\txt{\MainRed rat.\ solutions\\\MainRed of PLDE}\ar[u]&\ar[ul]\ar[u]\txt{$\sigma^*$-computable}\\
		&&&&\\
		\txt{$k_{1}$}&\txt{$\sigma$-computable}\ar@{.>}[uu]&\txt{\MainRed hg.\ candidates}\ar@{.>}[uu]\ar@{.>}[uur]&\txt{\MainRed rat.\ solutions\\\MainRed of PLDE}\ar@{.>}[uu]&\ar@{.>}[uul]\ar@{.>}[uu]\txt{$\sigma^*$-computable}\\
		\txt{$k_{0}$}&\txt{$\sigma$-computable}\ar[u]&\txt{\MainRed hg.\ candidates}\ar[u]\ar[ur]&\txt{\MainRed solutions\\\MainRed of PLDE}\ar[u]&\ar[ul]\ar[u]\txt{$\sigma^*$-computable}\\
	}$$
	\normalsize
\caption{Reduction step using $\sigma^*$-computability}\label{Fig:SigmaComputableRecursive}
\end{figure}
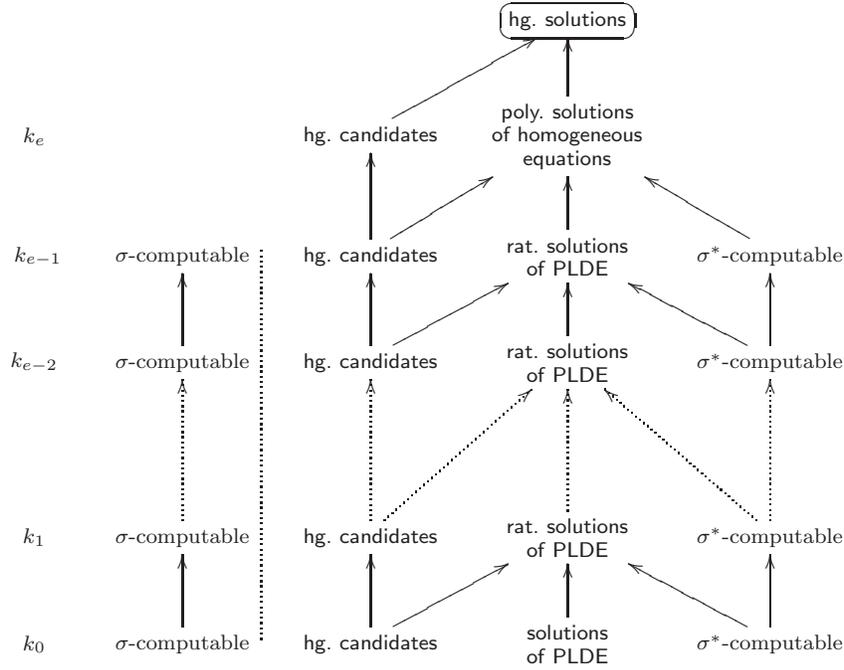

\subsection{A complete algorithm for \pisiSE-fields over certain constant fields}

If $(k_e,\sigma)$ is a \pisiSE-field over $k=k_0$, i.e., $\Const_\sigma(k)=k$ or equivalently $\sigma|_k=\text{id}$, the reduction process terminates in the base case $k_0$. In this case we also write $C:=k$. The following result enumerates properties of $k_0=k(=C)$ that enable one to tackle the base case and thus to turn the reduction process given in Figure~\ref{Fig:SigmaComputableRecursive} to a complete algorithm; compare also~\cite{KS:06}.

\begin{theorem}\label{Thm:GroundField}
	Let $(k,\sigma)$ be a difference field with $\Const_{\sigma}(k)=k$ in which we can perform the usual operations. If, in addition, there are algorithms available for

\vspace*{-0.3cm}

	\begin{enumerate}
		\item factoring in $k[t_1,\dots,t_e]$,
		\item computing for any $c_1,\dots,c_n \in k$ a $\mathbb{Z}$-basis of 
	$$M(c_1,\dots,c_n; k)= \{(e_1,\dots,e_n) \in \mathbb{Z}^n
	\st 1=\prod_{i=1}^n c_i^{e_i}\},$$
	\item recognizing whether any $c \in k$ is an integer,
	\end{enumerate}
	then $(k,\sigma)$ is $\sigma$-computable.
\end{theorem}
\begin{proof}
	Any constant field is torsion-free,
	because $\{\sigma^i(g)/g: g\in k^*\}=\{1\}$ for all $i\in\set N$.
	$\Pi$-regularity can be decided with \cite[Lemma~2]{Karr81} and using the property that one can solve the orbit problem (which is covered as special case by property~(2)).
	$\Sigma$-regularity can be decided with \cite[Lemma~3]{Karr81}, using again the property that one can solve the orbit problem and
	property~(3). The parameterized pseudo-orbit problem in the constant field is property~(3). Thus $k$ is $\sigma^*$-computable.\\
	Moreover,
	solving parameterized linear difference equations with coefficients over constants reduces immediately to solving a linear system.
	Finally, finding all hypergeometric candidates for an equation with constant coefficients is equivalent to finding all roots in $k$ of a certain univariate polynomial, and factoring this polynomial (property~(1)) yields all such roots.
\end{proof}

\begin{remark}
	\rm
	The pseudo-orbit problem is needed to check if a unimonomial $t$ over $k$ with $\eta:=\sigma(t)/t\in k$ is a \piE-monomial (i.e., if there is an $n\in\set Z$ with $\eta^n\sim_{k,\sigma}1$), or it is needed to derive the denominator and degree bounds to compute rational solutions in $k(t)$ when $t$ is a \piE-monomial over $k$ (see Theorem~\ref{th:PiBounds}). As mentioned above this problem can be solved in a tower of \pisiSE-extensions by dealing with the more general problem to solve the parameterized pseudo orbit problem in the fields below. In a \pisiSE-field this finally leads to the problem to deal with requirement~(2) in Theorem~\ref{Thm:GroundField}. If one deals only with \sigmaSE-extensions, the pseudo-problem does not arise directly. However, in general one enters in the problem to compute the spread in $k[t]$ and thus has to solve in certain instances the \piE- and $\Sigma$-Regularity problem. However, as stated in the proof of Theorem~\ref{Thm:GroundField}, it suffices to solve the orbit problem
	(and not the more general parameterized orbit problem, property (2) in Theorem~\ref{Thm:GroundField}): given $f,g\in\Const_{\sigma}(k)$, decide constructively if there is an $n\in\set Z$ with $f^k=g$. Precisely this problem has been considered in more details in~\cite{AbraBron2000}. 
\end{remark}

In short, the algorithmic machinery visualized in Figure~\ref{Fig:SigmaComputableRecursive} terminates with the base case given in Figure~\ref{Fig:ConstantField}.
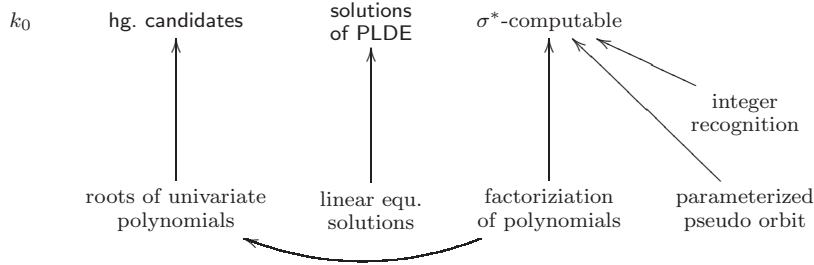
\begin{figure}[h]
\footnotesize
$$\xymatrix@R=0.5cm@C=0.5cm{
\txt{$k_{0}$}&\txt{\MainRed hg.\ candidates}&\txt{\MainRed solutions\\\MainRed of PLDE}&\txt{$\sigma^*$-computable}\\
&&&&\ar[ul]\txt{integer\\recognition}\\
&\txt{roots of univariate\\ polynomials}\ar@{->}[uu]&
\txt{linear equ.\\solutions}\ar@{->}[uu]&
\ar@/^1.8pc/@[black][ll]\txt{factoriziation\\ of polynomials}\ar@{->}[uu]&
\txt{parameterized\\ pseudo orbit}\ar@{->}[uul]\\
}$$	
\normalsize
	\normalsize
\caption{Base case: the constant field $k_0=k$}\label{Fig:ConstantField}
\end{figure}

As remarked above, special cases for the parameterized pseudo-orbit problem in the constant field (property (2) in Theorem~\ref{Thm:GroundField}) have been treated in \cite{AbraBron2000}. The general problem has been solved in~\cite{Ge:93b} for an algebraic number field $A$ and based on that has been solved for a rational function field over $A$ in~\cite[Thm~3.5]{Schneider05c}. As a consequence any rational function field over an algebraic number field is $\sigma$-computable. Summarizing, iterative application of Theorem~\ref{Thm:SigmaComputable} (or applying the reduction process sketched in Figure~\ref{Fig:SigmaComputableRecursive} together with the base case given in Figure~\ref{Fig:ConstantField}) results in the following algorithmic result.

\begin{corollary}
A rational function field $C=A(y_1,\dots,y_{\rho})$ over an algebraic number field\footnote{An algebraic number field $A$ is an algebraic field extension of $\set Q$ with finite degree; in particular, $A$ can be considered as a finite-dimensional vector space over $\set Q$.} $A$ is $\sigma$-computable. In particular, one can compute all solutions of parameterized linear difference equations and all hypergeometric solutions of homogeneous linear difference equations in a \pisiSE-field defined over $C$. 
\end{corollary}

In a nutshell, we obtain a complete algorithm to find all hypergeometric solutions of linear difference equations in a \pisiSE-field $(F,\sigma)$ with $F=C(t_1)\dots(t_e)$ over a constant field $C$ with certain algorithmic properties. These requirements hold in particular if $C=A(y_1\dots,y_{\rho})$ is a rational function field over an algebraic number field $A$. 
If one restricts to the rational case $F=C(t_1)$ with $\sigma(t_1)=t_1+1$ (resp. to the $q$-rational case $\sigma(t_1)=q\,t_1$ where $q\in\{y_1,\dots,y_{\rho}\}$), the derived algorithm is equivalent to the algorithm Hyper~\cite{Petkovsek1992} (resp.\ $q$-Hyper~\cite{APP:98}). Also the mixed multibasic version described in~\cite{BP:99} is contained in our generalized machinery. In \texttt{Sigma} most of these ideas are implemented (and are continuously challenged and improved by new examples).

\begin{example}\label{Exp:MMAHyper}
\rm
After loading

\begin{mma}
	\In << Sigma.m \\
	\vspace*{-0.1cm}
	\Print \LoadP{Sigma - A summation package by Carsten Schneider
		\copyright\ RISC-JKU}\\
\end{mma}

\vspace*{0.2cm}

\noindent into the computer algebra system Mathematica, we can perform the calculations illustrated in Example~\ref{Exp:HyperDFHn1} as follows. We enter the rational functions $(a_0,a_1,a_2)\in\set Q(x)(h)^3$ of the operator~\eqref{Equ:Operator2} in

\begin{mma}
	\In PL=	\{(1
	+h
	+h x
	)^2 (
	3
	+2 h
	+2 x
	+3 h x
	+h x^2
	)^2,\newline
\hspace*{1cm}	-h (1+x) (3+2 x) (
	3
	+2 h
	+2 x
	+3 h x
	+h x^2
	)^2,
	h (1+x)^2 (2+x)^3 (1
	+h
	+h x
	)\};
	\\
\end{mma}

\vspace*{0.1cm}

\noindent and define the automorphism $\sigma:\set Q(x)(h)\to\set Q(x)(h)$ with $\sigma(x)=x+1$ and $\sigma(h)=h+\frac1{x+1}$ by

\begin{mma}
\In tower=\{\{x,1,1\},\{h,1,\frac1{x+1}\}\};\\
\end{mma}

\vspace*{0.1cm}

\noindent Then we can compute the rational functions $r_j\in\set Q(x)(h)$ for $j=1,2$ of the first order factors  $E-r_j$ of ~\eqref{Equ:Operator2} with the function call

\begin{mma}
	\In SigmaHyper[PL,tower]\\
	\Out \{\frac{(1 + h (1 + x))^3}{(h^2 (1 + x)^3)}, \frac{(1 + h (1 + x))^2}{(h (1 + x)^2)}\}\\
\end{mma}

\vspace*{0.1cm}

\noindent As illustrated in Example~\ref{Exp:HyperDFHn2} it is beneficial if one extracts the underlying components $(v_j,u_j)\in k^2$ with $r_j=u_j\frac{\sigma(v_j)}{v_j}$. In our particular example this can be carried out with

\begin{mma}
\In SigmaHyperComponent[PL,tower,HyperStrategy \to "PolynomialVersion"]\\
\Out \{\{h^2, \frac{1 + h + h x}{1 + x}\}, \{h, \frac{1 + h + h x}{1 + x}\}\}\\
\end{mma} 

\vspace*{0.1cm}

\noindent following the steps given in Example~\ref{Exp:HyperDFHn4}. In particular, the improvements (1)--(3) of Remark~\ref{Remark:Improvement} are utilized; see also Example~\ref{Exp:HyperDFHn5}. With the option \MText{HyperStrategy $\to$ "RationalVersion"} (which is the standard choice in \texttt{Sigma}) also the improvement~(4) of Remark~\ref{Remark:Improvement} is activated (compare Example~\ref{Exp:HyperDFHn5b}) which leads to the following components $(v_j,u_j)$:

\begin{mma}\MLabel{MMA:RatCase}
	\In SigmaHyperComponent[PL,tower,HyperStrategy \to "RationalVersion"]\\
	\Out \{\{h^3, h\}, \{h^2, h\}\}\\
\end{mma} 

\vspace*{0.1cm}

\noindent The above function calls can be executed if the difference field $(k,\sigma)$ is explicitly given by a tower of \pisiSE-monomials. However, in most applications the problem is stated in terms of indefinite nested sums and products. 
Using the framework elaborated in~\cite{Karr81,Schneider:01,DR1,DR3,OS:18,Schneider:20} it suffices to enter the given recurrence, and the
construction of the underlying difference field (resp.\ ring) is completely automatized. 
For instance, after entering 

\begin{mma}
\In rec=	\big(
1
+H_\nu
+\nu H_\nu
\big)^2 (
3
+2 \nu
+2 H_\nu
+3 \nu H_\nu
+\nu^2 H_\nu)^2 F[\nu]\newline
\hspace*{2cm}-(1+\nu) (3+2 \nu) H_\nu\big(
3
+2 \nu
+2 H_\nu
+3 \nu H_\nu
+\nu^2 H_\nu
\big)^2 F[1+\nu]\newline 
\hspace*{1cm}+(1+\nu)^2 (2+\nu)^3 H_\nu\big(
1
+H_\nu
+\nu H_\nu
\big)  F[2+\nu]
==0;\\
\end{mma}

\vspace*{0.1cm}

\noindent from Example~\ref{Exp:RecExp}(I) one obtains the hypergeometric solutions in terms of the harmonic numbers with the function call

\begin{mma}
\In SolveRecurrence[rec,F[{\nu}]]\\

\vspace*{-0.1cm}

\Out \{\{0,H_{\nu}\,\prod_{i=1}^{\nu} H_i\},\{0,H_{\nu}^2\,\prod_{i=1}^{\nu} H_i\}\}\\
\end{mma}

\vspace*{0.1cm}

\noindent Similarly, one can solve the recurrence 

\begin{mma}
\In rec=-2 (1+\nu)^2 (2+\nu) \big(
7
+6 \nu
+\nu^2
+3 \nu!
+5 \nu \nu!
+2 \nu^2 \nu!
\big) (\nu!)^2F[\nu]\newline
\hspace*{1.5cm}+(1+\nu) (2+\nu) \big(
16
+16 \nu
+3 \nu^2
+7 \nu!
+12 \nu \nu!
+4 \nu^2 \nu!
\big) \nu!F[1+\nu]\newline
\hspace*{1cm}-\big(
2
+4 \nu
+\nu^2
+\nu!
+2 \nu \nu!
\big) F[2+\nu]==0;\\
\end{mma}

\vspace*{0.1cm}

\noindent from Example~\ref{Exp:RecExp}(II) straightforwardly:

\begin{mma}
\In SolveRecurrence[rec,F[{\nu}]]\\

\vspace*{-0.1cm}

\Out \{\{0,\prod_{i=1}^{\nu} i!\},\{0,({\nu}!+{\nu}^2)2^{\nu}\prod_{i=1}^{\nu} i!\}\}\\
\end{mma}
\end{example}

\subsection{Controlling the hypergeometric candidates and the extension of the constant field}\label{Sec:ControlConstants}

In various applications it is sufficient to carry out the proposed method in the smallest \pisiSE-field in which the coefficients of the recurrence can be represented. E.g., in calculations coming from particle physics~\cite{QCD} one computes all hypergeometric solutions in the field $C(t_1)$ with $\sigma(t_1)=t_1+1$ and constant field $C=\set Q(y_1,\dots,y_{\rho})$. 
In other applications it turns out that the constant field $C$ has to be extended to a larger field $\tilde{C}$ in order to find the desired hypergeometric solutions in $(\tilde{F},\sigma)$ with $\tilde{F}=\tilde{C}(t_1)\dots(t_e)$. Before we continue this discussion, we first show that $(\tilde{F},\sigma)$ forms again a \pisiSE-field over the larger constant field $\tilde{C}$. More generally, we obtain the following result.

\begin{proposition}
Let $(k,\sigma)$ be a difference field and let $(k(t_1),\dots(t_e),\sigma)$ be a tower of \pisiSE-monomials over $k$ with $\sigma(t_i)=\alpha_i\,t_i+\beta_i$ for $1\leq i\leq e$. Let $(k(S),\sigma)$ be a difference field extension of $(k,\sigma)$ where $S$ is a generator set\footnote{We do not want to specify $S$ further. However, if we take $y\in k(S)$ then we can take a finite set $S_0\subseteq S$ such that $y\in k(S_0)$. Here $S_0$ can be built either by algebraic or transcendental generators.} with $\sigma(c)=c$ for all $c\in S$. Take the difference field extension $(k(S)(t_1),\dots(t_e),\sigma)$ of $(k(S),\sigma)$ where $k(S)(t_1)\dots(t_e)$ is a rational function field extension of $k(S)$ and  $\sigma(t_i)=\alpha_i\,t_i+\beta_i$ for $1\leq i\leq e$. Then:
\begin{enumerate}
\item $(k(S)(t_1)\dots(t_e),\sigma)$ is built by a tower of \pisiSE-monomials over $k(S)$; in particular,
$\Const_{\sigma}{k(S)(t_1)\dots(t_e)}=\Const_{\sigma}{k(S)}$.
\item If $(k(t_1)\dots(t_e),\sigma)$ is a \pisiSE-field over $k$, then $(k(S)(t_1)\dots(t_e),\sigma)$ is a \pisiSE-field over $k(S)$. 
\end{enumerate} 
\end{proposition}
\begin{proof}
Suppose that $(k(t_1)\dots(t_e),\sigma)$ is built by a tower of \pisiSE-monomials over $k$ but that $(k(S)(t_1)\dots(t_e),\sigma)$ is not built by a tower of \pisiSE-monomials over $k(S)$. 
Then we can choose $u<e$ such that $t_1,\dots,t_{u-1}$ are \pisiSE-monomials and $t_u$ is not a \pisiSE-monomial. Hence we can take $y\in k(S)(t_1)\dots(t_{u})$ with $\sigma(y)=y$ which depends on $t_u$. By the definition of the generator set it follows that $y\in k(S_0)(t_1)\dots(t_{u})$ for a finite subset $S_0\subseteq S$, say $S_0=\{x_1,\dots,x_r\}$.
By rearranging the generators of $S_0$ we get $y\in k(t_1)\dots(t_{u})(x_1,\dots,x_r)$. Furthermore, by iterative application of Lemma~\ref{Lemma:XExt} we conclude that
\begin{align*}
\Const_{\sigma}&k(t_1)\dots(t_{u-1})(t_u)(x_1,\dots,x_r)\\
&=(\Const_{\sigma}k(t_1)\dots(t_u)(x_1,\dots,x_{r-1}))(x_r)\\
&=\dots=(\Const_{\sigma}k(t_1)\dots(t_u))(x_1,\dots,x_r)\\
&=\Const_{\sigma}k(x_1,\dots,x_r).
\end{align*}
Since $\sigma(y)=y$, it follows that $y\in k(x_1,\dots,x_r)$ and thus $y$ is free of $t_u$, a contradiction. Consequently $(k(S)(t_1)\dots(t_e),\sigma)$ is built by a tower of \pisiSE-monomials over $k(S)$ and hence $\Const_{\sigma}k(S)(t_1)\dots(t_e)=\Const_{\sigma}k(S)$; this proves statement (1).\\
Suppose in addition that $\Const_{\sigma}k=k$, i.e.,
$(k(t_1)\dots(t_e),\sigma)$ is a \pisiSE-field over $k$. Then we can conclude that $\Const_{\sigma}k(S)(t_1)\dots(t_e)=\Const_{\sigma}k(S)=k(S)$. Consequently, $(k(S)(t_1)\dots(t_e),\sigma)$ is a \pisiSE-field over $k(S)$ which proves statement (2).
\end{proof}

Looking at Figure~\ref{Fig:SigmaComputableRecursive} together with the base case given in Figure~\ref{Fig:ConstantField} (or at  Theorem~\ref{Thm:SigmaComputable} and Theorem~\ref{Thm:GroundField}) one can detect easily a source to find more solutions: one can try to extend the constant field $C$ to $\tilde{C}$ in order to find more hypergeometric candidates.
Namely, with part~(1) of Theorem~\ref{th:hyper} it follows that the hypergeometric candidates in $k_e$ are determined by the coefficients of the input difference equation and the hypergeometric candidates for a certain difference equation in $k_{e-1}$. Applying this argument iteratively leads to the base case $k_0=C$ and we apply Theorem~\ref{Thm:GroundField}: the hypergeometric candidates there are given by the roots of certain polynomials with coefficients in $C$. Thus taking the splitting field $\tilde{C}$ for all these polynomials produces the maximal set of possible hypergeometric candidates in the \pisiSE-field $(\tilde{F},\sigma)$ with
$\tilde{F}=\tilde{C}(t_1)\dots(t_e)$.
Finally, looking at the proof of Theorem~\ref{th:hyper},
all hypergeometric solutions in $\tilde{F}=\tilde{k}_e$ can be determined by looping through all hypergeometric candidates in $\tilde{k}_e$ and computing polynomial solutions of a particularly constructed homogeneous linear difference equation in $\tilde{C}(t_1,\dots,t_{e-1})[t_e]$. The following remarks are appropriate: 
\begin{enumerate}
\item So far we can compute these polynomial solutions in $\tilde{C}[(t_1,\dots,t_{e-1})[t_e]$ with our available machinery only if $\tilde{C}$ is given by a rational function field over an algebraic number field. Therefore we cannot apply our toolbox to the algebraic closure of $C=A(y_1,\dots,y_{\rho})$ if $\rho>0$. 
But we have to restrict to $\tilde{C}=\bar{\set Q}(y_1,\dots,y_{\rho})$ where $\bar{\set Q}$ is the algebraic closure of $\set Q$ or to $\tilde{C}=\tilde{A}(y_1,\dots,y_{\rho})$ where $\tilde{A}$ is an appropriately chosen algebraic number field that contains $A$.
\item When one solves the underlying linear difference equations whose coefficients are given in $\tilde{C}(t_1)\dots(t_e)$,
one might ask if one finds more solutions by extending the constant field $\tilde{C}$ further. The following proposition excludes this possibility for algebraic extensions\footnote{It can be shown that Proposition 2 holds for general constant field extensions (not necessarily algebraic) if $k$ is, e.g., a \pisiSE-field. Since this aspect is not crucial for our current considerations, we do not provide a proof here.}.
\end{enumerate}

\begin{proposition}\label{Prop:AlgebraicPLDE}
Let $(k,\sigma)$ be a difference field with $C=\Const_{\sigma}k$ and let	
$(k(S),\sigma)$ be an algebraic difference field extension of $(k,\sigma)$ where $S$ is a generator set with $\sigma(c)=c$ for all $c\in S$. Let $0\neq a=(a_0,a_1,\dots,a_\ell)\in k^{\ell+1}$ and
$b=(b_1,b_2,\dots,b_m)\in k^m$ and suppose that $f_1,\dots,f_n\in k\times C^{m}$ is a basis of ${\mathcal V}(a,b,k)$. Then
\begin{equation}\label{Equ:SolSpaceId}
{\mathcal V}(a,b,k(S))=\{\kappa_1\,f_1+\dots+\kappa_n\,f_n: \kappa_1,\dots,\kappa_n\in C(S)\};
\end{equation}
in particular, $f_1,\dots,f_n$ is also a basis of ${\mathcal V}(a,b,k(S))$. 
\end{proposition}
\begin{proof}
W.l.o.g.\ we may assume that $S$ is a minimal generating set of $k(S)$.
Clearly the inclusion $\supseteq$ in~\eqref{Equ:SolSpaceId} holds. In order to show the other inclusion, let $(g,c_1,\dots,c_m)\in{\mathcal V}(a,b,k(S))$. Then there is $\{x_1,\dots,x_r\}\subseteq S$ such that for the field extension $F=k(x_1,\dots,x_r)$ of $k$ with $\tilde{C}=\Const_{\sigma}(F)$ we have
$(g,c_1,\dots,c_m)\in F\times\tilde{C}^m$. Here we view $k(x_1,\dots,x_i)$ as an algebraic field extension of $k(x_1,\dots,x_{i-1})$ for $1\leq i\leq r$ where the minimal polynomial $\mu_i\in k(x_1,\dots,x_{i-1})[z]$ has degree $d_i>1$. Since $x_1,\dots,x_r$ are algebraic, all elements in $F=k(x_1,\dots,x_r)$ can be represented as polynomial expressions in $x_1,\dots,x_r$, i.e., we may write $F=k[x_1,\dots,x_r]$.
For each $i\geq1$ we define $[i]=\{0,\dots,i-1\}$. Note that by definition the monomials 
\begin{equation}\label{Equ:IndependentSet}
\{x^{\mu}: \mu\in[d_1]\times[d_2]\times\dots\times[d_r]\}
\end{equation}
are linearly independent over $k$  and their linear span over $k$ equals $F$. 
By iterative application of Lemma~\ref{Lemma:XExt} it follows that 
\begin{align*}
\tilde{C}=\,&\Const_{\sigma}k[x_1,\dots,x_r]=(\Const_{\sigma}k[x_1,\dots,x_{r-1}])[x_r]\\
&=\dots=(\Const_{\sigma}k)[x_1,\dots,x_r]=C[x_1,\dots,x_r].
\end{align*}
Thus we can write $g=\sum_{\mu\in U}g_{\mu}x^{\mu}$ and $c_i=\sum_{\mu\in U}c_{i,\mu}x^{\mu}$ with $1\leq i\leq m$ in multi-index notation for a finite set $U\subset[d_1]\times[d_2]\times\dots\times[d_r]$ with $g_{\mu}\in k$ and $c_{i,\mu}\in C$ for $\mu\in U$. In particular, we get
\begin{align*}
\sum_{\mu\in U}x^{\mu}(c_{1,\mu}b_1+\dots+c_{m,\mu}b_m)&=
c_1\,b_1+\dots+c_m\,b_m\\
&=\sum_{i=0}^{\ell}a_i\sigma^i(g)=\sum_{\mu\in U}x^{\mu}\sum_{i=0}^{\ell}a_i\sigma^i(g_{\mu}).
\end{align*}
Since the elements in~\eqref{Equ:IndependentSet} are linearly independent over $k$, it follows for all $\mu\in U$ that
$$c_{1,\mu}b_1+\dots+c_{m,\mu}b_m=\sum_{i=0}^{\ell}a_i\sigma^i(g_{\mu})$$
and thus $h_{\mu}=(g_{\mu},c_{1,\mu},\dots,c_{m,\mu})\in{\mathcal V}(a,b,k)$. Since $f_1,\dots,f_n$ is a basis of ${\mathcal V}(a,b,k)$, we get $h_{\mu}=d_{1,\mu}f_1+\dots+d_{n,\mu}f_n$ for some $d_{i,j}\in C$. Consequently,
\begin{align*}
(g,c_1,\dots,c_m)=&\sum_{\mu\in U}x^{\mu}(g_{\mu},c_{1,\mu},\dots,c_{m,\mu})=\sum_{\mu\in U}x^{\mu}h_{\mu}\\
&=\sum_{\mu\in U}x^{\mu}(d_{1,\mu}f_1+\dots+d_{n,\mu}f_n)=\kappa_1f_1+\dots\kappa_n f_n
\end{align*}
with $\kappa_i=\sum_{\mu\in U}x^{\mu}d_{i,\mu}\in C[x_1,\dots,x_r]\subseteq C(S)$ for $1\leq i\leq n$. Consequently $(g,c_1,\dots,c_m)$ is an element of the right-hand side in~\eqref{Equ:SolSpaceId} which proves~\eqref{Equ:SolSpaceId}. Now suppose that $f_1,\dots,f_n$ are linearly dependent over $C(S)$. Then one can take $\kappa_1,\dots,\kappa_n\in C[S]$, not all zero, with $\kappa_1\,f_1+\dots+\kappa_n\,f_n=0$. 
In addition, we can take a generator set $S_0=\{x_1,\dots,x_r\}\subseteq S$ with minimal polynomials as given above such that $\kappa_1,\dots,\kappa_n\in k[S_0]$.
Take any $\mu\in[d_1]\times[d_2]\times\dots\times[d_r]$ where $x^{\mu}$ arises in one of the $\kappa_1,\dots,\kappa_n$. By coefficient comparison w.r.t.\ $x^{\mu}$ and using the fact that the elements in~\eqref{Equ:IndependentSet} are linearly independent over $k$,
we get $c_1\,f_1+\dots+c_n\,f_n=0$ for some $c_1,\dots,c_n\in C$, not all zero. This contradicts the assumption that $f_1,\dots,f_n$ are linearly independent over $C$, and we conclude that $f_1,\dots,f_n$ form a basis of ${\mathcal V}(a,b,k(S))$. 
\end{proof}

Based on the above considerations, our algorithmic toolbox can be applied with two complementary scenarios. 

\vspace*{-0.3cm}

\begin{enumerate}
\item One can choose a \pisiSE-field $F=C(t_1)\dots(t_e)$ over the smallest possible constant field $C=A(y_1,\dots,y_{\rho})$ in which the linear difference equation can be formulated and computes all hypergeometric solutions in this difference field $(F,\sigma)$. This feature is activated in \texttt{Sigma} with the option \MText{WithAlgebraicNumbers$\to$False}.
\item One finds all hypergeometric solutions in the \pisiSE-field $F'=C'(t_1)\dots(t_e)$ with $C'=\bar{\set Q}(y_1,\dots,y_{\rho})$. More precisely, one takes the smallest (finite) field extension $\tilde{A}$ of $A$ such that the sets of hypergeometric candidates in $F'$ and in $\tilde{F}=\tilde{C}(t_1)\dots(t_{e})$ with $\tilde{C}=\tilde{A}(y_1)\dots(y_{\rho})$ are the same. Given this set of hypergeometric candidates one constructs all hypergeometric solutions in $\tilde{F}$ by solving for each candidate a linear difference equation in $\tilde{F}$. For the purpose of efficiency, one does not have to choose a common constant field $\tilde{A}$ in which all the arising difference equations are expressible, but for each linear difference equation one can choose the smallest constant field in which it is expressible. Namely, by Proposition~\ref{Prop:AlgebraicPLDE} the basis obtained in this optimal constant field will provide also a basis in the \pisiSE-field $(\tilde{F},\sigma)$ over $\tilde{C}$ and even more in the \pisiSE-field $(F',\sigma)$ over $C'$. As a consequence, one obtains all hypergeometric solutions in $(F',\sigma)$. With the option \MText{WithAlgebraicNumbers$\to$True} this feature is activated. Unfortunately, the available number field operations of Mathematica are not optimally implemented.
As a consequence, the routines of \texttt{Sigma}, that rely on these Mathematica operations, might not be executed adequately.
\end{enumerate}

In various applications one is interested in hypergeometric solutions in a \pisiSE-field
$B(y_1)\dots(y_{\rho})(t_1)\dots(t_e)$ where the algebraic number field $B$ is in between the two extreme cases (1) or (2) from above. Within the summation package \texttt{Sigma} the construction of $B$ can be controlled when one determines the hypergeometric candidates in $k_0=A(y_1,\dots,y_{\lambda})$. For instance, with the option \MText{WithAlgebraicNumbers$\to$size} one can specify that the root solutions do not require more memory than the byte value \texttt{size}. Similarly, one can control with the option  
\MText{DegreeInProducts$\to$d} the maximal degree \texttt{d} of the irreducible factors within the hypergeometric candidates. For instance with the option \MText{DegreeInProducts$\to$2} one excludes objects like $\prod_{i=1}^m(1 + i + i^2 + 2 i^3)$ (provided that the underlying algebraic number field is, e.g., $\set Q$).
Furthermore, the option \MText{IntegerSizeInProducts$\to$s} restricts to candidates where the irreducible factors have integer coefficients whose absolute values are not larger than \texttt{s}. In this way one can avoid for instance factorials like $(100 m)!$ if one sets the option \MText{IntegerSizeInProducts$\to$99}.

\section{Conclusion}\label{Sec.Conclusion}

Inspired and guided by the differential case~\cite{Singer:91} we elaborated a general framework to compute 
\begin{itemize}
	\item hypergeometric solutions of homogeneous linear difference equations, and
	\item rational solutions of parameterized linear difference equations 
\end{itemize}
in a difference field $k=K(t_1)\dots(t_e)$ built by a tower of \pisiSE-monomials over a ground difference field $K$ that is $\sigma$-computable (see Definition~\ref{Def:computable}). As a consequence, we obtain a complete algorithm of the above problems if $k$ is a \pisiSE-field where the constant field $K$ fulfills (algorithmic) properties given in Theorem~\ref{Thm:GroundField}. This is in particular the case if $K$ is a rational function field defined over an algebraic number field. These algorithms (or some mild variations) for such \pisiSE-fields are implemented within the summation package \texttt{Sigma}.

We emphasize that this framework gives room for further extensions. Whenever a difference field $(K,\sigma)$ can be equipped with the algorithmic properties in Definition~\ref{Def:computable}, one can activate our machinery to solve the above problems also in \pisiSE-monomials defined over $K$. This might be possible, e.g., for the free difference field that represents generic sequences~\cite{KS:06} or radical difference field extensions~\cite{KS:07} that model objects like $\sqrt{\nu}$.

We remark further that these new algorithms are not only interesting for recurrence solving but also for recurrence finding. For instance, finding solutions of parameterized linear recurrences is the backbone of the holonomic summation toolbox~\cite{Zeilberger:90a,Chyzak:00,Koutschan:13} to tackle sums whose summands are described by systems of linear difference equations. In the latter approaches only the special cases $k=C(x)$ or the $q$-rational are considered so far. However, in ~\cite{Schneider:05c} and recent refinements~\cite{DRHolonomic} this toolbox has been considered also in the \pisiSE-field setting. In the light of the new complete machinery for \pisiSE-fields, it will be interesting to see how these achievements can be applied to new classes of summation problems.

\section*{Acknowledgements}
S.A.A., M.P., and C.S.\ express their gratitude to Manuel Bronstein's widow Karola Bronstein for her support and encouragement to publish this article. In addition, we are grateful to the three referees for their very thorough reading and highly valuable suggestions which helped us improve the presentation.

Funding: This work was supported by the Russian Foundation for Basic Research (project No.\   19 01 00032), the Slovenian Research Agency (research core funding No.\ P1-0294), and the Austrian Science Foundation (FWF grant F5009-N15 in the framework of the Special Research Program ``Algorithmic and Enumerative Combinatorics''). 



\newcommand{\etalchar}[1]{$^{#1}$}
\providecommand{\bysame}{\leavevmode\hbox to3em{\hrulefill}\thinspace}
\providecommand{\MR}{\relax\ifhmode\unskip\space\fi MR }
\providecommand{\MRhref}[2]{%
	\href{http://www.ams.org/mathscinet-getitem?mr=#1}{#2}
}
\providecommand{\href}[2]{#2}

\end{document}